\documentclass[11pt]{article}
\usepackage{geometry, multirow, amsmath, amsfonts, amsthm, amssymb, graphicx, enumerate, natbib, bibentry, float, caption, subcaption, longtable,xcolor,bm}
\usepackage{tikz}
\usepackage[breaklinks=true]{hyperref} 
\usepackage[nameinlink]{cleveref} 
\usepackage{xcolor} 
\hypersetup{
	bookmarksnumbered=true,
	colorlinks=true,
	anchorcolor={red!50!black},
	linkcolor={red!50!black},
	citecolor={blue!50!black},
	urlcolor={blue!50!black}
}
\usepackage{txfonts}
\usepackage[title]{appendix}
\usepackage{setspace}

\usepackage{array,varwidth}

\newenvironment{keywords}{
       \list{}{\advance\topsep by0.35cm\relax\small
       \leftmargin=1cm
       \labelwidth=0.35cm
       \listparindent=0.35cm
       \itemindent\listparindent
       \rightmargin\leftmargin}\item[\hskip\labelsep
                                     \bfseries Keywords:]}
     {\endlist}

\newtheorem{Prop}{Proposition}
\newtheorem{Lem}{Lemma}
\newtheorem{Thm}{Theorem}
\newtheorem{Cor}{Corollary}
\newtheorem{Def}{Definition}

\let\epsilon\varepsilon

\DeclareMathOperator*{\argmax}{arg\max}

\newcommand\abs[1]{\left|#1 \right|}
\newcommand\met[1]{\left\|#1 \right\|}

\newcommand\Brace[1]{\left\{\begin{aligned}#1\end{aligned}\right.}

\newcommand\kh[1]{\left(#1\right)}
\newcommand\fkh[1]{\left[#1\right]}
\newcommand\hkh[1]{\left\{#1\right\}}

\newcommand\eqn[1]{\begin{align}#1\end{align}}
\newcommand\eqns[1]{\begin{align*}#1\end{align*}}

\begin{document}
\setstretch{1.5}
\begin{titlepage}
\title{\textbf{Robust Contracts with Exploration}}
\author{Chang Liu\thanks{\protect\setstretch{1.2}School of Economics, UNSW Business School; \href{mailto:chang.liu36@unsw.edu.au}{\nolinkurl{chang.liu36@unsw.edu.au}}.
I gratefully acknowledge funding from the National Science Foundation under grant DMS-1928930 and from the Alfred P. Sloan Foundation under grant G-2021-16778 during the Fall 2023 semester. I express my deep gratitude to my advisors Shengwu Li, Tomasz Strzalecki, Eric Maskin, and Benjamin Golub for their guidance and support throughout the project, and over the years. An earlier version of this paper appeared as the first chapter of my Ph.D. dissertation at Harvard. I would also like to thank Sejal Aggarwal, Ophir Averbuch, Benjamin Brooks, Shani Cohen, Juan Dodyk, Piotr Dworczak, Federico Echenique, Yannai Gonczarowski, Olivier Gossner, Simon Grant, Jerry Green, Yingni Guo, B{\aa}rd Harstad, Oliver Hart, Wei He, Zo{\"e} Hitzig, Ashwin Kambhampati, Michihiro Kandori, Rumen Kostadinov, Jacob Leshno, Jiangtao Li, Jonathan Libgober, John Macke, Erik Madsen, Stephen Morris, Maxwell Rosenthal, Giorgio Saponaro, Brit Sharoni, Cassidy Shubatt, Kathryn Spier, Juuso Toikka, Haoqi Tong, Alexander Wolitzky, Wanchang Zhang, and especially Paul Milgrom for valuable comments and discussion. In addition, I thank the audiences at UNSW Sydney, Tsinghua University, Simons Laufer Mathematical Sciences Institute, Monash University, Shanghai University of Finance and Economics, Warsaw School of Economics, Shandong University, HKUST, Peking University, University of Melbourne, CUHK-Shenzhen, University of Queensland, Asian Meeting of the Econometric Society, Stony Brook International Conference on Game Theory, North American Winter Meeting of the Econometric Society, Conference on Mechanism and Institution Design, World Congress of the Game Theory Society (GAMES 2024), Deakin Economic Theory Workshop, and Annual SAET Conference for their insightful feedback.}}
\date{\begin{tabular}{ rl } 
First version:&November 30, 2022\\
This version:&November 8, 2026
\end{tabular}}
\maketitle
\thispagestyle{empty}
\begin{abstract}
We study a moral hazard problem over an arbitrary finite horizon, with multiple agents whose action sets are unknown to the principal. The principal contracts with the agents sequentially and seeks to maximize the worst-case discounted sum of payoffs, where the worst case is over the possible action sets. The principal observes the action chosen by each agent, and then offers a new contract to the next agent based on the knowledge accumulated so far, thus having the opportunity to explore in every period before the final one. We introduce and compare three different notions of dynamic worst-case evaluation. Within each notion, we define a suitable rule of updating and characterize the principal's optimal payoff guarantee. We find that linear contracts are robustly optimal not only in static settings, but also in dynamic environments with exploration.
\end{abstract}
\end{titlepage}

\setstretch{1.5}

\newpage\setcounter{page}{2}

\section{Introduction}
Moral hazard models, in which a principal designs a contract to incentivize an agent, have been extensively studied and widely applied. In many canonical moral hazard models, however, optimal contracts require precise knowledge of the environment: the set of all possible actions together with the (stochastic) mappings from actions to outcomes. This aspect raises practical concerns, because in reality the principal's knowledge is rarely exact. How should the principal design contracts that have robust guarantees even if some details are  incorrect? The emerging area of  robust contract design follows the \emph{Wilson Doctrine} \citep{WilsonDoc}, which advocates for realistic approaches that are detail-free.  

The pioneering work by \cite{Carroll15} assumes that the principal knows only some of the actions available to the agent, and evaluates contracts based on their worst-case performance, over the unknown actions the agent might take. The results show that, very generally, the optimal contract is linear, which provides new foundations for the common use of  linear contracts in practice.

One suspicion, however, about the linear results in \cite{Carroll15} is how much they hinge on the principal's inability to explore the unknown, an opportunity that  arises naturally in models with multiple interactions.\footnote{One related but distinct criticism of the robust mechanism design literature is that most models are static in construction but assume commitment. We discuss this issue in the literature review below. See also \cite{LM22} for a corresponding  perspective in the area of informationally robust mechanism design.} It is not even clear how to model (non-Bayesian) exploration in the robust paradigm. Specifically, if the principal can observe an agent's chosen action, then she can gain insights into actions that were initially unknown but might be subsequently undertaken. Furthermore, based on the agent's rationality, she may also exclude certain actions that were not chosen. In such environments, are linear contracts still robustly optimal? 

A natural class of applications for robust contract design involves a principal hiring or consulting specialists whose expertise surpasses her own. This explains her limited knowledge of the agents' available actions and her lack of a prior over the unknown ones. For instance, consider an individual hiring gig workers from online platforms. While long-term contracts are typically not enforceable, she does have the opportunity to interact with a pool of workers. Given that the workers share similar professional training, the individual's knowledge about the capability of the pool from past experience is valuable for improving future interactions. Within this example, the main theoretical question of this paper is twofold: First, how should the individual structure contracts to best respond to new knowledge gained from exploration? Second, in anticipation of such opportunities, what contracts are optimal for acquiring new knowledge?

We study a $T$-period moral hazard problem, for any finite $T$ (Section~\ref{sec:model}). There are $T$ agents whose action sets are unknown to the principal. The principal contracts with the agents sequentially to provide incentives. She observes the action chosen by each agent and conditions subsequent contracts on the knowledge accumulated from these observations, thus having the opportunity to explore before the final period. The principal and agents are all risk neutral, and payments are constrained by limited liability.

The principal knows only some of the actions available to the agents. Other actions may also exist, and she has no well-defined prior belief about them. Faced with this nonquantifiable uncertainty, the principal seeks to maximize her worst-case discounted sum of payoffs, where the worst case is over the possible action sets. It is therefore crucial to specify which actions the principal considers possible after each history, that is, how her knowledge about unknown actions is updated across periods.

The main result of this paper is that linear contracts are robustly optimal not just in static settings, but also in dynamic environments with exploration. To obtain this conclusion, we introduce and compare three distinct notions of dynamic worst-case evaluation: \emph{independent technology}, \emph{advancing technology}, and \emph{constant technology}. In the first period, the principal considers possible any action set containing the known actions; the three notions differ in how this family is updated after she observes an agent's response to a contract. Equivalently, one can imagine an adversarial ``nature'' that selects each agent's action set to minimize the principal's payoff, and the notions differ in the restrictions they impose on nature's moves across periods. Within each notion, we define the rule of updating and characterize the principal's optimal payoff guarantee.

We begin by considering the case of \emph{independent technology}, where the action sets of different agents are not related; in other words, nature can select the action set for each agent independently. In this case, the choices made by earlier agents do not provide the principal with information about what actions later agents can take. Therefore, the learning aspect is essentially nullified, and the principal's overall payoff guarantee is maximized by adopting a straightforward approach: offering the optimal static contract identified by \cite{Carroll15} in every period. This case serves as a no-learning benchmark for the other two notions.

Next, we analyze the first restriction that facilitates meaningful exploration: the case of \emph{advancing technology} (Section~\ref{sec:adv}). In this case, the action set may expand between periods, but cannot shrink. In other words, nature can only introduce new actions across periods, but is not allowed to delete old ones. The main result for the case of advancing technology is that linear contracts are robustly optimal in every period and \emph{after every feasible history} (Theorem~\ref{prop:1grow}). Toward this conclusion, we solve the principal's dynamic problem via backward induction. After the principal offers some contract and observes the action chosen by the current agent, she learns that this action exists and may be taken again by every subsequent agent. Since nature may introduce new actions in any period, earlier agents' rationality excludes no action from later action sets. Therefore, the principal's final-period problem simplifies to a single-period problem in \cite{Carroll15} with respect to the updated knowledge of the set of actions, and thus optimal final-period contracts are linear.

Before the final period, the principal's contract also affects which action she observes and hence what she knows in all later periods. Extending linear optimality to an arbitrary finite horizon therefore requires more than solving the final-period problem. By backward induction, we show that the continuation value is weakly increasing as the set of known actions expands, depends on these actions only through their expected outputs and costs, and is upper semicontinuous. Given these properties, the argument at each date boils down to two steps. The first step shows that any nonlinear contract can be improved into a linear contract, thereby (weakly) increasing the overall payoff guarantee (Lemma~\ref{lem:affine}). The second step further shows that the maximum of the principal's current problem exists within the class of linear contracts (Lemma~\ref{lem:optlinear}). The value at that date inherits the same properties, so the argument repeats at every earlier date. Thus, even with repeated opportunities to explore, no contract outperforms the best linear contract. This holds after every feasible history, including histories generated by nonlinear contracts.

Moving on to an alternative notion with more restrictions, the case of \textit{constant technology}, we assume all agents share the same set of actions unknown to the principal (Section~\ref{sec:const}). In other words, nature can neither introduce new actions across periods nor delete old ones. The main result for the case of constant technology is Theorem~\ref{prop:1}, which shows that linear contracts are robustly optimal in every period, although \textit{not after every feasible history}. Specifically, the final-period analysis shows that, following \textit{nonlinear} contracts in earlier periods, optimal final-period contracts may also be \textit{nonlinear} in some cases. Nonetheless, upon backward induction, an optimal contract after every feasible history can be found among convex combinations of a linear contract and the previous contracts. In the first period, there are no previous contracts, so it is robustly optimal to use a linear first-period contract. After every history of linear contracts, these combinations are again linear, so linear contracts remain optimal in all subsequent periods \textit{on the path}.

This difference from the case of advancing technology stems from a more subtle rule of updating. After observing the actions chosen in earlier periods, the principal considers possible any action set that (i) contains the initially known actions and the observed actions, and (ii) contains no action that would have given an earlier agent a strictly higher payoff than the action he chose, under the contract offered to him. We call such action sets \emph{compatible} with the history (Definition~\ref{def:comp}). Requirement (i) records that the chosen actions exist, while requirement (ii) captures the additional inference the principal can draw from the rationality of the earlier agents.

Compatibility already matters in the final period, where the backward induction begins. Unlike under advancing technology, the final-period problem is not a direct adaptation of the single-period problem in \cite{Carroll15}, because requirement (ii) excludes certain actions. We fully characterize the principal's \textit{optimal final-period payoff guarantee} and identify the contracts that attain it (Lemma~\ref{lem:second}). The optimal guarantee is achieved by the best of three kinds of responses to the knowledge accumulated in earlier periods: (i) repeating a previous contract, (ii) modifying a previous contract by adding a share of the residual output to compensate the final agent, and (iii) offering the linear contract that is optimal in the static problem of \cite{Carroll15} for the updated set of known actions. Hence, when a response of kind (i) or (ii) built from a nonlinear previous contract is optimal, the optimal guarantee is achieved by a nonlinear contract.

As before, the final-period characterization is only the starting point. In earlier periods, the current contract determines not only the principal's current payoff but also what she learns, and the restrictions revealed by the agent's choice persist in every later period. The continuation value therefore depends on all previous contracts, and the final-period formula cannot simply be iterated. Instead, we develop a recursive argument that relies only on properties of the continuation value preserved by backward induction, notably convexity and monotonicity. The argument at each date now has three steps. The first step shows that the principal's guarantee depends on a contract only through its \emph{upper payment envelope} and its payments to the known actions (Lemma~\ref{lem:envelopeconst}). The second step uses a supporting-hyperplane argument to show that any contract can be improved into a convex combination of a linear contract and the previous contracts, thereby (weakly) increasing the overall payoff guarantee (Lemma~\ref{lem:improvementconst}). The third step further shows that the maximum of the principal's current problem exists within this class of contracts (Lemma~\ref{lem:existenceconst}). Together, these steps prove Theorem~\ref{prop:1}. This recursive argument is the main technical contribution of the constant-technology analysis.

Finally, we examine the structure of the optimal linear first-period contract when $T=2$ and the principal knows a single action (Appendix~\ref{subsec:opt}), and compare it with the optimal static contract identified by \cite{Carroll15}. Under advancing technology, the principal raises the first-period share above the optimal static share to explore, since a higher share makes it costlier for nature to conceal valuable actions; this exploration effect is largest when the two periods receive equal weight and disappears as she becomes very patient, because the final period then dominates and the worst-case discovery becomes nearly worthless. Under constant technology, the effect is reversed and persists however patient she is: because the first agent's choice also rules out actions he would have preferred, the principal can profit in the final period by reusing the first-period contract, possibly with compensation for the next agent, and this option is worth more when she retains a larger share of output in the first period.

\paragraph{Related Literature}\label{subsec:lit}
Foundations for linear incentive contracts have received extensive research attention. The seminal work of  \cite{HM87} considers a dynamic framework where output is produced gradually over time, the agent is aware of his own progress, and the principal pays the agent at the end. Although the principal is allowed to use the entire history of output to determine the payment, the optimal contract depends only on the number of realizations of each output level, and is linear in these counts.  In a continuous-time version of their problem where the agent controls the drift of a multidimensional Brownian motion,  the optimal contract can be expressed as a linear function that depends only on the endpoint.\footnote{Following \cite{HM87}, \cite{Sung95} further shows that the optimal contract can still be linear when the agent controls the variance; \cite{HS02} provide discrete-time approximations of the continuous-time model.} However, the stationary structure of their model is critical for this linearity result,\footnote{For example, \cite{SS93} show that a time-dependent technology makes the optimal contract nonlinear.} because linear contracts provide the agent with constant incentives to move forward independent of his past performance.  In our model, the principal offers multiple contracts during the process,  and exploration makes the principal's problem inherently non-stationary. Therefore, our paper considers a different form of foundation for linear contracts. Furthermore, \cite{Diamond98} and \cite{BGS18} provide arguments for linear contracts using static Bayesian frameworks. 

More recently, pioneered by \cite{Carroll15}, this issue has been investigated by a wave of research using robust models of contract design, which requires contract performance to be robust to limited knowledge of the environment. \cite{Carroll19}  provides a comprehensive review of this approach, as well as an overview of the evolving field of robust mechanism design that adopts many other notions of robustness. Most work in robust contract design, however, analyzes static or one-shot models, which precludes the opportunity for designers to better understand parts of the environment they do not know. While starting with nonquantifiable uncertainty, designers may still be able to gradually gain a better understanding of the environment in which they repeatedly engage through exploration. 
Our dynamic model provides the principal with the opportunity to explore the unknown,  in order to understand how the principal should design contracts that are robustly optimal given this exploration opportunity.

As stated by \cite{Carroll19}, ``another challenge is that trying to write dynamic models with non-Bayesian decision makers leads to well-known problems of dynamic inconsistency, except in special cases (e.g., \cite{ES03}). This may be one reason why there has been relatively little work to date on robust mechanism design in dynamic settings.'' A broad literature in decision theory has developed approaches to evaluating ambiguity in dynamic settings, including \cite{ES03} and, more recently, \cite{Grant26}. In the spirit of this literature, we formulate the principal's worst-case evaluation recursively: after each history, she takes the worst case over the updated family of possible action sets and anticipates applying the same criterion in later periods. This recursive formulation ensures consistency between her current and continuation evaluations.

This paper is relevant to the recent research that examines robust contracting in different organizational environments. Specifically, \cite{DTFC} analyze moral hazard in teams,  \cite{MO19} study a common agency model, and \cite{CB22} investigate a model with double moral hazard. \cite{WCFC} provide a general framework that goes beyond simple bilateral relationships and allows for rich internal organizational structures. Our model analyzes a simple contracting environment, and aims to capture the main issue in terms of exploration. In particular, due to exploration, the analysis of our dynamic model cannot be directly derived using the conclusions in \cite{WCFC}.\footnote{We articulate the specific differences between our dynamic model and the general static framework in \cite{WCFC} in Section \ref{sec:period2}.}

The revealed preference reasoning in this paper is related to the recent work by \cite{BR23} and \cite{AG23}, who consider a static robust contracting problem with revealed preference data. In \cite{BR23} and \cite{AG23}, the principal's only knowledge is the agent's best responses to a finite number of given contracts, and she seeks to maximize her worst-case payoffs over all action sets that can rationalize the data. Under constant technology, where all agents share the same action set, the principal's additional knowledge in the final period of our model consists exactly of the earlier agents' best responses to the contracts they were offered. Therefore, our final-period characterization (Lemma~\ref{lem:second}) contains a compensation component similar to theirs. However, the principal in our model also knows some of the available actions from the outset, so the optimal contracts are not identical in structure.\footnote{Another reason the results are similar but not identical is the assumption about what is observed:  \cite{BR23} and \cite{AG23} assume that the output distribution, but not the effort cost, associated with the best response is observed. We instead assume that both the distribution and the cost are observed, which is more consistent with the principal's initial knowledge in our model.} More importantly, in their settings, the principal's revealed preference data are exogenously provided, whereas our model places a significant emphasis on endogenizing this aspect through the optimal exploration design before the final period.

From a broader perspective,  \cite{MO19} and \cite{CB22} are in a similar spirit to our work on how the designers' robust objectives interact with their policy choices. In \cite{MO19}, several principals compete to contract with a common agent. In \cite{CB22}, the principal supplies her own input in addition to contracting with the agent. However, both studies apply the maxmin objective only once, whereas in our model the principal applies it in every period. In informationally robust mechanism design, \cite{LM22} study durable good monopoly without commitment and introduce the notion of a \emph{sequentially worst-case} information structure. In our model, by contrast, the uncertainty concerns the agents' feasible actions, and the principal's contracts endogenously generate observations that restrict the action sets she considers possible in subsequent periods.

A number of other recent papers considering static models of robust contracts are related to our work, because  the principal is aware of some additional characteristics of the unknown actions in addition to the concern that they may exist. Like \cite{Kambhampati22}, who studies robust performance evaluation of independent agents, we leave the action set of an individual agent unrestricted beyond containing the known actions and impose restrictions only across agents. Under constant technology, all agents have the same action set, as in his analysis of symmetric, binary environments. However, the purposes differ: his independence assumption rules out the known rationales for joint performance evaluation, whereas our restrictions make the principal's observations of chosen actions valuable. In addition, \cite{Antic21} assumes that the principal knows a lower bound on the productivity of every unknown action. Furthermore, in \cite{DRT20}, the principal only knows the first moment of the distribution over output induced by each possible action, but not the full distribution.
\\ \\ 
\indent The rest of the paper is organized as follows. Section~\ref{sec:model} lays out the model, and analyzes the case of independent technology. The first main part, Section~\ref{sec:adv}, analyzes the case of advancing technology, and shows that linear contracts are robustly optimal \textit{after every feasible history}. The second main part, Section~\ref{sec:const}, analyzes the case of constant technology and shows that linear contracts are robustly optimal \textit{on the path}: a linear contract is optimal in the first period and after every history of linear contracts. After histories with nonlinear contracts, however, optimal contracts may be nonlinear. Section~\ref{sec:concln} concludes. Appendix~\ref{app:proof} contains the proofs of the results in the main text. Appendix~\ref{subsec:opt} solves for the optimal first-period share in the two-period specification with a single known action, under both advancing and constant technology, and compares it with the optimal static share.

\section{Model}\label{sec:model}
\subsection{Notation}
We denote by $\Delta\kh{ {X}}$ the set of (Borel) probability measures on a set $ {X}\subseteq \mathbb{R}$, equipped with the weak topology. For $x\in {X}$, we write $\delta_x$ for the degenerate distribution that puts probability one on $x$. 

\subsection{Setup}
We study a $T$-period moral hazard problem with one principal (she) and $T$ agents (he), where $T\ge2$ is finite. The principal contracts with each agent sequentially to provide incentives. Each agent cares only about his own period's payoff, and his reservation payoff is zero. All parties are risk neutral, and the principal's discount factor is $\beta\in\kh{0,\infty}$.

In each period $t=1,\ldots,T$, agent $t$ takes a costly action that generates a stochastic output. The realized output $y$ belongs to a compact set $Y\subseteq\mathbb{R}_+$, which may be finite or infinite. We normalize the lowest possible output to zero.

An \emph{action} is a pair $a=\kh{F,c}\in\Delta\kh{Y}\times\mathbb{R}_+$, with the interpretation that if an agent chooses $a$, he incurs cost $c$ and output is drawn according to $F$. We equip $\Delta\kh{Y}\times\mathbb{R}_+$ with the product topology.

A \emph{technology} is a nonempty compact set of actions. Agent $t$ has technology $A_t$, which he knows but the principal does not. The principal initially knows a compact set $A_0$ of actions available to every agent, so $A_0\subseteq A_t$ for every $t$. She does not know what additional actions may be available and has no prior over the possible technologies. To ensure that the principal may benefit from contracting with the agents, assume that $\mathbb{E}_{F}\fkh{y}-c>0$ for some $\kh{F,c}\in A_0$.\footnote{Note that it is necessary for the principal to know at least one action that generates a strictly positive surplus, because otherwise it is always possible that every available action costs at least as much as its expected output, so that no contract guarantees her a positive payoff.}

Each agent can quit, in which case both parties receive zero payoff in that period. We represent this option by the null action $\kh{\delta_0,0}$, which produces zero output at zero cost, and include it in $A_0$ and hence in every technology.

Output is observable and contractible. To capture the idea of exploration, we assume that the principal also observes the action chosen by each agent, including its output distribution and cost, and can use this observation when contracting with later agents. The action itself, however, is not contractible.\footnote{It is a strong assumption that the chosen action becomes observable to the principal, especially since $F$ represents a distribution. One interpretation is that each period summarizes (the ``average'' state of) a horizon for which the contract needs to remain fixed, while the agent repeatedly takes the chosen action. During this process, the principal can keep observing him and figure out which action he is taking, in particular what $F$ and $c$ are. However, knowing that the action exists is still not the same as being able to write it into a contract. The action itself may be too complex to be accurately described in contract terms, or its inclusion in the contract may be directly prohibited by law.}

Agents have limited liability, so payments cannot be negative. A \emph{contract} is a continuous\footnote{The continuity assumption is made only to ensure the existence of best responses of the agents under every compact technology, and becomes vacuous if $Y$ is a finite set; see also \citet[footnote 1]{Carroll15}, \citet[footnote 3]{WCFC}, and \citet[footnote 1]{CB22}.} function $w:Y\to\mathbb{R}_+$ satisfying $w\kh{0}=0$.\footnote{One foundation for $w\kh{0}=0$ is \emph{two-sided limited liability} \citep[Definition 6]{BR23}, which also caps payments at output, a restriction we do not impose. Another is \emph{free disposal} of output: since the agent would discard output whenever $w\kh{y}<w\kh{0}$, the principal may as well pay at least $w\kh{0}$ everywhere, and then shifting all payments down by $w\kh{0}$ leaves incentives unchanged and saves her money.}

The timing within each period $t$ is as follows:
\begin{enumerate}
\item The principal offers a contract $w_t$.
\item Agent $t$ chooses an action $a_t=\kh{F_t,c_t}\in A_t$, which is observed by the principal.
\item Output $y_t\sim F_t$ is realized.
\item Payoffs are $y_t-w_t\kh{y_t}$ to the principal and $w_t\kh{y_t}-c_t$ to agent $t$.
\end{enumerate}

\subsection{Histories and Updating}\label{subsec:model-histories}
We restrict attention to deterministic policies under which the current contract depends only on the initial set of known actions and the contracts and actions observed in earlier periods. Formally, write $h^1=A_0$ and, for $t=2,\ldots,T+1$, $h^t=\kh{A_0,w_1,a_1,\ldots,w_{t-1},a_{t-1}}$, so that $h^{t+1}=\kh{h^t,w_t,a_t}$ is obtained by appending the current contract and observed action to $h^t$. A \emph{policy} specifies a contract after each nonterminal history. Realized outputs determine current payments but do not enter the conditioning history for subsequent contracts.\footnote{Conditional on the observed action, the realized output provides no further information about the technology. A policy could nevertheless condition subsequent contracts on it, using the output as a randomization device. Randomizing over linear contracts can strictly raise the principal's guarantee \citep{Kambhampati23,KPTTV26}; such policies fall outside the policy class studied here.}

The principal's objective is to maximize her worst-case expected discounted sum of payoffs. The worst case is evaluated recursively: in each period, it is taken over the technologies that remain possible given her observations so far. Therefore, it is crucial to articulate what actions the principal considers possible in each period and to determine how her knowledge about unknown actions are updated across periods. The relation between the technologies of different agents determines what the principal can learn from observing an action. We consider three notions of this relation:
\begin{enumerate}
\item \emph{Independent technology.} The technologies in different periods need not be related beyond each containing $A_0$. An action observed in an earlier period need not remain available.
\item \emph{Advancing technology.} Technologies satisfy $A_0\subseteq A_1\subseteq\cdots\subseteq A_T$. Actions available in one period remain available in later periods, but new actions may be introduced.
\item \emph{Constant technology.} All agents share the same technology: $A_1=\cdots=A_T=A$. Actions can be neither introduced nor removed across periods.
\end{enumerate}

Under advancing or constant technology, every previously observed action remains available. In these two cases, after history $h^t$ the principal knows that every action in
\eqn{K\kh{h^t}\equiv A_0\cup\hkh{a_1,\ldots,a_{t-1}}\label{eqn:model-known-actions}}
is available in period $t$, with $K\kh{h^1}=A_0$ and $K\kh{h^{t+1}}=K\kh{h^t}\cup\hkh{a_t}$.

A history is \emph{feasible} if, for some sequence of technologies satisfying the restrictions of the relevant notion, each observed action is a best response to the contract offered in its period. For a given notion and a feasible history $h^t$, let $\Omega\kh{h^t}$ denote the family of technologies that the principal considers possible in period $t$. Initially, $\Omega\kh{h^1}$ consists of all technologies containing $A_0$. An updating rule specifies how this family changes after a new contract and observed action.

To better understand the connections and distinctions among the three notions, it is helpful to imagine an adversarial ``nature'' that, in each period $t$, selects from $\Omega\kh{h^t}$ the technology against which the principal's current contract is evaluated, so as to minimize her payoff. The three notions differ in the restrictions they impose on nature's moves across periods, that is, in their updating rules. In the following sections, we define a suitable rule of updating within each notion and characterize the principal's optimal payoff guarantee.

\subsection{Recursive Payoff Guarantees}\label{subsec:model-guarantees}
Given a contract $w$ and a technology $A$, the agent's best-response correspondence is
\eqn{BR\kh{w\mid A}\equiv\argmax_{\kh{F,c}\in A}\hkh{\mathbb{E}_{F}\fkh{w\kh{y}}-c}.\label{eqn:model-br}}
Since $A$ is compact and $w$ is continuous, this set is nonempty and compact. Because the null action belongs to every technology, the agent's expected payoff at any best response is nonnegative.

Fix one of the three notions and its updating rule, and set $V_{T+1}^*\kh{h^{T+1}}=0$. If the principal offers contract $w_t$ after history $h^t$ and observes action $a_t=\kh{F_t,c_t}$, her \emph{interim payoff guarantee}, defined as her current expected payoff plus the discounted optimal continuation guarantee, is
\eqn{U_t\kh{w_t\mid h^t,a_t}\equiv\mathbb{E}_{F_t}\fkh{y-w_t\kh{y}}+\beta V_{t+1}^*\kh{h^t,w_t,a_t}.\label{eqn:model-interim}}

Her \emph{overall payoff guarantee} from $w_t$, defined as the worst-case interim payoff guarantee over all technologies in $\Omega\kh{h^t}$, is
\eqn{U_t\kh{w_t\mid h^t}\equiv\inf_{A_t\in\Omega\kh{h^t}}\hkh{\max_{a_t\in BR\kh{w_t\mid A_t}}U_t\kh{w_t\mid h^t,a_t}},\label{eqn:model-overall}}
where ties among best responses are broken in the principal's favor according to the interim payoff guarantee \eqref{eqn:model-interim}, which includes the continuation value.\footnote{This tie-breaking assumption ensures the existence of optimal contracts and minimizes the departure from standard models of robust contracting. Other tie-breaking rules may introduce technical complications. For example, the principal's optimal payoff guarantee may be approached, but not attained, by linear contracts. See also \citet[Section D]{Carroll15}, \citet[footnote 4]{DTFC}, and \cite{CB22}.}

The principal's \emph{optimal payoff guarantee} is
\eqn{V_t^*\kh{h^t}\equiv\sup_{w_t}U_t\kh{w_t\mid h^t},\label{eqn:model-value}}
where the supremum is over all contracts. The results below establish that the maxima over best responses are attained and that optimal contracts exist under each of the three notions.

Equations~\eqref{eqn:model-interim}-\eqref{eqn:model-value} define the principal's evaluation recursively. At each history, she evaluates the current contract using the updated family of possible technologies and optimally chosen continuation contracts. The updating rule is suppressed in the notation. In Section~\ref{sec:const}, we use hats, $\hat U_t$ and $\hat V_t^*$, to distinguish the guarantees under constant technology.

\subsection{Independent Technology}\label{subsec:independent}
We begin by considering the case of independent technology. Since the technologies of different agents need not be related beyond each containing $A_0$, observations about earlier agents do not restrict the current agent's technology, and every technology containing $A_0$ remains possible. The principal therefore faces the same static robust contracting problem in every period, and her overall payoff guarantee is maximized by repeatedly offering the optimal static contract identified by \cite{Carroll15}.

We briefly recap the analysis in \cite{Carroll15}, as it lays the foundation for subsequent analyses. For a contract $w$ and a technology $A$, let
$$V\kh{w\mid A}\equiv\max_{\kh{F,c}\in BR\kh{w\mid A}}\mathbb{E}_{F}\fkh{y-w\kh{y}},\quad V\kh{w}\equiv\inf_{A\supseteq A_0}V\kh{w\mid A}.$$
The first is the principal's expected payoff under technology $A$, with ties broken in her favor, and the second is her static payoff guarantee.

Define
\eqn{\Phi_0\equiv\max_{\kh{F,c}\in A_0}\hkh{\sqrt{\mathbb{E}_{F}\fkh{y}}-\sqrt{c}}.\label{eqn:static-known-value}}
The assumption that some known action generates strictly positive surplus implies $\Phi_0>0$. The key result of \cite{Carroll15} is that the principal's optimal static payoff guarantee, $\max_w V\kh{w}$, is attained by a linear contract $w\kh{y}=sy$. Specifically, the solution to the principal's static problem can be summarized as follows:
\begin{enumerate}
\item Choose $a^*=\kh{F^*,c^*}$ attaining the maximum in equation \eqref{eqn:static-known-value}.
\item Set $s^*=\sqrt{c^*/\mathbb{E}_{F^*}\fkh{y}}$ as the share, and offer the linear contract $w^*\kh{y}=s^*y$.
\end{enumerate}
The resulting optimal guarantee is equal to $\Phi_0^2=\kh{\sqrt{\mathbb{E}_{F^*}\fkh{y}}-\sqrt{c^*}}^2$. Under independent technology, the continuation problem is unaffected by the current contract and observed action, so it is optimal to offer $w^*$ in every period, and the principal's optimal payoff guarantee is $V_1^*\kh{h^1}=\kh{\sum_{t=1}^{T}\beta^{t-1}}\Phi_0^2$.

\section{Advancing Technology}\label{sec:adv}
The case of independent technology might be overly pessimistic, as observing earlier agents' actions does not restrict the set of possible technologies in later periods. In this section, we analyze the first restriction on nature's moves across periods that facilitates meaningful exploration: the case of \emph{advancing technology}. Here, nature may introduce new actions across periods but cannot remove existing ones. Once the principal observes an action, she therefore knows that it remains available for future contracting. This persistence creates an opportunity for exploration: the current contract affects both the principal's current payoff and the actions she may learn about for subsequent periods.

The main result of this section, Theorem~\ref{prop:1grow}, establishes that linear contracts are robustly optimal in every period and \emph{after every feasible history}. Thus, linear contracts remain optimal when the principal accounts for the value of exploration in her current contracting decision.

\subsection{Rule of Updating and Final-Period Analysis}
Since a new action may be introduced in any period, earlier agents' choices cannot rule out the availability of an unobserved action. Accordingly, after every feasible history $h^t=\kh{A_0,w_1,a_1,\ldots,w_{t-1},a_{t-1}}$, $\Omega\kh{h^t}$ consists of all technologies containing $K_t\equiv K\kh{h^t}=A_0\cup\hkh{a_1,\ldots,a_{t-1}}$.

The updating rule implies that, at a given date, histories with the same known set give rise to the same continuation problem. In this section, we therefore also write $V_t^*\kh{K_t}$ for the optimal payoff guarantee when the current known set is $K_t$. In particular, if $a_t=\kh{F_t,c_t}$, then
$$U_t\kh{w_t\mid h^t,a_t}=\mathbb{E}_{F_t}\fkh{y-w_t\kh{y}}+\beta V_{t+1}^*\kh{K_t\cup\hkh{a_t}}.$$

We solve the principal's problem by backward induction. In the final period, the continuation value is zero, and the principal faces the static robust contracting problem with the updated known set. For a (nonempty) compact set $B$ of actions, define
\eqn{\Phi\kh{B}\equiv\max_{\kh{F,c}\in B}\hkh{\sqrt{\mathbb{E}_{F}\fkh{y}}-\sqrt{c}}.\label{eqn:Phi}}
In particular, $\Phi\kh{A_0}=\Phi_0$. Applying the static result of \cite{Carroll15} gives $V_T^*\kh{K_T}=\Phi\kh{K_T}^2$. An optimal final-period contract is linear and can be obtained by applying the construction in Section~\ref{subsec:independent} with $K_T$ in place of $A_0$.

\subsection{Optimal Contracts}\label{sec:advopt}
In earlier periods, the principal also takes into account how the current contract affects her continuation guarantee through the observed action. The following theorem establishes that an optimal contract remains linear.
\begin{Thm}\label{prop:1grow}
Under advancing technology, for every feasible history $h^t$ with $t=1,\ldots,T$,
\eqn{V_t^*\kh{h^t}=\max_{s_t\in\fkh{0,1}}U_t\kh{s_ty\mid h^t}.\label{eqn:adv-linear-optimal}} 
Consequently, there exists an optimal policy that offers a linear contract after every feasible history. 
\end{Thm} 
\begin{proof}
All proofs of the results in the main text are in Appendix \ref{app:proof}.
\end{proof}
The theorem also applies after histories generated by nonlinear contracts. The optimal share may depend on the date and on the expected output and cost pairs of the known actions, but not on the contracts under which those actions were discovered. Under constant technology, studied in Section~\ref{sec:const}, these contracts also matter, because the agents' responses to them restrict which technologies remain possible.

The proof proceeds by backward induction. Alongside linear optimality, it establishes that $V_t^*\kh{K_t}$ is upper semicontinuous in $K_t$, is weakly increasing under set inclusion, and depends on $K_t$ only through its set of expected output and cost pairs. These properties hold automatically for $V_{T+1}^*\equiv0$. Assuming them at date $t+1$, the argument at date $t$ has two main steps. First, Lemma~\ref{lem:affine} constructs a linear improvement of any current contract. Second, Lemma~\ref{lem:optlinear} establishes the existence of an optimal linear share.

\subsubsection{Proof Step 1: Improving Nonlinear Contracts}
Fix a period $t$, a feasible history $h^t$, and a contract $w_t$. Let $a_t^0=\kh{F_t^0,c_t^0}\in K_t$ be agent $t$'s best response when his technology is just the known set $K_t$. Let $\hat w_t$ denote the following linear contract:
\eqn{\hat w_t\kh{y}=s_ty\quad\text{with}\quad s_t=\Brace{&\frac{\mathbb{E}_{F_t^0}\fkh{w_t\kh{y}}}{\mathbb{E}_{F_t^0}\fkh{y}},&&\text{if }\mathbb{E}_{F_t^0}\fkh{y-w_t\kh{y}}>0,\\&0,&&\text{otherwise}.}\label{eqn:affine}}
Figure~\ref{fig:w1hat} illustrates the construction when the principal's current payoff from $a_t^0$ is positive.

\begin{figure}[!htbp]\centering
\begin{tikzpicture}[x=1.1cm,y=1.1cm,>=stealth,font=\small,declare function={w(\y)=0.34*\y^3-2.5*\y^2+5.2*\y;}]
\fill[gray!25] plot[domain=0:1.1765,samples=40] (\x,{w(\x)}) to (5,6) to cycle plot[domain=3.6765:5,samples=40] (\x,{w(\x)}) to (0,0) to cycle;
\draw[->] (0,0) to (5.6,0) node[right] {Realized output, $y$};
\draw[->] (0,0) to (0,6.6) node[above] {Payment to agent, $w\kh{y}$};
\node[below left] at (0,0) {$0$};
\draw[densely dashed,thin] (3,0) node[below] {$\mathbb{E}_{F_t^0}\fkh{y}$} to (3,2.6) to (0,2.6) node[left] {$\mathbb{E}_{F_t^0}\fkh{w_t\kh{y}}$};
\draw[very thick] plot[domain=0:5,samples=120] (\x,{w(\x)});
\draw[thick,dashed] (0,0) to (5,{5*2.6/3});
\fill (3,2.6) circle (1.6pt);
\draw (0.6,0) arc[start angle=0,end angle={atan(2.6/3)},radius=0.6];
\node at (17:0.8) {$s_t$};
\draw[very thick] (6.4,3.7) to (7.2,3.7) node[right] {$w_t\kh{y}$};
\draw[thick,dashed] (6.4,3.1) to (7.2,3.1) node[right] {$\hat w_t\kh{y}$};
\end{tikzpicture}
\caption{The linear contract $\hat w_t$ constructed from $w_t$.}\label{fig:w1hat}
\end{figure}
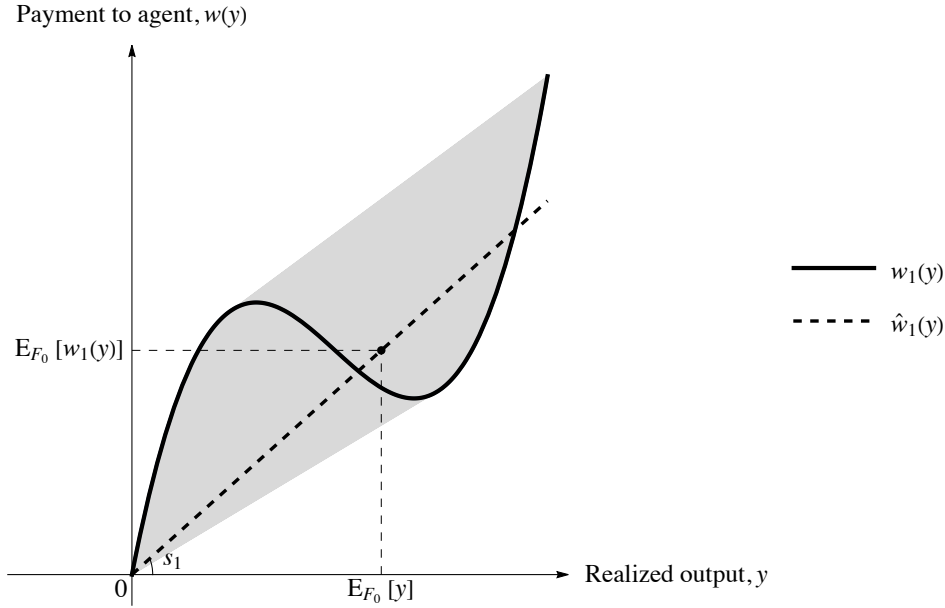
 
The solid curve represents the contract $w_t$, which may be nonlinear and non-monotonic. Consider the point whose coordinates are the expected output $\mathbb{E}_{F_t^0}\fkh{y}$ and the expected payment $\mathbb{E}_{F_t^0}\fkh{w_t\kh{y}}$ to agent $t$ if he takes action $a_t^0$. This point must lie within the convex hull of the curve $w_t$, represented by the shaded area in the figure. The constructed linear contract $\hat w_t$ is exactly the dashed line connecting the origin and this point, with slope $s_t$. Thus, in this case, $\hat w_t$ preserves the expected payment from action $a_t^0$, and hence both the agent's payoff and the principal's current payoff from that action.\footnote{This construction adapts to our model the alternative approach suggested by Lucas Maestri in \citet[Appendix~C]{Carroll15}.}

If the principal's current payoff from $a_t^0$ is nonpositive, then $\hat w_t$ is the zero contract, which is a linear improvement. Under $w_t$, nature can choose $A_t=K_t$, so that agent $t$ takes $a_t^0$ and the known set remains $K_t$; this holds the principal to at most $\beta V_{t+1}^*\kh{K_t}$. Under the zero contract, by contrast, her current payoff is nonnegative and her set of known actions can only grow, so she obtains at least $\beta V_{t+1}^*\kh{K_t}$. Together with the case of a positive payoff, discussed below, this yields the following lemma.
\begin{Lem}\label{lem:affine}
For every period $t$, every feasible history $h^t$, and every contract $w_t$, the linear contract $\hat w_t$ constructed in equation~\eqref{eqn:affine} satisfies $U_t\kh{\hat w_t\mid h^t}\ge U_t\kh{w_t\mid h^t}$.
\end{Lem}

When the principal's current payoff from $a_t^0$ is positive, suppose she offers the linear contract $\hat w_t$ and agent $t$ chooses action $a_t$ from his technology $A_t$. We need to show that the principal's interim payoff guarantee, $U_t\kh{\hat w_t\mid h^t,a_t}$, is at least $U_t\kh{w_t\mid h^t}$. If there exists another action $a_t'$, which may be taken by agent $t$ under $w_t$ and some other technology $A_t'$, such that
\eqn{U_t\kh{\hat w_t\mid h^t,a_t}\ge U_t\kh{w_t\mid h^t,a_t'}\label{eqn:ineq4}}
holds, then $U_t\kh{\hat w_t\mid h^t,a_t}\ge U_t\kh{w_t\mid h^t,a_t'}\ge U_t\kh{w_t\mid h^t}$, and thus the desired conclusion is established. The proof of Lemma~\ref{lem:affine} explicitly constructs such an alternative action $a_t'$ for each possible $a_t$.

Specifically, the principal's interim payoff guarantee consists of two parts: her payoff in the current period, plus the discounted optimal continuation guarantee. If agent $t$'s response $a_t$ to $\hat w_t$ generates at least as much expected output as $a_t^0$, then $a_t^0$ itself serves as $a_t'$. Because $\hat w_t$ leaves the principal a fixed share of output, her current payoff from $a_t$ is at least her current payoff from $a_t^0$, which is the same under $\hat w_t$ and $w_t$. Since observing $a_t$ only enlarges the set of known actions, which every future technology must contain, her continuation guarantee is also at least as high. If $a_t$ generates less expected output, then $a_t'$ can be taken to be a diluted version of $a_t^0$ that nature could have added under $w_t$. In either case, equation~\eqref{eqn:ineq4} holds.

By establishing Lemma~\ref{lem:affine}, we have shown that any nonlinear contract can be improved by a linear one. To complete the proof of Theorem~\ref{prop:1grow}, it suffices to show that, within the class of linear contracts, the maximum of the principal's overall payoff guarantee exists. We characterize the overall payoff guarantee of an arbitrary linear contract and prove the existence of a maximum through its upper semicontinuity in the share.

\subsubsection{Proof Step 2: Existence of an Optimal Linear Contract}\label{subsec:payoffg}
For a linear contract $w_t\kh{y}=s_ty$ with $s_t\in\fkh{0,1}$, the overall payoff guarantee is
\eqn{U_t\kh{s_ty\mid h^t}=\inf_{A_t\supseteq K_t}\hkh{\max_{\kh{F_t,c_t}\in BR\kh{s_ty\mid A_t}}\hkh{\kh{1-s_t}\mathbb{E}_{F_t}\fkh{y}+\beta V_{t+1}^*\kh{K_t\cup\hkh{\kh{F_t,c_t}}}}}.\label{eqn:adv-linear-guarantee}}
\begin{Lem}\label{lem:optlinear}
For every period $t$ and every feasible history $h^t$, the function $s_t\mapsto U_t\kh{s_ty\mid h^t}$ is upper semicontinuous on $\fkh{0,1}$ and attains a maximum.
\end{Lem}

For each fixed technology $A_t\supseteq K_t$, the agent's best-response correspondence is compact-valued and upper hemicontinuous in the share. Upper semicontinuity of the continuation value therefore implies upper semicontinuity of the principal's interim payoff from her preferred best response. Taking the infimum over technologies preserves upper semicontinuity, and compactness of $\fkh{0,1}$ gives a maximizing share.

The appendix also establishes that $V_t^*\kh{K_t}$ is upper semicontinuous in $K_t$, is weakly increasing under set inclusion, and depends on $K_t$ only through the set of expected output and cost pairs of its actions. These are the three properties of the continuation value used above, so the argument can be repeated at each preceding date.

Combining Lemmas~\ref{lem:affine} and~\ref{lem:optlinear} and completing the backward induction, we prove the main result of this section, Theorem~\ref{prop:1grow}, which establishes the robust optimality of linear contracts in every period and after every feasible history.

\section{Constant Technology}\label{sec:const}
In the previous section, we have focused on the case of advancing technology ($A_0\subseteq A_1\subseteq\cdots\subseteq A_T$) and shown that linear contracts are robustly optimal in every period and after every feasible history under that notion of dynamic worst-case evaluation. This section analyzes an alternative notion with more restrictions: the case of \emph{constant technology} $A_1=\cdots=A_T=A$. Here, all agents have the same technology, which is unknown to the principal. In other words, nature can neither introduce new actions across periods nor remove existing ones.

The main result for the case of constant technology is Theorem~\ref{prop:1}, which shows that linear contracts are robustly optimal in every period but \emph{not after every feasible history}. Specifically, it establishes their robust optimality after every feasible history consisting of linear contracts. After histories containing nonlinear contracts, an optimal contract in any subsequent period may be formed from these contracts and thus be nonlinear, as Section~\ref{sec:period2} illustrates. Nonetheless, the initial history contains no previous contracts, so a linear contract is robustly optimal in the first period, and offering optimal linear contracts ensures that every subsequent history is of this form. Thus, linear contracts are optimal \emph{on the path} (Section~\ref{sec:period1}). The reason for obtaining different results compared to the previous case of advancing technology is a different and more subtle rule of updating, which we refer to as \emph{compatibility} (Definition~\ref{def:comp}).

\subsection{Rule of Updating: Compatibility}\label{sec:comp} 
After history $h^t=\kh{A_0,w_1,a_1,\ldots,w_{t-1},a_{t-1}}$, the principal knows that every action in $K_t=K\kh{h^t}=A_0\cup\hkh{a_1,\ldots,a_{t-1}}$ is available. Since the technology is constant, each earlier agent $\tau<t$ chose $a_\tau$ from the same technology, so $a_\tau$ must have been a best response to $w_\tau$ among all of its actions. We call the technologies with both properties \emph{compatible} with $h^t$.
\begin{Def}[Compatibility]\label{def:comp}
A technology $A$ is \emph{compatible} with history $h^t$ if:
\begin{enumerate}
\item $A\supseteq K_t$.
\item $\mathbb{E}_{F}\fkh{w_\tau\kh{y}}-c\le\mathbb{E}_{F_\tau}\fkh{w_\tau\kh{y}}-c_\tau$, for every $\kh{F,c}\in A$ and every $\tau<t$.
\end{enumerate}
\end{Def}
A technology $A$ is compatible with $h^t$ if it contains every known action, and does not contain any action strictly better than $a_\tau$ under $w_\tau$ for any $\tau<t$. The first requirement in Definition~\ref{def:comp} indicates that the principal learns that the actions $a_1,\ldots,a_{t-1}$ exist (in addition to the initially known actions in $A_0$), and believes that agent $t$ may also take any of them again. The second requirement in Definition~\ref{def:comp} captures the additional inference she can draw from earlier agents' rationality in this case of constant technology: the true technology $A$ cannot contain any action $\kh{F,c}$ that would have given some agent $\tau<t$ a strictly higher payoff than $a_\tau$, that is, it is impossible that $\mathbb{E}_{F}\fkh{w_\tau\kh{y}}-c>\mathbb{E}_{F_\tau}\fkh{w_\tau\kh{y}}-c_\tau$.

Under constant technology, $\Omega\kh{h^t}$ consists of all technologies compatible with $h^t$. After the principal offers $w_t$ and observes $a_t$, the updated family is
\eqn{\Omega\kh{h^{t+1}}=\hkh{A\in\Omega\kh{h^t}:a_t\in BR\kh{w_t\mid A}}.\label{eqn:const-updating}}
Thus, each observation adds a known action and a new incentive restriction, while all earlier restrictions remain in force. The shapes of the previous contracts determine these restrictions. Consequently, the continuation problem depends on the previous contracts and associated expected payoffs as well as the known set. We write $\hat U_t$ and $\hat V_t^*$ for the payoff guarantees defined in Section~\ref{subsec:model-guarantees} when the family of possible technologies is given by compatibility.

\subsection{Final-Period Analysis}\label{sec:period2}
We begin with the final period, when the principal has observed the contracts and actions in the first $T-1$ periods. We fully characterize the principal's optimal final-period payoff guarantee, $\hat V_T^*\kh{h^T}$, and identify the contract that attains it in various cases. The analysis reveals three ways the principal may respond to the knowledge gained from observing $a_1,\ldots,a_{T-1}$.

The main result for the final-period analysis is Lemma~\ref{lem:second}, which shows that $\hat V_T^*\kh{h^T}$ is achieved by offering the best among three kinds of contracts: (i) a previous contract $w_\tau$ again, (ii) a modified $w_\tau$ with compensation for agent $T$, and (iii) a linear contract that corresponds to the optimal static contract in \cite{Carroll15} for the known set $K_T$.

Lemma~\ref{lem:second} reveals that the analysis in this section is not a direct adaptation of the single-period problem in \cite{Carroll15}, since an optimal contract may be nonlinear. This difference is precisely due to the second requirement of compatibility, where the principal draws additional inferences from the rationality of earlier agents, excluding certain actions. Note that the analysis is also not covered by the work of \cite{WCFC}, which establishes a general static framework that allows for rich organizational structures, and identifies two properties of the counterparty's possible responses which jointly imply that a linear contract solves the principal's single-period maxmin problem. Specifically, their \emph{Richness} property requires that the set of possible responses to a given contract be sufficiently and unboundedly broad. The Richness property is violated in the case of constant technology exactly because of the principal's exploration and inference in earlier periods: the true technology cannot contain any action that is strictly better for an earlier agent $\tau$ than the observed action $a_\tau$ under the contract $w_\tau$.\footnote{The other property in \cite{WCFC}, \emph{Responsiveness}, indicates that the counterparty's behavior is responsive to the incentive provided by expected payment, and allows comparison of the principal's payoff guarantees from two different contracts. The Responsiveness property is satisfied in our model. As a converse result, \cite{WCFC} also show that Responsiveness is necessary for linearity under a strengthened version of Richness. This result is in parallel with our analysis, since it is Richness that is not satisfied in our model.}

Suppose that the principal has offered contracts $w_1,\ldots,w_{T-1}$ and observed actions $a_\tau=\kh{F_\tau,c_\tau}$ for $\tau<T$. She learns that the true technology $A$ is compatible with $h^T$; that is, it contains $K_T$, and does not contain any action strictly better than $a_\tau$ for agent $\tau$ under $w_\tau$, for any $\tau<T$.

In the final period, if she offers a previous contract again, $w_T=w_\tau$, then she knows that $a_\tau$ is still a best response for agent $T$, because all agents have the same technology. This guarantees her at least the payoff she obtained in period $\tau$, $\mathbb{E}_{F_\tau}\fkh{y-w_\tau\kh{y}}$. Lemma~\ref{lem:second} below shows that, in some cases, doing so is already optimal for the principal, which means that an optimal final-period contract may be nonlinear following nonlinear contracts in earlier periods.

Offering a previous contract again is only one response of the principal to the knowledge gained by observing $a_1,\ldots,a_{T-1}$, and there are plenty of other possible responses. For example, if another known action $\kh{F,c}\in K_T$ may lead to a higher payoff for the principal under $w_\tau$ (that is, $\mathbb{E}_{F}\fkh{y-w_\tau\kh{y}}>\mathbb{E}_{F_\tau}\fkh{y-w_\tau\kh{y}}$), then it might be tempting for the principal to try to obtain the payoff $\mathbb{E}_{F}\fkh{y-w_\tau\kh{y}}$ instead. However, achieving this payoff requires the principal to use $w_\tau$ to induce action $\kh{F,c}$, and this may violate agent $T$'s incentive constraint. Indeed, in period $\tau$, the chosen action $a_\tau$ provides agent $\tau$ with a (weakly) higher payoff compared to the known action $\kh{F,c}$, and this relationship gets transferred to the final period because all agents have the same technology. This gives rise to the following notion of the \emph{incentive gap}.
\begin{Def}[Incentive gap]\label{def:g}
Given a previous period $\tau<t$ and a known action $\kh{F,c}\in K_t$, the \emph{incentive gap}, $g_\tau\kh{F,c}$, denotes the difference in agent $\tau$'s payoff between choosing $a_\tau$ and $\kh{F,c}$ under $w_\tau$. Formally,
$$g_\tau\kh{F,c}\equiv\mathbb{E}_{F_\tau}\fkh{w_\tau\kh{y}}-c_\tau-\kh{\mathbb{E}_{F}\fkh{w_\tau\kh{y}}-c}.$$
\end{Def}
At a feasible history, compatibility implies $g_\tau\kh{F,c}\ge0$. If the principal wants to induce action $\kh{F,c}$ using a contract ``similar to'' $w_\tau$, then agent $T$ needs to be compensated for not choosing $a_\tau$, and the amount of compensation increases with the incentive gap $g_\tau\kh{F,c}$. Lemma~\ref{lem:second} shows that the incentive gap sometimes becomes a real cost. Specifically, suppose that $\mathbb{E}_{F}\fkh{y-w_\tau\kh{y}}>g_\tau\kh{F,c}$. The principal can then modify $w_\tau$ by giving agent $T$ an additional share of the residual output as compensation, and this guarantees her a final-period payoff of at least $\kh{\sqrt{\mathbb{E}_{F}\fkh{y-w_\tau\kh{y}}}-\sqrt{g_\tau\kh{F,c}}}^2$. Note that if the incentive gap is small, this value becomes close to $\mathbb{E}_{F}\fkh{y-w_\tau\kh{y}}$, and may be better for the principal than simply offering $w_\tau$ again.

After observing $a_1,\ldots,a_{T-1}$, the principal learns that the true technology $A$ must contain $K_T$. If the principal ignores the second requirement of compatibility (Definition~\ref{def:comp}) and applies the single-period problem in \cite{Carroll15}, her optimal guarantee would be equal to $\Phi\kh{K_T}^2$, achieved by applying the construction in Section~\ref{subsec:independent} with $K_T$ in place of $A_0$. With the additional inference in place, the guarantee from this procedure can only increase. Lemma~\ref{lem:second} shows that, when this payoff guarantee is larger than those from the previous two responses (a previous contract again, or a modified previous contract with compensation), it is optimal for the principal to offer this linear contract, and doing so exactly attains this payoff guarantee.

We are now ready to present Lemma~\ref{lem:second}, which establishes that the best of the three kinds of contracts described above is optimal. Write $z_+\equiv\max\hkh{z,0}$.
\begin{Lem}\label{lem:second}
Suppose the principal offers contracts $w_1,\ldots,w_{T-1}$, and each agent $\tau<T$ chooses $a_\tau=\kh{F_\tau,c_\tau}$ in response to $w_\tau$. The principal's optimal final-period payoff guarantee is $\hat V_T^*\kh{h^T}=\hat\Phi\kh{h^T}^2$, where
\eqn{\hat\Phi\kh{h^T}\equiv\max\hkh{\max_{\tau<T}\sqrt{\kh{\mathbb{E}_{F_\tau}\fkh{y-w_\tau\kh{y}}}_+},\max_{\substack{\tau<T\\\kh{F,c}\in K_T}}\hkh{\sqrt{\kh{\mathbb{E}_{F}\fkh{y-w_\tau\kh{y}}}_+}-\sqrt{g_\tau\kh{F,c}}},\Phi\kh{K_T}}.\label{eqn:optsecond}}
Specifically,
\begin{enumerate}
\item\label{part:1} If $\sqrt{\kh{\mathbb{E}_{F_\tau}\fkh{y-w_\tau\kh{y}}}_+}$ attains the maximum in equation~\eqref{eqn:optsecond} for some $\tau<T$, then the principal's optimal final-period payoff guarantee is achieved by $w_T=w_\tau$.
\item\label{part:2} If $\sqrt{\kh{\mathbb{E}_{F}\fkh{y-w_\tau\kh{y}}}_+}-\sqrt{g_\tau\kh{F,c}}$ attains the maximum in equation~\eqref{eqn:optsecond} for some $\tau<T$ and $\kh{F,c}\in K_T$, then the principal's optimal final-period payoff guarantee is achieved by
\eqn{w_T\kh{y}=w_\tau\kh{y}+m\cdot\kh{y-w_\tau\kh{y}}\quad\text{with}\quad m=\sqrt{\frac{g_\tau\kh{F,c}}{\mathbb{E}_{F}\fkh{y-w_\tau\kh{y}}}}\in[0,1).\label{eqn:m}}
\item\label{part:3} If $\Phi\kh{K_T}$ attains the maximum in equation~\eqref{eqn:optsecond}, then the principal's optimal final-period payoff guarantee is achieved by $w_T\kh{y}=s_Ty$ with $s_T=\sqrt{c/\mathbb{E}_{F}\fkh{y}}$, where $\kh{F,c}\in K_T$ attains $\Phi\kh{K_T}$.
\end{enumerate}
\end{Lem}

Since $g_\tau\kh{a_\tau}=0$, the first term in equation~\eqref{eqn:optsecond} equals the second term at $\kh{F,c}=a_\tau$, and Part~\ref{part:1} is Part~\ref{part:2} with $m=0$. We state this case separately because re-offering a previous contract is the most direct way to exploit past exploration: it requires no compensation, it guarantees the principal at least her payoff from period $\tau$, and it already shows that an optimal final-period contract may be nonlinear.

The proof of Lemma~\ref{lem:second} mainly consists of two parts. The first part is to prove that, when a term in equation~\eqref{eqn:optsecond} attains the maximum, offering the corresponding contract guarantees the principal at least $\hat\Phi\kh{h^T}^2$ in the final period. This requires providing lower bounds on the principal's final-period payoffs under every compatible technology. The second part is to show that, under arbitrary final-period contracts, the principal's payoff guarantee is not strictly higher than $\hat\Phi\kh{h^T}^2$. This requires constructing worst-case technologies. Each contains, besides the known actions, an action that the principal has not observed, whose output distribution dilutes that of a known best response to the contract; nonnegative cost and the historical incentive restrictions limit this dilution. The binding restriction shows that the payoff guarantee is no higher than the square of at least one term in equation~\eqref{eqn:optsecond}.

Note that compared to the case of advancing technology, the principal acquires more knowledge from observing $a_1,\ldots,a_{T-1}$ under constant technology. As an implication, her optimal final-period payoff guarantee takes a more complex form that depends directly on the previous contracts $w_1,\ldots,w_{T-1}$: how you exploit is related to how you explore.

Lemma~\ref{lem:second} indicates that, as long as a previous contract $w_\tau$ is nonlinear and the observed actions are such that a term of Part~\ref{part:1} or Part~\ref{part:2} for this $\tau$ attains the maximum in equation~\eqref{eqn:optsecond}, the principal's optimal final-period payoff guarantee $\hat V_T^*\kh{h^T}$ is achieved by a nonlinear contract. On the other hand, if all previous contracts are linear, then the contracts in all three parts of Lemma~\ref{lem:second} are linear. This shows that the optimal way for the principal to respond to the knowledge gained is closely related to the specific approach she chooses to explore in earlier periods.

\subsection{Optimal Contracts}\label{sec:period1}
In the previous section, we have focused on the principal's problem in the final period and fully characterized her optimal final-period payoff guarantee. This section analyzes the principal's problem in earlier periods, that is, choosing a contract $w_t$ after each feasible history $h^t$ to maximize her overall payoff guarantee $\hat U_t\kh{w_t\mid h^t}$, which also accounts for how the current contract affects her continuation guarantee.

The relevant contract class now allows combinations of several previous contracts. For a vector of coefficients $\theta_t=\kh{\alpha,\eta_1,\ldots,\eta_{t-1}}$, define the coefficient set
$$\Theta_t\equiv\hkh{\theta_t\in\mathbb{R}_+^t:\alpha+\sum_{\tau<t}\eta_\tau\le1}.$$
For $\theta_t\in\Theta_t$, define the contract
$$w_{\theta_t}\kh{y}\equiv\alpha y+\sum_{\tau<t}\eta_\tau w_\tau\kh{y}.$$
These contracts form the convex hull $\operatorname{co}\hkh{0,y,w_1,\ldots,w_{t-1}}$, where $0$ and $y$ denote the linear contracts with shares zero and one.

We first state the main result of this section, Theorem~\ref{prop:1}, which establishes that this convex hull contains an optimal contract after every feasible history.
\begin{Thm}\label{prop:1}
Under constant technology, for every feasible history $h^t$ with $t=1,\ldots,T$,
\eqn{\hat V_t^*\kh{h^t}=\max_{\theta_t\in\Theta_t}\hat U_t\kh{w_{\theta_t}\mid h^t}.\label{eqn:optimalclassconst}}
Consequently, there exists an optimal policy that offers a linear contract after every feasible history consisting of linear contracts.
\end{Thm}
Although Lemma~\ref{lem:second} shows that, following nonlinear contracts in earlier periods, optimal final-period contracts may also be nonlinear, the principal's optimal overall payoff guarantee is achieved through a linear first-period contract, together with optimally chosen linear contracts in all subsequent periods. The principal has the opportunity to explore in earlier periods, and linear contracts are optimal in terms of utilizing this exploration opportunity, making them even more robust. The linearity follows directly from the form of the convex hull. At the initial history, there are no previous contracts, so $\operatorname{co}\hkh{0,y}$ consists of linear contracts. More generally, if $w_\tau$ is linear for every $\tau<t$, then every contract in $\operatorname{co}\hkh{0,y,w_1,\ldots,w_{t-1}}$ is linear.

Along a history of linear contracts $w_\tau\kh{y}=s_\tau y$, the incentive restrictions take a particularly simple form: every action $\kh{F,c}$ of a compatible technology satisfies
$$c\ge\max\hkh{0,\,c_1+s_1\kh{\mathbb{E}_{F}\fkh{y}-\mathbb{E}_{F_1}\fkh{y}},\ldots,c_{t-1}+s_{t-1}\kh{\mathbb{E}_{F}\fkh{y}-\mathbb{E}_{F_{t-1}}\fkh{y}}}.$$
These restrictions involve an action only through its expected output and cost. Each observation therefore adds a lower bound on the cost of an action with a given expected output, and information accumulates through both the known actions and these cost bounds.

Similar to Theorem~\ref{prop:1grow}, the proof of Theorem~\ref{prop:1} first improves any contract to one in $\operatorname{co}\hkh{0,y,w_1,\ldots,w_{t-1}}$ and then establishes the existence of an optimal contract in this convex hull. It uses two backward inductions. The first summarizes each history by a state that records, for every known action, its expected output, cost, and expected payments under previous contracts, together with the payoffs revealed by previous choices, and shows that the continuation value, extended to a convex set of such states, is convex and weakly increasing in the expected outputs and payments of known actions. The second shows that this extension equals $\hat V_t^*$ at feasible histories and establishes contract improvement and attainment.

Both inductions start from the final period, where Lemma~\ref{lem:second} applies. At each earlier date, the argument has three steps. First, Lemma~\ref{lem:envelopeconst} shows that the principal's overall payoff guarantee from a contract depends only on its \emph{upper payment envelope}, defined below, and on its expected payments to the known actions. Second, Lemma~\ref{lem:improvementconst} uses a supporting hyperplane of the envelope to replace any contract with one in the convex hull that does at least as well. Third, Lemma~\ref{lem:existenceconst} shows that the best contract in the convex hull exists.

\subsubsection{Proof Step 1: The Upper Payment Envelope}\label{sec:envelope}
The first step shows that the principal's overall payoff guarantee depends on a contract only through two of its features. Fix a period $t$, a feasible history $h^t$, and a current contract $w_t$. Under advancing technology, an observed action matters for the future only through its expected output and cost. Under constant technology, its expected payments under the earlier contracts matter as well, because they determine whether it satisfies the incentive restrictions from earlier periods. For an output distribution $F$, consider the moments
$$x=\mathbb{E}_{F}\fkh{y},\quad q_\tau=\mathbb{E}_{F}\fkh{w_\tau\kh{y}},\quad\tau<t,$$
and write $q=\kh{q_1,\ldots,q_{t-1}}$. For each $\kh{x,q}$ attained by some output distribution, define the \emph{upper payment envelope} of $w_t$ by
$$\overline p_{w_t}\kh{x,q}\equiv\max_{\substack{F\in\Delta\kh{Y}\\\mathbb{E}_{F}\fkh{y}=x\\\mathbb{E}_{F}\fkh{w_\tau\kh{y}}=q_\tau,\ \tau<t}}\mathbb{E}_{F}\fkh{w_t\kh{y}}.$$
Equivalently, $\overline p_{w_t}$ is the smallest concave function of the moments satisfying $$\overline p_{w_t}\kh{y,w_1\kh{y},\ldots,w_{t-1}\kh{y}}\ge w_t\kh{y}\quad \text{for every }y\in Y.$$ It is also upper semicontinuous. We call the left-hand side of this inequality, viewed as a function of output, the \emph{history-relative concavification} of $w_t$ and denote it by $w_t^{\mathrm{cav}}$, so the inequality reads $w_t^{\mathrm{cav}}\ge w_t$. Since the previous contracts may be nonlinear, $w_t^{\mathrm{cav}}$ need not be concave in output.\footnote{Since $\overline p_{w_t}$ is only upper semicontinuous, $w_t^{\mathrm{cav}}$ need not be continuous when $Y$ is infinite. Corollary~\ref{cor:history-relative-concavification} below therefore concerns continuous contracts between $w_t$ and $w_t^{\mathrm{cav}}$. When $Y$ is finite, continuity is automatic.} At the initial history, there are no historical-payment moments, and $w_1^{\mathrm{cav}}=\overline p_{w_1}$ is the ordinary concave closure of $w_1$, the upper boundary of the shaded region in Figure~\ref{fig:w1hat}.

The envelope does not record everything the principal needs: the expected payment that $w_t$ gives a known action may lie strictly below the envelope at that action's moments. The following lemma shows that these two features suffice.
\begin{Lem}\label{lem:envelopeconst}
Fix a period $t$ and a feasible history $h^t$. If two current contracts $w_t$ and $w_t'$ satisfy
$\overline p_{w_t'}\kh{x,q}=\overline p_{w_t}\kh{x,q}$ for every attainable $\kh{x,q}$, and
$\mathbb{E}_{F}\fkh{w_t'\kh{y}}\ge\mathbb{E}_{F}\fkh{w_t\kh{y}}$ for every $\kh{F,c}\in K_t$, then
$\hat U_t\kh{w_t'\mid h^t}\ge\hat U_t\kh{w_t\mid h^t}$.
\end{Lem}
In particular, two current contracts with the same upper payment envelope and the same expected payments to all known actions have the same overall payoff guarantee.

Under both contracts, nature faces the same upper boundary of attainable payments at every vector of moments. Call an action of a compatible technology that gives the agent a strictly higher payoff than every known action a \emph{strict new response}. If a strict new response is possible, the higher payments to the known actions under $w_t'$ raise the payoff that such a response must give the agent and weakly raise the continuation value, so nature can harm the principal no more than under $w_t$. If no strict new response is possible, the agent chooses a known action under both contracts. The best such action receives the same payment under both, so the principal's current payoff is unchanged, and in neither case does she learn anything new. The proof treats these two cases separately.

The lemma has a useful consequence for the history-relative concavification.
\begin{Cor}[History-relative concavification]\label{cor:history-relative-concavification}
If a contract $w_t'$ satisfies $w_t\le w_t'\le w_t^{\mathrm{cav}}$, then $\hat U_t\kh{w_t'\mid h^t}\ge\hat U_t\kh{w_t\mid h^t}$.
\end{Cor}
Indeed, $w_t'\ge w_t$ gives $\overline p_{w_t'}\ge\overline p_{w_t}$ and raises the payments to the known actions. Conversely, for every output distribution $F$, Jensen's inequality and the concavity of $\overline p_{w_t}$ give
$$\mathbb{E}_{F}\fkh{w_t'\kh{y}}\le\mathbb{E}_{F}\fkh{w_t^{\mathrm{cav}}\kh{y}}\le\overline p_{w_t}\kh{\mathbb{E}_{F}\fkh{y},\mathbb{E}_{F}\fkh{w_1\kh{y}},\ldots,\mathbb{E}_{F}\fkh{w_{t-1}\kh{y}}},$$
so $\overline p_{w_t'}\le\overline p_{w_t}$, and Lemma~\ref{lem:envelopeconst} applies.

Lemma~\ref{lem:envelopeconst} compares contracts with the same upper payment envelope. The contract constructed in the next step generally has a different envelope, so showing that it does at least as well requires a further argument.

\subsubsection{Proof Step 2: Improving Contracts into the Convex Hull}
The second step establishes that the principal loses nothing by restricting attention to contracts in the convex hull.
\begin{Lem}\label{lem:improvementconst}
For every period $t$, every feasible history $h^t$, and every current contract $w_t$, there exists a contract $w_t'\in\operatorname{co}\hkh{0,y,w_1,\ldots,w_{t-1}}$ such that $\hat U_t\kh{w_t'\mid h^t}\ge\hat U_t\kh{w_t\mid h^t}$.
\end{Lem}
Fix a feasible history $h^t$ and a current contract $w_t$. The proof distinguishes two cases. Suppose first that $w_t$ admits no strict new response. Then the agent chooses a known action. Under constant technology, such a choice usually still informs the principal, since it shows that no action gives the agent a higher payoff under $w_t$. Here, however, this restriction rules out no technology that was possible before: no action of a compatible technology gives the agent more than the best known action. Now treat period $t$ as if it were the final period, and consider the contract that Lemma~\ref{lem:second} prescribes. This contract belongs to $\operatorname{co}\hkh{0,y,w_1,\ldots,w_{t-1}}$. It guarantees at least the same current payoff as $w_t$, and observing the agent's response to it can only add information, so it weakly outperforms $w_t$.

Suppose instead that $w_t$ admits a strict new response. Nature's choice then matters only through the agent's response, which determines both the principal's current payoff and what she learns. Both depend on the response only through its moments $\kh{x,q}$, its expected payment $p=\mathbb{E}_{F}\fkh{w_t\kh{y}}$ under the current contract, and the agent's payoff $u=p-c$: the current payoff is $x-p$, and by the first induction, the continuation value depends on the new observation only through these numbers. Nature's problem therefore reduces to choosing $\kh{x,q,p,u}$, and the appendix shows that the resulting problem is convex.\footnote{The problem also includes responses that the agent only weakly prefers to the known actions. These are limits of responses that he strictly prefers, so including them is consistent with the tie-breaking rule.} For fixed $\kh{x,q}$ and $u$, nature can choose the largest possible payment, $p=\overline p_{w_t}\kh{x,q}$: raising $p$ together with the cost leaves the agent's payoff unchanged, lowers the principal's current payoff, and does not raise her continuation value. Nature's worst response therefore lies on the upper payment envelope.

A saddle-point argument for this convex problem selects a worst-case moment vector together with an affine function that supports the upper payment envelope from above:
$$\alpha x+\sum_{\tau<t}\eta_\tau q_\tau+\kappa\ge\overline p_{w_t}\kh{x,q}\quad\text{for every attainable }\kh{x,q},$$
with equality at the worst-case vector. The saddle-point conditions imply $\theta_t=\kh{\alpha,\eta_1,\ldots,\eta_{t-1}}\in\Theta_t$, and evaluating the inequality at the zero moment vector, where the envelope is zero, gives $\kappa\ge0$. Replacing $x$ by $y$ and each $q_\tau$ by $w_\tau\kh{y}$ turns the affine function into a function of output, $w_{\theta_t}\kh{y}+\kappa$, which we call the \emph{supporting payment schedule}. By the inequality above, it pays at least $w_t^{\mathrm{cav}}\kh{y}$, and hence at least $w_t\kh{y}$, at every output $y$. However, it pays $\kappa$ at zero output. The proof first raises the payments to the known actions up to those under the supporting payment schedule, which weakly raises the continuation value. It then removes the intercept $\kappa$ from the current-period payments and from the payoff revealed by the current choice, leaving everything recorded about earlier periods unchanged; this again weakly raises the continuation value. The result is the contract $w_{\theta_t}$, which pays nothing at zero output and lies in $\operatorname{co}\hkh{0,y,w_1,\ldots,w_{t-1}}$.

At the initial history, there are no previous contracts, so the improved contract is linear. The appendix handles the remaining technical steps, including Slater's condition and the boundary cases of the supporting argument.

\subsubsection{Proof Step 3: Existence of an Optimal Contract}
The third step shows that the principal's overall payoff guarantee attains a maximum over contracts in the convex hull.
\begin{Lem}\label{lem:existenceconst}
For every period $t$ and every feasible history $h^t$, the overall payoff guarantee $\hat U_t\kh{w_{\theta_t}\mid h^t}$ is upper semicontinuous in $\theta_t$ and attains a maximum over $\Theta_t$.
\end{Lem}

The coefficient set is compact. The appendix establishes upper semicontinuity of the payoff guarantee jointly in the coefficients and the feasible history. Maximizing over the coefficients therefore gives an optimal contract and preserves upper semicontinuity of the optimal value in the history.

Starting from the final-period characterization in Lemma~\ref{lem:second}, the backward induction establishes the continuation properties needed for contract improvement and the regularity needed for existence at each earlier date. Combining Lemmas~\ref{lem:improvementconst} and~\ref{lem:existenceconst} proves Theorem~\ref{prop:1}.

\section{Conclusion}\label{sec:concln}
This paper studies a moral hazard problem over an arbitrary finite horizon. The principal does not know the agents' action sets and demands contracts that are robust to this uncertainty. She observes each chosen action and uses what she has learned to choose later contracts, so she can explore in every period before the final one. We introduce three notions of dynamic worst-case evaluation, define a rule of updating for each, and characterize the principal's optimal payoff guarantee. The characterization shows how she should respond to new knowledge and which contracts best acquire it. The results show that linear contracts are robustly optimal not just in static settings, but also in dynamic environments with exploration.

The three notions differ in how far this optimality extends. Under advancing technology, linear contracts are optimal after every feasible history. Under constant technology, the principal also learns from the actions earlier agents did not choose, and after nonlinear contracts this can make a nonlinear contract optimal. Linear contracts are nonetheless optimal on the path. This last result rests on a recursive argument based on the convexity and monotonicity of the continuation value.

More broadly, this paper offers one way to bring exploration into robust mechanism design. Even when uncertainty is nonquantifiable, a designer who interacts repeatedly can learn about the environment through an appropriate rule of updating. We hope this approach will prove useful in other models. 

\bibliographystyle{myref}
\bibliography{RCE}

\newpage
\begin{appendices}
\setstretch{1}

\renewcommand{\theequation}{\thesection.\arabic{equation}}
\setcounter{equation}{0}
\renewcommand{\theLem}{\thesection.\arabic{Lem}}
\setcounter{Lem}{0}
\renewcommand{\theDef}{\thesection.\arabic{Def}}
\setcounter{Def}{0}
\renewcommand{\theProp}{\thesection.\arabic{Prop}}
\setcounter{Prop}{0}

\section{Proofs of Results in the Main Text}\label{app:proof}

\subsection{Proofs for Section~\ref{sec:adv}}\label{app:adv}

Under advancing technology, a history $h^t$ enters the continuation problem only through its known set $K_t=K\kh{h^t}$. The proofs therefore work with an arbitrary known set, by which we mean a compact set $B$ of actions containing the null action. For such a set $B$, a contract $w_t$, and an action $a_t=\kh{F_t,c_t}$, define
\eqns{U_t\kh{w_t\mid B,a_t}&\equiv\mathbb{E}_{F_t}\fkh{y-w_t\kh{y}}+\beta V_{t+1}^*\kh{B\cup\hkh{a_t}},\\U_t\kh{w_t\mid B}&\equiv\inf_{A\supseteq B}\max_{a_t\in BR\kh{w_t\mid A}}U_t\kh{w_t\mid B,a_t},\\V_t^*\kh{B}&\equiv\sup_{w_t}U_t\kh{w_t\mid B},}
where $V_{T+1}^*\equiv0$, $A$ ranges over technologies, and ties among best responses are broken in the principal's favor as in Section~\ref{subsec:model-guarantees}. Since $\Omega\kh{h^t}$ consists of the technologies containing $K_t$ and equation~\eqref{eqn:model-known-actions} gives $K\kh{h^t,w_t,a_t}=K_t\cup\hkh{a_t}$, backward induction shows that these definitions agree with equations~\eqref{eqn:model-interim}-\eqref{eqn:model-value} at every feasible history $h^t$: $U_t\kh{w_t\mid h^t,a_t}=U_t\kh{w_t\mid K_t,a_t}$, $U_t\kh{w_t\mid h^t}=U_t\kh{w_t\mid K_t}$, and $V_t^*\kh{h^t}=V_t^*\kh{K_t}$.

Fix a metric on the set of actions that induces its product topology, and equip known sets and technologies with the induced Hausdorff metric. Under this metric, adjoining an action to a compact set and taking the union of two compact sets are continuous operations. Moreover, if technologies $A^n$ converge to $A$, then $A\cup\bigcup_nA^n$ is compact, and every limit point of a sequence with $a^n\in A^n$ belongs to $A$.

Let $\mathcal M\kh{B}\equiv\hkh{\kh{\mathbb{E}_{F}\fkh{y},c}:\kh{F,c}\in B}$ denote the set of expected output and cost pairs of a known set $B$, and define $\mathcal M\kh{A}$ in the same way for a technology $A$. We prove Theorem~\ref{prop:1grow} by backward induction. Along with linear optimality, the induction shows that, at every date $t$, the function $V_t^*$ has the following properties.
\begin{enumerate}
\item\label{item:advapp-bound} $V_t^*$ satisfies $0\le V_t^*\kh{B}\le\kh{\max Y}\sum_{j=0}^{T-t}\beta^j$ for every known set $B$, with the empty sum equal to zero when $t=T+1$.
\item\label{item:advapp-usc} $V_t^*$ is upper semicontinuous.
\item\label{item:advapp-mono} $V_t^*$ is weakly increasing under set inclusion.
\item\label{item:advapp-pairs} $V_t^*\kh{B}=V_t^*\kh{B'}$ whenever $B$ and $B'$ are known sets with $\mathcal M\kh{B}=\mathcal M\kh{B'}$.
\end{enumerate}
The induction covers every known set, whether or not it contains an action with positive surplus.

If $V_{t+1}^*$ satisfies properties~\ref{item:advapp-bound} and~\ref{item:advapp-usc}, then $a_t\mapsto U_t\kh{w_t\mid B,a_t}$ is upper semicontinuous for every known set $B$ and contract $w_t$, because $w_t$ is continuous and adjoining an action is continuous. Since $BR\kh{w_t\mid A}$ is nonempty and compact, the maximum in the definition of $U_t\kh{w_t\mid B}$ is attained.

The proofs below follow the order of Section~\ref{sec:adv}. The proofs of Lemmas~\ref{lem:affine} and~\ref{lem:optlinear} carry out the two steps of Section~\ref{sec:advopt} at a date $t$, taking the properties at date $t+1$ as given, and the proof of Theorem~\ref{prop:1grow} completes the induction.

\begin{proof}[Proof of Lemma~\ref{lem:affine}]
Fix a date $t$, and suppose that $V_{t+1}^*$ satisfies properties~\ref{item:advapp-bound}-\ref{item:advapp-pairs}. We prove the inequality for every known set $B$, and the lemma is the case $B=K_t$. Fix $B$ and a contract $w_t$, let $a_t^0=\kh{F_t^0,c_t^0}$ be agent $t$'s best response when his technology is $B$, selected under the tie-breaking convention, and let $\hat w_t$ be the linear contract in equation~\eqref{eqn:affine}. Every action in $B$ leaves the known set unchanged, so $U_t\kh{w_t\mid B,a_t^0}=\mathbb{E}_{F_t^0}\fkh{y-w_t\kh{y}}+\beta V_{t+1}^*\kh{B}$. Since nature can choose the technology $B$, under which agent $t$ takes $a_t^0$,
\eqn{U_t\kh{w_t\mid B}\le\mathbb{E}_{F_t^0}\fkh{y-w_t\kh{y}}+\beta V_{t+1}^*\kh{B}.\label{eqn:advapp-known-upper-bound}}

Suppose first that $\mathbb{E}_{F_t^0}\fkh{y-w_t\kh{y}}\le0$, so that $\hat w_t$ is the zero contract. For every technology $A\supseteq B$ and every $\kh{F,c}\in BR\kh{0\mid A}$, the principal's current payoff $\mathbb{E}_{F}\fkh{y}$ is nonnegative because $Y\subseteq\mathbb{R}_+$, and property~\ref{item:advapp-mono} gives $V_{t+1}^*\kh{B\cup\hkh{\kh{F,c}}}\ge V_{t+1}^*\kh{B}$. Hence $U_t\kh{0\mid B}\ge\beta V_{t+1}^*\kh{B}$, and equation~\eqref{eqn:advapp-known-upper-bound} gives
$$U_t\kh{w_t\mid B}\le\mathbb{E}_{F_t^0}\fkh{y-w_t\kh{y}}+\beta V_{t+1}^*\kh{B}\le\beta V_{t+1}^*\kh{B}\le U_t\kh{0\mid B}.$$

Suppose now that $\mathbb{E}_{F_t^0}\fkh{y-w_t\kh{y}}>0$. Then $\mathbb{E}_{F_t^0}\fkh{y}>\mathbb{E}_{F_t^0}\fkh{w_t\kh{y}}\ge0$, so the share in equation~\eqref{eqn:affine} satisfies $0\le s_t<1$ and $s_t\mathbb{E}_{F_t^0}\fkh{y}=\mathbb{E}_{F_t^0}\fkh{w_t\kh{y}}$. Fix a technology $A\supseteq B$, and let $a_t=\kh{F_t,c_t}$ be the best response to $\hat w_t$ in $A$ selected under the tie-breaking convention. Since $a_t^0\in A$,
\eqn{s_t\mathbb{E}_{F_t}\fkh{y}-c_t\ge s_t\mathbb{E}_{F_t^0}\fkh{y}-c_t^0=\mathbb{E}_{F_t^0}\fkh{w_t\kh{y}}-c_t^0.\label{eqn:advapp-agent-comparison}}
We show that $U_t\kh{\hat w_t\mid B,a_t}\ge U_t\kh{w_t\mid B}$.

If $\mathbb{E}_{F_t}\fkh{y}\ge\mathbb{E}_{F_t^0}\fkh{y}$, then, since $s_t<1$ and $B\subseteq B\cup\hkh{a_t}$, property~\ref{item:advapp-mono} and equation~\eqref{eqn:advapp-known-upper-bound} give
\eqns{U_t\kh{\hat w_t\mid B,a_t}&=\kh{1-s_t}\mathbb{E}_{F_t}\fkh{y}+\beta V_{t+1}^*\kh{B\cup\hkh{a_t}}\\&\ge\kh{1-s_t}\mathbb{E}_{F_t^0}\fkh{y}+\beta V_{t+1}^*\kh{B}\\&=\mathbb{E}_{F_t^0}\fkh{y-w_t\kh{y}}+\beta V_{t+1}^*\kh{B}\ge U_t\kh{w_t\mid B}.}

If equation~\eqref{eqn:advapp-agent-comparison} holds with equality, then $a_t^0$ is also a best response to $\hat w_t$ in $A$. Since $\hat w_t$ and $w_t$ give the same expected payment under $a_t^0$, the tie-breaking convention and equation~\eqref{eqn:advapp-known-upper-bound} give
$$U_t\kh{\hat w_t\mid B,a_t}\ge U_t\kh{\hat w_t\mid B,a_t^0}=\mathbb{E}_{F_t^0}\fkh{y-w_t\kh{y}}+\beta V_{t+1}^*\kh{B}\ge U_t\kh{w_t\mid B}.$$

In the remaining case, $\mathbb{E}_{F_t}\fkh{y}<\mathbb{E}_{F_t^0}\fkh{y}$ and the inequality in equation~\eqref{eqn:advapp-agent-comparison} is strict, and we dilute $a_t^0$. Let $\lambda=\mathbb{E}_{F_t}\fkh{y}/\mathbb{E}_{F_t^0}\fkh{y}\in[0,1)$, $F_t'=\lambda F_t^0+\kh{1-\lambda}\delta_0$, and $a_t'=\kh{F_t',c_t}$, an action that draws output from $F_t^0$ with probability $\lambda$, yields zero output otherwise, and has the cost of $a_t$. Since $w_t\kh{0}=0$,
$$\mathbb{E}_{F_t'}\fkh{y}=\mathbb{E}_{F_t}\fkh{y},\quad\mathbb{E}_{F_t'}\fkh{w_t\kh{y}}=\lambda\mathbb{E}_{F_t^0}\fkh{w_t\kh{y}}=s_t\mathbb{E}_{F_t}\fkh{y}.$$
Together with the strict inequality in equation~\eqref{eqn:advapp-agent-comparison} and $a_t^0\in BR\kh{w_t\mid B}$, this gives
$$\mathbb{E}_{F_t'}\fkh{w_t\kh{y}}-c_t>\mathbb{E}_{F_t^0}\fkh{w_t\kh{y}}-c_t^0=\max_{\kh{F,c}\in B}\hkh{\mathbb{E}_{F}\fkh{w_t\kh{y}}-c}.$$
Thus $a_t'$ is the unique best response to $w_t$ in the technology $B\cup\hkh{a_t'}$, and the principal's current payoff from it is $\mathbb{E}_{F_t'}\fkh{y-w_t\kh{y}}=\kh{1-s_t}\mathbb{E}_{F_t}\fkh{y}$. Moreover, $a_t'$ and $a_t$ have the same expected output and cost, so $\mathcal M\kh{B\cup\hkh{a_t'}}=\mathcal M\kh{B\cup\hkh{a_t}}$, and property~\ref{item:advapp-pairs} gives $V_{t+1}^*\kh{B\cup\hkh{a_t'}}=V_{t+1}^*\kh{B\cup\hkh{a_t}}$. Since nature can choose the technology $B\cup\hkh{a_t'}$,
$$U_t\kh{w_t\mid B}\le U_t\kh{w_t\mid B,a_t'}=\kh{1-s_t}\mathbb{E}_{F_t}\fkh{y}+\beta V_{t+1}^*\kh{B\cup\hkh{a_t}}=U_t\kh{\hat w_t\mid B,a_t}.$$

Thus, for every technology $A\supseteq B$, the response selected under $\hat w_t$ gives the principal an interim payoff guarantee of at least $U_t\kh{w_t\mid B}$. Taking the infimum over technologies proves $U_t\kh{\hat w_t\mid B}\ge U_t\kh{w_t\mid B}$.
\end{proof}

\begin{proof}[Proof of Lemma~\ref{lem:optlinear}]
Fix a date $t$, and suppose that $V_{t+1}^*$ satisfies properties~\ref{item:advapp-bound} and~\ref{item:advapp-usc}. We show that $\kh{s,B}\mapsto U_t\kh{sy\mid B}$ is jointly upper semicontinuous in the share $s\in\fkh{0,1}$ and the known set $B$, and that the maximum of $U_t\kh{s_ty\mid B}$ over $s_t\in\fkh{0,1}$ exists and is upper semicontinuous in $B$. With $B=K_t$, these statements give the lemma.
 
Take a sequence $\kh{s^n,B^n}\to\kh{s,B}$, fix a technology $A\supseteq B$, and let $A^n=A\cup B^n$. Then $A^n$ is a technology containing $B^n$, and $A^n\to A\cup B=A$. Let $a^n=\kh{F^n,c^n}$ be the best response to $s^ny$ in $A^n$ selected under the tie-breaking convention. Since nature can choose $A^n$,
\eqn{U_t\kh{s^ny\mid B^n}\le U_t\kh{s^ny\mid B^n,a^n}=\kh{1-s^n}\mathbb{E}_{F^n}\fkh{y}+\beta V_{t+1}^*\kh{B^n\cup\hkh{a^n}}.\label{eqn:advapp-usc-comparison}}
By property~\ref{item:advapp-bound}, $U_t\kh{s^ny\mid B^n,a^n}$ is bounded in $n$. Together with compactness of $A\cup\bigcup_nA^n$, this gives a subsequence along which $U_t\kh{s^ny\mid B^n,a^n}$ converges to its limit superior and $a^n$ converges to some $a=\kh{F,c}\in A$. Every $\kh{F',c'}\in A$ belongs to $A^n$, so $s^n\mathbb{E}_{F^n}\fkh{y}-c^n\ge s^n\mathbb{E}_{F'}\fkh{y}-c'$ for every $n$, and taking limits along this subsequence shows that $a\in BR\kh{sy\mid A}$. Moreover, $B^n\cup\hkh{a^n}\to B\cup\hkh{a}$ along this subsequence. Equation~\eqref{eqn:advapp-usc-comparison} and property~\ref{item:advapp-usc} therefore give
\eqns{\limsup_nU_t\kh{s^ny\mid B^n}&\le\kh{1-s}\mathbb{E}_{F}\fkh{y}+\beta V_{t+1}^*\kh{B\cup\hkh{a}}\\&=U_t\kh{sy\mid B,a}\le\max_{a_t\in BR\kh{sy\mid A}}U_t\kh{sy\mid B,a_t}.}
Since $A\supseteq B$ was arbitrary, taking the infimum over technologies gives $\limsup_nU_t\kh{s^ny\mid B^n}\le U_t\kh{sy\mid B}$, which proves joint upper semicontinuity.
 
For each $B$, upper semicontinuity in the share and compactness of $\fkh{0,1}$ give a maximizing share. Now let $B^n\to B$, and let $s^n$ maximize $U_t\kh{s_ty\mid B^n}$ over $s_t\in\fkh{0,1}$. By property~\ref{item:advapp-bound}, these maximized values are bounded, so compactness of $\fkh{0,1}$ gives a subsequence along which $U_t\kh{s^ny\mid B^n}$ converges to its limit superior and $s^n$ converges to some $s\in\fkh{0,1}$. Joint upper semicontinuity, applied along this subsequence, gives
\eqns{\limsup_n\max_{s_t\in\fkh{0,1}}U_t\kh{s_ty\mid B^n}&=\limsup_nU_t\kh{s^ny\mid B^n}\\&\le U_t\kh{sy\mid B}\le\max_{s_t\in\fkh{0,1}}U_t\kh{s_ty\mid B},}
which proves that the maximized value is upper semicontinuous in $B$.
\end{proof}

\begin{proof}[Proof of Theorem~\ref{prop:1grow}]
The function $V_{T+1}^*\equiv0$ satisfies properties~\ref{item:advapp-bound}-\ref{item:advapp-pairs}. Suppose that $V_{t+1}^*$ satisfies them for some $t\le T$. The proof of Lemma~\ref{lem:affine} shows that, for every known set $B$, every contract is weakly dominated by a linear contract with share in $\fkh{0,1}$. Hence
$$V_t^*\kh{B}=\sup_{s_t\in\fkh{0,1}}U_t\kh{s_ty\mid B},$$
and the proof of Lemma~\ref{lem:optlinear} shows that this supremum is attained and that $V_t^*$ satisfies property~\ref{item:advapp-usc}. It remains to verify properties~\ref{item:advapp-bound}, \ref{item:advapp-mono}, and~\ref{item:advapp-pairs} at date $t$.

For property~\ref{item:advapp-bound}, the proof of Lemma~\ref{lem:affine} shows that the zero contract guarantees at least $\beta V_{t+1}^*\kh{B}\ge0$. Since payments are nonnegative, the principal's current payoff from any action is at most $\max Y$, so the upper bound on $V_{t+1}^*$ gives
$$0\le V_t^*\kh{B}\le\max Y+\beta\kh{\max Y}\sum_{j=0}^{T-t-1}\beta^j=\kh{\max Y}\sum_{j=0}^{T-t}\beta^j.$$

For property~\ref{item:advapp-mono}, let $B\subseteq B'$ and fix a contract $w_t$. For every technology $A\supseteq B'$ and every $a_t\in BR\kh{w_t\mid A}$, property~\ref{item:advapp-mono} of $V_{t+1}^*$ gives $U_t\kh{w_t\mid B',a_t}\ge U_t\kh{w_t\mid B,a_t}$. Since every technology containing $B'$ also contains $B$,
\eqns{U_t\kh{w_t\mid B'}&=\inf_{A\supseteq B'}\max_{a_t\in BR\kh{w_t\mid A}}U_t\kh{w_t\mid B',a_t}\\&\ge\inf_{A\supseteq B'}\max_{a_t\in BR\kh{w_t\mid A}}U_t\kh{w_t\mid B,a_t}\ge U_t\kh{w_t\mid B}.}
Taking the supremum over contracts gives $V_t^*\kh{B'}\ge V_t^*\kh{B}$.

For property~\ref{item:advapp-pairs}, suppose that $\mathcal M\kh{B}=\mathcal M\kh{B'}$, and fix $s_t\in\fkh{0,1}$. Under the contract $s_ty$, agent $t$'s payoff and the principal's current payoff from an action depend only on its expected output and cost. Given a technology $A\supseteq B$, let $A'=A\cup B'$. Then $A'$ is a technology containing $B'$, and $\mathcal M\kh{A'}=\mathcal M\kh{A}$ because $\mathcal M\kh{B'}=\mathcal M\kh{B}\subseteq\mathcal M\kh{A}$. For every $a'\in BR\kh{s_ty\mid A'}$, choose $a\in A$ with the same expected output and cost. Then $a\in BR\kh{s_ty\mid A}$, the two actions give the same current payoff, and $\mathcal M\kh{B'\cup\hkh{a'}}=\mathcal M\kh{B\cup\hkh{a}}$. Property~\ref{item:advapp-pairs} of $V_{t+1}^*$ therefore gives $U_t\kh{s_ty\mid B',a'}=U_t\kh{s_ty\mid B,a}$, so
$$U_t\kh{s_ty\mid B'}\le\max_{a'\in BR\kh{s_ty\mid A'}}U_t\kh{s_ty\mid B',a'}\le\max_{a\in BR\kh{s_ty\mid A}}U_t\kh{s_ty\mid B,a}.$$
Taking the infimum over technologies $A\supseteq B$ gives $U_t\kh{s_ty\mid B'}\le U_t\kh{s_ty\mid B}$, and exchanging the roles of $B$ and $B'$ gives equality. Maximizing over $s_t$ yields $V_t^*\kh{B'}=V_t^*\kh{B}$.

Properties~\ref{item:advapp-bound}-\ref{item:advapp-pairs} therefore hold at date $t$, and backward induction establishes them, together with Lemmas~\ref{lem:affine} and~\ref{lem:optlinear}, at every date. For every feasible history $h^t$,
$$V_t^*\kh{h^t}=V_t^*\kh{K_t}=\max_{s_t\in\fkh{0,1}}U_t\kh{s_ty\mid K_t}=\max_{s_t\in\fkh{0,1}}U_t\kh{s_ty\mid h^t},$$
which is equation~\eqref{eqn:adv-linear-optimal}. Offering a maximizing linear contract after every feasible history therefore gives an optimal policy. Because the guarantee from each linear contract depends on the known set only through its expected output and cost pairs, the maximizing share can be chosen to depend on $h^t$ only through $t$ and $\mathcal M\kh{K_t}$.

This completes the proof.
\end{proof}

\subsection{Proofs for Section~\ref{sec:const}}\label{app:const}
Under constant technology, observing the agent's choice $a_t$ under $w_t$ tells the principal two things: $a_t$ is available, and no available action gives the agent a higher payoff under $w_t$. If $a_t$ is a known action and no compatible technology contains an action that the agent strictly prefers to $a_t$ under $w_t$, the principal already knows both, and the family of compatible technologies does not change. The simplest such period is a \emph{silent period}: the principal offers the zero contract, and the null action is recorded as a best response. Appending a silent period to a feasible history $h^t$ thus gives a feasible history $\kh{h^t,0,\kh{\delta_0,0}}$ with the same known set and the same family of compatible technologies. For $t'\ge t$, write $h^{t\to t'}$ for the history at date $t'$ obtained from $h^t$ by appending $t'-t$ silent periods, so that $h^{t\to t+1}=\kh{h^t,0,\kh{\delta_0,0}}$. The proofs use silent periods to describe what follows a period in which the principal learns nothing.

Fix a date $t$ and a feasible history $h^t$, and write $u_\tau=\mathbb{E}_{F_\tau}\fkh{w_\tau\kh{y}}-c_\tau$ for the payoff revealed by the choice of agent $\tau<t$. For an action $a=\kh{F,c}$, write
$$x\kh{a}\equiv\mathbb{E}_{F}\fkh{y},\quad c\kh{a}\equiv c,\quad q_\tau\kh{a}\equiv\mathbb{E}_{F}\fkh{w_\tau\kh{y}},\quad q\kh{a}\equiv\kh{q_1\kh{a},\ldots,q_{t-1}\kh{a}}.$$
The action $a$ is \emph{compatible} with $h^t$ if $q_\tau\kh{a}-c\le u_\tau$ for every $\tau<t$, that is, if it meets the second requirement of Definition~\ref{def:comp}. Since $K_t$ is compact, it is itself a technology, so $h^t$ is feasible exactly when every known action is compatible with it. Because the null action is known, its compatibility gives $u_\tau\ge0$ for every $\tau<t$.

The coefficient set $\Theta_t$ and the contracts $w_{\theta_t}$ with $\theta_t\in\Theta_t$, which form $\operatorname{co}\hkh{0,y,w_1,\ldots,w_{t-1}}$, are as defined in Section~\ref{sec:period1}.

As in Section~\ref{app:adv}, fix a metric on the set of actions that induces its product topology, and equip known sets and technologies with the induced Hausdorff metric, so that, if technologies $A^n$ converge to $A$, then $A\cup\bigcup_nA^n$ is compact and every limit point of a sequence with $a^n\in A^n$ belongs to $A$. Equip previous contracts with the uniform topology, and histories at date $t$ with the product topology on their contracts and observed actions; $A_0$ is fixed. Then $K\kh{h^t}$ and every $u_\tau$ are continuous in $h^t$, and the feasible histories at date $t$ form a closed set, because the weak inequalities that define compatibility are preserved under limits.

The proof of Theorem~\ref{prop:1} summarizes the information at date $t$ by a \emph{state}, which records the expected output, the expected payments under previous contracts and the cost of each initially known action and each previous observation, together with the revealed payoffs. It extends the continuation value to a convex set of states, including states that no history produces, and shows that the extension is convex. Formally, a state $S$ at date $t$ has one coordinate for each action in $A_0$ and one for each previous period $\tau<t$, and records revealed payoffs $u_1,\ldots,u_{t-1}\ge0$. The value of each coordinate is a triple $\kh{x,q,c}$ of an expected output $x\ge0$, expected payments $q=\kh{q_1,\ldots,q_{t-1}}\ge0$ under the previous contracts and a cost $c\ge0$, and the values of the coordinates of $A_0$ are continuous in the action. Coordinates with equal values remain distinct; write $\mathcal M\kh{S}$ for the set of their values, which is compact as the union of a continuous image of the compact set $A_0$ and finitely many points. The state is \emph{compatible} if every triple $\kh{x,q,c}\in\mathcal M\kh{S}$ satisfies $q_\tau-c\le u_\tau$ for every $\tau<t$. The state of a feasible history $h^t$ assigns the triple $\kh{x\kh{a},q\kh{a},c\kh{a}}$ to the coordinate of each $a\in A_0$ and the triple $\kh{x\kh{a_\tau},q\kh{a_\tau},c\kh{a_\tau}}$ to the coordinate of each period $\tau<t$, and records the revealed payoffs of $h^t$.

In a general state, the values need not be generated by distributions on $Y$, and the value of the coordinate of an action $a\in A_0$ need not be $\kh{x\kh{a},q\kh{a},c\kh{a}}$. For fixed $t$, convex combinations of states are taken entry by entry, pointwise in the actions of $A_0$. With the uniform norm over the coordinates of $A_0$ and the Euclidean norm on the remaining entries, the compatible states form a closed convex set.

Write $\mathcal V_t$ for the extension of the optimal overall payoff guarantee to states at date $t$, with $\mathcal V_{T+1}\equiv0$; property~\ref{item:constapp-exact} below shows that $\mathcal V_t\kh{S}=\hat V_t^*\kh{h^t}$ at the state $S$ of a feasible history $h^t$. At the final date, the extension comes from Lemma~\ref{lem:second}. For a principal's payoff $r\in\mathbb{R}$ and an incentive gap $g\ge0$, write $\xi\kh{r,g}\equiv\kh{\sqrt{r_+}-\sqrt g}_+^2$, and, for a state $S$ at any date $t\le T$, let
\eqn{\mathcal V_T\kh{S}=\max\hkh{\max_{\kh{x,q,c}\in\mathcal M\kh{S}}\xi\kh{x,c},\max_{\substack{\tau<t\\\kh{x,q,c}\in\mathcal M\kh{S}}}\xi\kh{x-q_\tau,u_\tau-q_\tau+c}},\label{eqn:final-state}}
with the last inner maximum omitted if $t=1$. In the proof of Lemma~\ref{lem:second}, $\xi\kh{x,c}$ is the payoff guaranteed by a linear contract, and $\xi\kh{x-q_\tau,u_\tau-q_\tau+c}$ is the payoff guaranteed by the contract in equation~\eqref{eqn:m} for $w_\tau$, where $x-q_\tau$ is the principal's payoff under $w_\tau$ and $u_\tau-q_\tau+c$ is the incentive gap; that proof shows that $\mathcal V_T\kh{S}=\hat V_T^*\kh{h^T}$ at the state $S$ of a feasible history $h^T$. Appending zero payments, coordinates with zero values and zero revealed payoffs, as silent periods do, does not change $\mathcal V_T$, because each appended term is zero or equals a term $\xi\kh{x,c}$; at a state of an earlier date, $\mathcal V_T$ is thus the final-period value after silent periods.

For $\theta_t=\kh{\alpha,\eta_1,\ldots,\eta_{t-1}}\in\Theta_t$ and a moment vector $\kh{x,q}\in\mathbb{R}_+^t$, write
$$p_{\theta_t}\kh{x,q}=\alpha x+\sum_{\tau<t}\eta_\tau q_\tau,$$
and let the payoff floor be
$$\underline u_{\theta_t}\kh{S}=\max\hkh{0,\max_{\kh{x,q,c}\in\mathcal M\kh{S}}\kh{p_{\theta_t}\kh{x,q}-c}}.$$
At the state of a feasible history, $p_{\theta_t}\kh{x\kh{a},q\kh{a}}$ is the expected payment of $w_{\theta_t}$ to $a$, and $\underline u_{\theta_t}\kh{S}$ is the highest payoff that a known action obtains under $w_{\theta_t}$. A \emph{relaxed response} to $\theta_t$ is a triple $\kh{x',q',c'}\in\mathbb{R}_+\times\mathbb{R}_+^{t-1}\times\mathbb{R}_+$ such that
$$q_\tau'-c'\le u_\tau\quad\text{for every }\tau<t,\quad u_{\theta_t}\equiv p_{\theta_t}\kh{x',q'}-c'\ge\underline u_{\theta_t}\kh{S}.$$
A relaxed response is less demanding than a \emph{strict new response} (Section~\ref{sec:envelope}) in two ways: its moments $\kh{x',q'}$ need not come from a distribution on $Y$, and the agent need only weakly prefer it to the known actions. Given a relaxed response $\kh{x',q',c'}$, let $S_{\theta_t}^+\kh{x',q',c'}$ be the state at date $t+1$ that records it as the observation of period $t$: each existing coordinate also records its payment under $w_{\theta_t}$, which is $p_{\theta_t}\kh{x,q}$ for a coordinate with value $\kh{x,q,c}$; the new coordinate for period $t$ has value $\kh{x',\kh{q',p_{\theta_t}\kh{x',q'}},c'}$; and $u_{\theta_t}$ joins the revealed payoffs. This updated state is compatible, and, for fixed $\theta_t$, it is jointly affine in the old state and the response. For $\theta_t=0$, the zero triple is a relaxed response, and $S^0=S_{\theta_t}^+\kh{0,0,0}$ describes a silent period: every existing coordinate records a zero payment, the new coordinate has value zero, and the revealed payoff is zero. In particular, $S^0$ is the state of $h^{t\to t+1}$ when $S$ is the state of $h^t$. Let $\Theta_t^<$ be the set of $\theta_t\in\Theta_t$ whose weights sum to less than one, $\alpha+\sum_{\tau<t}\eta_\tau<1$, and, for $t<T$ and $\theta_t\in\Theta_t^<$, define
\eqn{\mathcal U_t\kh{\theta_t\mid S}&=\inf_{\kh{x',q',c'}}\hkh{x'-p_{\theta_t}\kh{x',q'}+\beta\mathcal V_{t+1}\kh{S_{\theta_t}^+\kh{x',q',c'}}},\label{eqn:relaxed-value}\\\mathcal V_t\kh{S}&=\max\hkh{\mathcal V_T\kh{S}+\beta\mathcal V_{t+1}\kh{S^0},\sup_{\theta_t\in\Theta_t^<}\mathcal U_t\kh{\theta_t\mid S}},\label{eqn:state-recursion}}
where the infimum is over relaxed responses to $\theta_t$. The value $\mathcal U_t\kh{\theta_t\mid S}$ is what the principal can guarantee with $w_{\theta_t}$ when nature may choose any relaxed response: a response $\kh{x',q',c'}$ gives her the current payoff $x'-p_{\theta_t}\kh{x',q'}$, and she continues from the updated state. The value $\mathcal V_t\kh{S}$ is the larger of two branches. The first covers contracts under which the agent can do no better than a known action; such a contract teaches the principal nothing, so it gives her at most the final-period value now plus the continuation value after a silent period, and an optimal final-period contract guarantees this sum. The second covers contracts that admit a strict new response, so that nature can make the agent choose a new action, which adds an observation and an incentive restriction to the state. This branch is restricted to $\Theta_t^<$. When the weights sum to less than one, a relaxed response with a large expected output gives the principal a large current payoff, so the infimum in equation~\eqref{eqn:relaxed-value} is finite and attained (property~\ref{item:constapp-convex} below). The restriction loses nothing: for every contract that admits a strict new response, the proof of Lemma~\ref{lem:improvementconst} finds $\theta_t\in\Theta_t^<$ whose value $\mathcal U_t\kh{\theta_t\mid S}$ falls short of that contract's guarantee by an arbitrarily small amount.
 
We prove Theorem~\ref{prop:1} by two backward inductions. The first concerns only states: it shows that, at every date $t$, $\mathcal V_t$ has the five properties below on states at date $t$, and that $\mathcal V_T$ has them on states at every date. In these properties, $X\kh{S}\equiv\max_{\kh{x,q,c}\in\mathcal M\kh{S}}x$ is the largest expected output in $S$, and each change in properties~\ref{item:constapp-outputs}-\ref{item:constapp-columns} must produce a compatible state; in particular, the values of the coordinates of $A_0$ must stay continuous in the action.
\begin{enumerate}
\item\label{item:constapp-convex} $\mathcal V_t$ is finite, convex and lower semicontinuous on the compatible states, and satisfies $0\le\mathcal V_t\kh{S}\le X\kh{S}\sum_{j=0}^{T-t}\beta^j$, with the empty sum equal to zero when $t=T+1$. If $t<T$ and $\theta_t\in\Theta_t^<$, the infimum in equation~\eqref{eqn:relaxed-value} is finite and attained.
\item\label{item:constapp-outputs} $\mathcal V_t$ is weakly increasing in the expected outputs of the coordinates, one for each action in $A_0$ and each previous period.
\item\label{item:constapp-payments} $\mathcal V_t$ is weakly increasing in the payments of the coordinates.
\item\label{item:constapp-costs} $\mathcal V_t$ weakly decreases if, for some coordinates, the cost $c$ rises by $d\ge0$ while the expected output $x$ and each payment $q_\tau$ rise by at most $d$, that is,
$$c\mapsto c+d,\quad x\mapsto x+\gamma_0d,\quad q_\tau\mapsto q_\tau+\gamma_\tau d,$$
with $\gamma_0,\ldots,\gamma_{t-1}\in\fkh{0,1}$, where $d$ and $\gamma_0,\ldots,\gamma_{t-1}$ may differ across coordinates.
\item\label{item:constapp-columns} Fix $\kappa=\kh{\kappa_1,\ldots,\kappa_{t-1}}\ge0$ and a set of coordinates, and obtain $S^-$ from $S$ by lowering each revealed payoff $u_\tau$, and the payment $q_\tau$ of every coordinate in the set, by $\kappa_\tau$. Then $\mathcal V_t\kh{S^-}\ge\mathcal V_t\kh{S}$.
\end{enumerate}
The second shows that, at every date $t$, the function $\hat V_t^*$ has the following properties.
\begin{enumerate}\setcounter{enumi}{5}
\item\label{item:constapp-bound} $\hat V_t^*$ satisfies $0\le\hat V_t^*\kh{h^t}\le\kh{\max Y}\sum_{j=0}^{T-t}\beta^j$ for every feasible history $h^t$, with the empty sum equal to zero when $t=T+1$.
\item\label{item:constapp-usc} $\hat V_t^*$ is upper semicontinuous on the feasible histories at date $t$.
\item\label{item:constapp-info} $\hat V_t^*$ weakly increases as the family of compatible technologies shrinks: $\hat V_t^*\kh{\tilde h^t}\ge\hat V_t^*\kh{h^t}$ whenever $h^t$ and $\tilde h^t$ are feasible histories with $\Omega\kh{\tilde h^t}\subseteq\Omega\kh{h^t}$, with equality when the two families coincide.
\item\label{item:constapp-exact} $\hat V_t^*\kh{h^t}=\mathcal V_t\kh{S}$ for every feasible history $h^t$ with state $S$.
\end{enumerate}

If $\hat V_{t+1}^*$ satisfies properties~\ref{item:constapp-bound} and~\ref{item:constapp-usc}, then, for every feasible history $h^t$, contract $w_t$ and technology $A_t\in\Omega\kh{h^t}$, the map $a_t\mapsto\hat U_t\kh{w_t\mid h^t,a_t}$ is upper semicontinuous on $BR\kh{w_t\mid A_t}$, because $w_t$ is continuous and $A_t$ makes every history $\kh{h^t,w_t,a_t}$ with $a_t\in BR\kh{w_t\mid A_t}$ feasible. Since $BR\kh{w_t\mid A_t}$ is nonempty and compact, the maximum in the definition of $\hat U_t\kh{w_t\mid h^t}$ is attained.
 
For a feasible history $h^t$ and a contract $w_t$, write $p\kh{a}=\mathbb{E}_{F}\fkh{w_t\kh{y}}$ for the expected payment to an action $a=\kh{F,c}$, and $\underline u_{w_t}=\max_{a=\kh{F,c}\in K_t}\kh{p\kh{a}-c}\ge0$ for the highest payoff that a known action obtains under $w_t$, so that $\underline u_{\theta_t}\kh{S}=\underline u_{w_{\theta_t}}$ for every $\theta_t\in\Theta_t$ when $S$ is the state of $h^t$. A moment vector $\kh{x,q}$ is \emph{attainable} if it is generated by some $F\in\Delta\kh{Y}$, that is, $\mathbb{E}_F\fkh{y}=x$ and $\mathbb{E}_F\fkh{w_\tau\kh{y}}=q_\tau$ for every $\tau<t$. The attainable moment vectors form the compact convex set $\operatorname{co}\hkh{\kh{y,w_1\kh{y},\ldots,w_{t-1}\kh{y}}:y\in Y}$, which contains the zero moment vector because $0\in Y$ and every contract pays nothing at zero output. The upper payment envelope $\overline p_{w_t}$ of Section~\ref{sec:envelope} is defined at every attainable $\kh{x,q}$, and the maximum in its definition is attained because the distributions with moments $\kh{x,q}$ form a nonempty compact subset of $\Delta\kh{Y}$. On the attainable set, $\overline p_{w_t}$ is concave and upper semicontinuous; it is zero at the zero moment vector because $\delta_0$ is the only distribution with mean zero; and $p\kh{a}\le\overline p_{w_t}\kh{x\kh{a},q\kh{a}}$ for every action $a$. An action with expected payments $q$ under the previous contracts is compatible with $h^t$ exactly when its cost is at least the cost floor
$$\underline c\kh{q}=\max\hkh{0,q_1-u_1,\ldots,q_{t-1}-u_{t-1}}.$$
The contract $w_t$ \emph{admits a strict new response} at $h^t$ if some action compatible with $h^t$ obtains a payoff strictly above $\underline u_{w_t}$ under $w_t$. Since an action with moments $\kh{x,q}$ receives at most $\overline p_{w_t}\kh{x,q}$ and costs at least $\underline c\kh{q}$, and both bounds can be met, this happens exactly when
\eqn{\overline p_{w_t}\kh{x,q}-\underline c\kh{q}>\underline u_{w_t}\quad\text{for some attainable }\kh{x,q}.\label{eqn:strict-response-condition}}
 
Under $w_t$, an observed action matters to the principal only through its moments $\kh{x,q}$, its current payment $p$ and the agent's payoff $u$, which together determine its cost $p-u$. For such $\kh{x,q,p,u}$, let $S_{w_t}^+\kh{x,q,p,u}$ be the state at date $t+1$ that records it as the observation of period $t$: starting from the state of $h^t$, each existing coordinate also records the payment that $w_t$ gives its action, which is $p\kh{a}$ for the coordinate of $a\in A_0$ and $p\kh{a_\tau}$ for the coordinate of period $\tau<t$; the new coordinate for period $t$ has value $\kh{x,\kh{q,p},p-u}$; and $u$ joins the revealed payoffs. This state is compatible exactly when $x\ge0$, $q\ge0$, $u\ge\underline u_{w_t}$ and $p-u\ge\underline c\kh{q}$; call such $\kh{x,q,p,u}$ a \emph{candidate observation}, and write
$$\mathcal V_{w_t}^+\kh{x,q,p,u}\equiv\mathcal V_{t+1}\kh{S_{w_t}^+\kh{x,q,p,u}}$$
for its continuation value. By property~\ref{item:constapp-convex}, $\mathcal V_{w_t}^+$ is finite, convex and lower semicontinuous on the candidate observations. By properties~\ref{item:constapp-outputs} and~\ref{item:constapp-payments}, it is weakly increasing in $x$ and in each $q_\tau$ on the candidate observations. By property~\ref{item:constapp-costs}, it is weakly decreasing in $p$, and it weakly decreases when $x$, every $q_\tau$ and $p$ rise by the same amount with $u$ fixed, because the cost $p-u$ then rises by the same amount. If an action $a_t=\kh{F_t,c_t}$ is compatible with $h^t$ and $p\kh{a_t}-c_t\ge\underline u_{w_t}$, then $a_t$ is a best response to $w_t$ in the compatible technology $K_t\cup\hkh{a_t}$, so the history $\kh{h^t,w_t,a_t}$ is feasible and its state is $S_{w_t}^+\kh{x\kh{a_t},q\kh{a_t},p\kh{a_t},p\kh{a_t}-c_t}$. Hence, if $\hat V_{t+1}^*$ satisfies property~\ref{item:constapp-exact},
\eqn{\hat U_t\kh{w_t\mid h^t,a_t}=x\kh{a_t}-p\kh{a_t}+\beta\mathcal V_{w_t}^+\kh{x\kh{a_t},q\kh{a_t},p\kh{a_t},p\kh{a_t}-c_t}.\label{eqn:physical-interim}}
 
The proofs below follow the order of Section~\ref{sec:const}. The proof of Lemma~\ref{lem:second} gives the final-period value; the proofs of Lemmas~\ref{lem:envelopeconst}-\ref{lem:existenceconst} and Corollary~\ref{cor:history-relative-concavification} carry out the three steps of Section~\ref{sec:period1} at a date $t$, taking the properties at date $t+1$ as given; and the proof of Theorem~\ref{prop:1} completes both inductions.

\begin{proof}[Proof of Lemma~\ref{lem:second}]
Fix a feasible history $h^T$ with state $S$. For an action $a=\kh{F,c}$ and $\tau<T$, let $r_\tau\kh{a}=x\kh{a}-q_\tau\kh{a}$ be the principal's payoff from $a$ under $w_\tau$, and let $g_\tau\kh{a}=g_\tau\kh{F,c}=u_\tau-q_\tau\kh{a}+c$ be the incentive gap of Definition~\ref{def:g}, which is nonnegative when $a$ is compatible with $h^T$. The first inner maximum in equation~\eqref{eqn:optsecond} is dominated by the second, because $g_\tau\kh{a_\tau}=0$. Since $\hat\Phi\kh{h^T}\ge\Phi\kh{K_T}\ge\Phi_0>0$, squaring gives
$$\hat\Phi\kh{h^T}^2=\max\hkh{\max_{a\in K_T}\xi\kh{x\kh{a},c\kh{a}},\max_{\substack{\tau<T\\a\in K_T}}\xi\kh{r_\tau\kh{a},g_\tau\kh{a}}}.$$
The values of the coordinates of $S$ are the triples of the known actions, so equation~\eqref{eqn:final-state} gives the same maximum, $\hat\Phi\kh{h^T}^2=\mathcal V_T\kh{S}$. Since $K_T$ is compact and every term is continuous in $a$, some term attains the maximum, and that term is positive because $\hat\Phi\kh{h^T}>0$.

We first show that the contracts in Parts~\ref{part:1}, \ref{part:2} and~\ref{part:3} guarantee the corresponding terms. We call the terms $\xi\kh{x\kh{a},c\kh{a}}$ static and the terms $\xi\kh{r_\tau\kh{a},g_\tau\kh{a}}$ historical. Consider a positive static term $\xi\kh{x\kh{a},c\kh{a}}$, where $a\in K_T$ and $c\kh{a}>0$, so that $x\kh{a}>c\kh{a}$, and offer the linear contract $w_T\kh{y}=s_Ty$ with $s_T=\sqrt{c\kh{a}/x\kh{a}}\in\kh{0,1}$. Under every compatible technology, every best response $b$ satisfies $s_Tx\kh{b}-c\kh{b}\ge s_Tx\kh{a}-c\kh{a}$, so $s_Tx\kh{b}\ge s_Tx\kh{a}-c\kh{a}$ because $c\kh{b}\ge0$, and the principal's payoff is
$$\kh{1-s_T}x\kh{b}\ge\kh{1-s_T}\kh{x\kh{a}-\frac{c\kh{a}}{s_T}}=\xi\kh{x\kh{a},c\kh{a}}.$$
If $c\kh{a}=0$, then $s_T=0$, and the zero contract guarantees $x\kh{a}=\xi\kh{x\kh{a},0}$, since $a$ is a best response and ties are broken in the principal's favor.

Next consider a positive historical term $\xi\kh{r_\tau\kh{a},g_\tau\kh{a}}$ with $g_\tau\kh{a}>0$, so that $r_\tau\kh{a}>g_\tau\kh{a}$, and offer the contract in equation~\eqref{eqn:m} for $w_\tau$ and $a$, with share $m=\sqrt{g_\tau\kh{a}/r_\tau\kh{a}}\in\kh{0,1}$. Because $w_T=w_\tau+m\kh{y-w_\tau}$, this contract pays an action $b$ the expected amount $q_\tau\kh{b}+mr_\tau\kh{b}$. Fix a compatible technology and a best response $b$ in it. First, the known action $a$ belongs to the technology, so $b$ gives the agent at least as much as $a$:
$$q_\tau\kh{b}+mr_\tau\kh{b}-c\kh{b}\ge q_\tau\kh{a}+mr_\tau\kh{a}-c\kh{a}.$$
Second, $b$ is compatible with $h^T$, so under $w_\tau$ it would not have given agent $\tau$ more than his payoff $u_\tau$:
$$q_\tau\kh{b}-c\kh{b}\le u_\tau.$$
Subtracting the second inequality from the first gives
$$mr_\tau\kh{b}\ge mr_\tau\kh{a}+q_\tau\kh{a}-c\kh{a}-u_\tau=mr_\tau\kh{a}-g_\tau\kh{a},$$
where the equality is the definition of $g_\tau\kh{a}$. Dividing by $m>0$ gives $r_\tau\kh{b}\ge r_\tau\kh{a}-g_\tau\kh{a}/m$. Since $y-w_T\kh{y}=\kh{1-m}\kh{y-w_\tau\kh{y}}$, the principal's payoff from $b$ under $w_T$ is
$$\kh{1-m}r_\tau\kh{b}\ge\kh{1-m}\kh{r_\tau\kh{a}-\frac{g_\tau\kh{a}}{m}}=\kh{\sqrt{r_\tau\kh{a}}-\sqrt{g_\tau\kh{a}}}^2=\xi\kh{r_\tau\kh{a},g_\tau\kh{a}}.$$

If $g_\tau\kh{a}=0$, then $m=0$ and $w_T=w_\tau$. By compatibility, every action of a compatible technology gives the agent at most $u_\tau$ under $w_\tau$, and $a$ gives exactly $u_\tau$ because $g_\tau\kh{a}=0$. So $a$ is a best response, and favorable tie-breaking guarantees the principal at least $r_\tau\kh{a}=\xi\kh{r_\tau\kh{a},0}$. Part~\ref{part:1} is the special case $a=a_\tau$, since $g_\tau\kh{a_\tau}=0$. Thus every positive term is guaranteed by its own contract: a static term by a linear contract, and a historical term by the contract in equation~\eqref{eqn:m} for $w_\tau$, which is $w_\tau$ itself when the incentive gap is zero. Some positive term attains the maximum $\hat\Phi\kh{h^T}^2$, and its contract guarantees the principal at least this maximum, so $\hat V_T^*\kh{h^T}\ge\hat\Phi\kh{h^T}^2$.

Conversely, we show that no final-period contract guarantees more than $\hat\Phi\kh{h^T}^2$. Fix a contract $w_T$. Among the best responses to $w_T$ in $K_T$, let $a=\kh{F,c}$ be one that gives the principal her highest payoff, and write $p=p\kh{a}$ for its expected payment and $r=x\kh{a}-p$ for the principal's payoff from it, so that $\underline u_{w_T}=p-c$. If $r\le\hat\Phi\kh{h^T}^2$, nature can choose the technology $K_T$, which belongs to $\Omega\kh{h^T}$; the agent then chooses among the known actions, and favorable tie-breaking gives the principal exactly $r$. 

Suppose instead that $r>\hat\Phi\kh{h^T}^2$. Then $c>0$. Otherwise the static term of $a$ would be $\xi\kh{x\kh{a},0}=x\kh{a}\ge r$, so $\hat\Phi\kh{h^T}^2\ge r$. Since $\underline u_{w_T}=p-c\ge0$, it follows that $p\ge c>0$. Moreover, $q_\tau\kh{a}>p$ for every $\tau<T$ with $g_\tau\kh{a}=0$. Otherwise $r_\tau\kh{a}=x\kh{a}-q_\tau\kh{a}\ge x\kh{a}-p=r$, and the historical term $\xi\kh{r_\tau\kh{a},0}=r_\tau\kh{a}$ would again give $\hat\Phi\kh{h^T}^2\ge r$.

We first construct a strict new response by moving a little probability from $F$ to zero output and lowering the cost by slightly more than the payment this saves. Choose $\sigma>0$ smaller than $q_\tau\kh{a}-p$ for every $\tau<T$ with $g_\tau\kh{a}=0$, which is possible by the previous paragraph. For $\epsilon>0$, consider the action
$$\tilde a=\kh{\kh{1-\epsilon}F+\epsilon\delta_0,c-\epsilon\kh{p+\sigma}}.$$
Its cost is nonnegative for small $\epsilon$, because $c>0$. Since contracts pay nothing at zero output, $\tilde a$ receives the expected payment $\kh{1-\epsilon}p$ under $w_T$ and $\kh{1-\epsilon}q_\tau\kh{a}$ under $w_\tau$. Its payoff under $w_T$ is therefore $\kh{1-\epsilon}p-c+\epsilon\kh{p+\sigma}=\underline u_{w_T}+\epsilon\sigma$, which exceeds $\underline u_{w_T}$, and its payoff under $w_\tau$ is $\kh{1-\epsilon}q_\tau\kh{a}-c+\epsilon\kh{p+\sigma}=q_\tau\kh{a}-c+\epsilon\kh{p+\sigma-q_\tau\kh{a}}$. Compatibility with $h^T$ requires this payoff to be at most $u_\tau$ for every $\tau<T$. If $g_\tau\kh{a}=0$, then $q_\tau\kh{a}-c=u_\tau$, and the $\epsilon$-term is negative because $p+\sigma<q_\tau\kh{a}$ by the choice of $\sigma$, so the payoff is below $u_\tau$. If $g_\tau\kh{a}>0$, then $q_\tau\kh{a}-c<u_\tau$, so the payoff is at most $u_\tau$ for small $\epsilon$. Hence, for small $\epsilon$, $\tilde a$ is compatible with $h^T$, and since its payoff under $w_T$ exceeds $\underline u_{w_T}$, it is a strict new response.

Next, we dilute $a$. For $\lambda\in\fkh{0,1}$, the diluted action
$$a\kh{\lambda}=\kh{\lambda F+\kh{1-\lambda}\delta_0,\lambda p-\underline u_{w_T}}$$
draws output from $F$ with probability $\lambda$ and yields zero output otherwise. Under $w_T$, it receives $\lambda p$, so its cost leaves the agent exactly $\underline u_{w_T}$, and it gives the principal $\lambda\kh{x\kh{a}-p}=\lambda r$. Its cost is nonnegative exactly when $\lambda p\ge\underline u_{w_T}$. Under $w_\tau$, it receives $\lambda q_\tau\kh{a}$, so it is compatible with $h^T$ exactly when $\lambda q_\tau\kh{a}-\kh{\lambda p-\underline u_{w_T}}\le u_\tau$ for every $\tau<T$. Together, these conditions are
$$\lambda p\ge\underline u_{w_T},\quad\lambda\kh{p-q_\tau\kh{a}}\ge\underline u_{w_T}-u_\tau\quad\text{for every }\tau<T.$$
Since $a\kh{1}=a$, these inequalities hold at $\lambda=1$. As $\lambda$ decreases, an inequality with $p\le q_\tau\kh{a}$ continues to hold, because its left-hand side does not decrease. The first inequality is equivalent to $\lambda\ge\underline u_{w_T}/p$, and an inequality with $p>q_\tau\kh{a}$ is equivalent to $\lambda\ge\kh{\underline u_{w_T}-u_\tau}/\kh{p-q_\tau\kh{a}}$. The feasible values of $\lambda$ therefore form the interval $\fkh{\lambda^*,1}$, where
$$\lambda^*=\max\hkh{\frac{\underline u_{w_T}}{p},\max_{\tau<T:\,p>q_\tau\kh{a}}\frac{\underline u_{w_T}-u_\tau}{p-q_\tau\kh{a}}}\in\fkh{0,1},$$
with the inner maximum omitted if no $\tau<T$ has $p>q_\tau\kh{a}$.

Suppose first that $\lambda^*=\underline u_{w_T}/p$. Since $\underline u_{w_T}=p-c$, this means $\lambda^*=1-c/p$, and expanding both sides shows that
$$\lambda^*r=\kh{1-\frac{c}{p}}\kh{x\kh{a}-p}=\kh{\sqrt{x\kh{a}}-\sqrt c}^2-\frac{\kh{p-\sqrt{cx\kh{a}}}^2}{p}.$$
As a function of $p$, the principal's payoff from $a\kh{\lambda^*}$ is thus largest at $p=\sqrt{cx\kh{a}}$. Since $r>0$ and $p\ge c$, we have $x\kh{a}>c$, so $\kh{\sqrt{x\kh{a}}-\sqrt c}^2=\xi\kh{x\kh{a},c}$, and hence $\lambda^*r\le\xi\kh{x\kh{a},c}\le\hat\Phi\kh{h^T}^2$, where the last inequality holds because the static term of $a$ is one of the terms in $\hat\Phi\kh{h^T}^2$.

Otherwise, $\lambda^*$ is attained by some $\tau<T$ with $p>q_\tau\kh{a}$, and $\lambda^*>\underline u_{w_T}/p\ge0$. Since $\underline u_{w_T}=p-c$, we have $\underline u_{w_T}-u_\tau=p-q_\tau\kh{a}-g_\tau\kh{a}$, so
$$\lambda^*=1-\frac{g_\tau\kh{a}}{p-q_\tau\kh{a}},$$
and $\lambda^*>0$ gives $g_\tau\kh{a}<p-q_\tau\kh{a}$. Moreover, $r=x\kh{a}-p=r_\tau\kh{a}-p+q_\tau\kh{a}$, and $r>0$ gives $p-q_\tau\kh{a}<r_\tau\kh{a}$. Expanding both sides shows that
\eqns{\lambda^*r&=\kh{1-\frac{g_\tau\kh{a}}{p-q_\tau\kh{a}}}\kh{r_\tau\kh{a}-p+q_\tau\kh{a}}\\&=\kh{\sqrt{r_\tau\kh{a}}-\sqrt{g_\tau\kh{a}}}^2-\frac{\kh{p-q_\tau\kh{a}-\sqrt{g_\tau\kh{a}r_\tau\kh{a}}}^2}{p-q_\tau\kh{a}},}
so the payoff is largest when $p-q_\tau\kh{a}=\sqrt{g_\tau\kh{a}r_\tau\kh{a}}$. Since $r_\tau\kh{a}>g_\tau\kh{a}$, the first term on the right is $\xi\kh{r_\tau\kh{a},g_\tau\kh{a}}$, and hence $\lambda^*r\le\xi\kh{r_\tau\kh{a},g_\tau\kh{a}}\le\hat\Phi\kh{h^T}^2$, where the last inequality holds because the historical term of $a$ for $\tau$ is one of the terms in $\hat\Phi\kh{h^T}^2$.

The action $a\kh{\lambda^*}$ ties with the best known actions, so favorable tie-breaking could still select a known action. To break the tie, form mixtures of $a\kh{\lambda^*}$ and $\tilde a$ that put a small weight on $\tilde a$ in both the distribution and the cost, and let this weight vanish. Both actions have nonnegative cost and are compatible with $h^T$, and these properties are linear inequalities in the distribution and the cost, so every mixture has them too. Every mixture with positive weight on $\tilde a$ gives the agent a payoff strictly above $\underline u_{w_T}$ under $w_T$, because $a\kh{\lambda^*}$ gives exactly $\underline u_{w_T}$ and $\tilde a$ gives more. Adding such a mixture to $K_T$ therefore gives a compatible technology in which the mixture is the unique best response, and as the weight on $\tilde a$ vanishes, the principal's payoff from it converges to $\lambda^*r\le\hat\Phi\kh{h^T}^2$. Since nature can choose these technologies, $\hat U_T\kh{w_T\mid h^T}\le\hat\Phi\kh{h^T}^2$ in this case as well. This holds for every final-period contract $w_T$, so $\hat V_T^*\kh{h^T}\le\hat\Phi\kh{h^T}^2$. Together with the lower bound, $\hat V_T^*\kh{h^T}=\hat\Phi\kh{h^T}^2$, and the contracts in Parts~\ref{part:1}, \ref{part:2} and~\ref{part:3} attain it. 

This completes the proof.
\end{proof}

\begin{proof}[Proof of Lemma~\ref{lem:envelopeconst}]
Fix a date $t$ and a feasible history $h^t$, and suppose that $\mathcal V_{t+1}$ has properties~\ref{item:constapp-convex}-\ref{item:constapp-columns} and that $\hat V_{t+1}^*$ satisfies properties~\ref{item:constapp-info} and~\ref{item:constapp-exact}. We first express the overall payoff guarantee of a contract $w_t$ in each of the two cases described in Section~\ref{sec:envelope}, and then compare $w_t'$ with $w_t$. The proof of Lemma~\ref{lem:improvementconst} uses the same expressions.

Suppose first that $w_t$ admits a strict new response. We show that
\eqn{\hat U_t\kh{w_t\mid h^t}=\min_{\kh{x,q,p,u}}\hkh{x-p+\beta\mathcal V_{w_t}^+\kh{x,q,p,u}},\label{eqn:physical-program}}
where the minimum is over the candidate observations whose moment vector $\kh{x,q}$ is attainable and whose payment satisfies $p\le\overline p_{w_t}\kh{x,q}$, and that every minimizer of this program satisfies $p=\overline p_{w_t}\kh{x,q}$. By equation~\eqref{eqn:physical-interim}, the objective is the interim payoff guarantee from an action with moments $\kh{x,q}$, payment $p$ and payoff $u$, so the program describes nature's choice of the agent's response.

The feasible set of program~\eqref{eqn:physical-program} is compact: the attainable set is compact, every feasible point satisfies
$$0\le\underline u_{w_t}\le u\le u+\underline c\kh{q}\le p\le\overline p_{w_t}\kh{x,q}\le\max_{y\in Y}w_t\kh{y},$$
and every constraint is closed, because $\overline p_{w_t}$ is upper semicontinuous and $\underline c$ is continuous. The feasible set is also nonempty: a known action $a$ with $p\kh{a}-c\kh{a}=\underline u_{w_t}$ gives the feasible point $\kh{x\kh{a},q\kh{a},p\kh{a},\underline u_{w_t}}$, because $p\kh{a}\le\overline p_{w_t}\kh{x\kh{a},q\kh{a}}$ and compatibility gives $c\kh{a}\ge\underline c\kh{q\kh{a}}$. Since the objective is lower semicontinuous, the program has a minimizer. At a feasible point with $p<\overline p_{w_t}\kh{x,q}$, raising $p$ to $\overline p_{w_t}\kh{x,q}$ with $x$, $q$ and $u$ fixed keeps the point feasible, strictly lowers the principal's current payoff $x-p$ and weakly lowers $\mathcal V_{w_t}^+$. Hence every minimizer satisfies $p=\overline p_{w_t}\kh{x,q}$, and it is generated by an action that pairs a distribution attaining the maximum in the definition of $\overline p_{w_t}\kh{x,q}$ with the cost $p-u$. This action is compatible with $h^t$, because its cost is at least $\underline c\kh{q}$, and it gives the agent the payoff $u\ge\underline u_{w_t}$.

We now show that the minimum equals $\hat U_t\kh{w_t\mid h^t}$. Let $a_t=\kh{F_t,c_t}$ be a best response to $w_t$ in a technology $A_t\in\Omega\kh{h^t}$. Because $A_t$ contains $K_t$ and meets the second requirement of Definition~\ref{def:comp}, the action $a_t$ gives the agent at least $\underline u_{w_t}$ and is compatible with $h^t$. Since also $p\kh{a_t}\le\overline p_{w_t}\kh{x\kh{a_t},q\kh{a_t}}$, its point $\kh{x\kh{a_t},q\kh{a_t},p\kh{a_t},p\kh{a_t}-c_t}$ is feasible, and, by equation~\eqref{eqn:physical-interim}, the interim payoff guarantee from $a_t$ equals the objective at this point, which is at least the minimum. This holds for every best response in every compatible technology, so $\hat U_t\kh{w_t\mid h^t}$ is at least the minimum.

Conversely, let $a$ be an action that generates a minimizer, and let $\tilde a$ be a strict new response, that is, an action compatible with $h^t$ that gives the agent more than $\underline u_{w_t}$ under $w_t$. Form mixtures of $a$ and $\tilde a$ that put a small positive weight on $\tilde a$ in both the distribution and the cost. As in the proof of Lemma~\ref{lem:second}, every such mixture is compatible with $h^t$ and gives the agent a payoff strictly above $\underline u_{w_t}$, so it is the unique best response in the compatible technology obtained by adding it to $K_t$. Moments, payments and costs are linear in the distribution and the cost, so the point of the mixture is the same mixture of the points of $a$ and $\tilde a$. By equation~\eqref{eqn:physical-interim} and the convexity of $\mathcal V_{w_t}^+$, the interim payoff guarantee from the mixture is therefore at most the same mixture of the objective values at the two points, which converges to the minimum as the weight on $\tilde a$ vanishes. Since nature can choose these technologies, $\hat U_t\kh{w_t\mid h^t}$ is at most the minimum, which proves equation~\eqref{eqn:physical-program}.

Suppose instead that $w_t$ admits no strict new response. We show that
\eqn{\hat U_t\kh{w_t\mid h^t}=\max_{a\in K_t:\,p\kh{a}-c\kh{a}=\underline u_{w_t}}\kh{x\kh{a}-p\kh{a}}+\beta\hat V_{t+1}^*\kh{h^{t\to t+1}}.\label{eqn:no-strict-response}}
Since no action compatible with $h^t$ gives the agent more than $\underline u_{w_t}$ under $w_t$, every known action $a$ that gives exactly $\underline u_{w_t}$ is a best response in every technology in $\Omega\kh{h^t}$, and equation~\eqref{eqn:const-updating} shows that the family of compatible technologies after $\kh{h^t,w_t,a}$ is $\Omega\kh{h^t}$. A silent period also leaves this family unchanged, so property~\ref{item:constapp-info} gives $\hat V_{t+1}^*\kh{h^t,w_t,a}=\hat V_{t+1}^*\kh{h^{t\to t+1}}$. Under every compatible technology, favorable tie-breaking therefore gives the principal at least the right-hand side of equation~\eqref{eqn:no-strict-response}, and under the technology $K_t$, whose best responses are exactly these known actions, she obtains exactly this value, which proves the equation.

We now compare $w_t'$ with $w_t$. Write $p'\kh{a}=\mathbb{E}_{F}\fkh{w_t'\kh{y}}$ for the expected payment of $w_t'$ to an action $a=\kh{F,c}$. By assumption, $\overline p_{w_t'}=\overline p_{w_t}$ on the attainable set, and $p'\kh{a}\ge p\kh{a}$ for every $a\in K_t$, so $\underline u_{w_t'}\ge\underline u_{w_t}$.

Suppose first that $w_t$ admits a strict new response, so that $\hat U_t\kh{w_t\mid h^t}$ is the minimum in program~\eqref{eqn:physical-program}. Let $a_t=\kh{F_t,c_t}$ be a best response to $w_t'$ in a compatible technology, and write $p'=p'\kh{a_t}$ and $u'=p'-c_t$. As before, $a_t$ is compatible with $h^t$, and $u'\ge\underline u_{w_t'}\ge\underline u_{w_t}$. Moreover, $p'\le\overline p_{w_t'}\kh{x\kh{a_t},q\kh{a_t}}=\overline p_{w_t}\kh{x\kh{a_t},q\kh{a_t}}$, because the two contracts have the same envelope. Hence $\kh{x\kh{a_t},q\kh{a_t},p',u'}$ is feasible in program~\eqref{eqn:physical-program}. By equation~\eqref{eqn:physical-interim} for $w_t'$, the interim payoff guarantee from $a_t$ under $w_t'$ is
$$x\kh{a_t}-p'+\beta\mathcal V_{w_t'}^+\kh{x\kh{a_t},q\kh{a_t},p',u'}.$$
The states $S_{w_t'}^+\kh{x\kh{a_t},q\kh{a_t},p',u'}$ and $S_{w_t}^+\kh{x\kh{a_t},q\kh{a_t},p',u'}$ differ only in the payments that the coordinates of $A_0$ and of the periods $\tau<t$ record for period $t$: $p'\kh{a}$ instead of $p\kh{a}$ for the coordinate of each $a\in A_0$, and $p'\kh{a_\tau}$ instead of $p\kh{a_\tau}$ for the coordinate of each period $\tau<t$. These payments are weakly higher under $w_t'$ by assumption, and the state remains compatible, because $u'\ge\underline u_{w_t'}\ge p'\kh{a}-c\kh{a}$ for every known action $a$. Property~\ref{item:constapp-payments} therefore gives $\mathcal V_{w_t'}^+\ge\mathcal V_{w_t}^+$ at this point. The interim payoff guarantee from $a_t$ under $w_t'$ is thus at least the objective of program~\eqref{eqn:physical-program} at a feasible point, and hence at least $\hat U_t\kh{w_t\mid h^t}$. This holds for every best response in every compatible technology, so $\hat U_t\kh{w_t'\mid h^t}\ge\hat U_t\kh{w_t\mid h^t}$.

Suppose instead that $w_t$ admits no strict new response, so that $\overline p_{w_t}\kh{x,q}-\underline c\kh{q}\le\underline u_{w_t}$ for every attainable $\kh{x,q}$. For every known action $a$, compatibility gives $c\kh{a}\ge\underline c\kh{q\kh{a}}$, and the two contracts have the same envelope, so
\eqn{p'\kh{a}-c\kh{a}\le\overline p_{w_t'}\kh{x\kh{a},q\kh{a}}-\underline c\kh{q\kh{a}}=\overline p_{w_t}\kh{x\kh{a},q\kh{a}}-\underline c\kh{q\kh{a}}\le\underline u_{w_t}.\label{eqn:known-payoff-bound}}
Thus $\underline u_{w_t'}=\underline u_{w_t}$, and, since the two contracts also have the same envelope, equation~\eqref{eqn:strict-response-condition} shows that $w_t'$ admits no strict new response either. If a known action $a$ gives the agent $\underline u_{w_t}$ under $w_t$, then $p'\kh{a}\ge p\kh{a}$ and equation~\eqref{eqn:known-payoff-bound} give $p'\kh{a}=p\kh{a}$, so $a$ also gives the agent $\underline u_{w_t'}$ under $w_t'$, and it gives the principal the same current payoff. Equation~\eqref{eqn:no-strict-response}, applied to both contracts, therefore gives $\hat U_t\kh{w_t'\mid h^t}\ge\hat U_t\kh{w_t\mid h^t}$.

This completes the proof.
\end{proof}

\begin{proof}[Proof of Corollary~\ref{cor:history-relative-concavification}]
We verify the two conditions of Lemma~\ref{lem:envelopeconst}. Since $w_t'\ge w_t$, every known action receives a weakly higher expected payment under $w_t'$, and $\overline p_{w_t'}\ge\overline p_{w_t}$ at every attainable moment vector. For the reverse inequality, note that $\overline p_{w_t}$ is concave and upper semicontinuous on the compact attainable set, so it is the infimum of the affine functions that lie above it there, and Jensen's inequality applies to it. For every $F\in\Delta\kh{Y}$, since $w_t'\le w_t^{\mathrm{cav}}$,
$$\mathbb{E}_{F}\fkh{w_t'\kh{y}}\le\mathbb{E}_{F}\fkh{\overline p_{w_t}\kh{y,w_1\kh{y},\ldots,w_{t-1}\kh{y}}}\le\overline p_{w_t}\kh{\mathbb{E}_{F}\fkh{y},\mathbb{E}_{F}\fkh{w_1\kh{y}},\ldots,\mathbb{E}_{F}\fkh{w_{t-1}\kh{y}}}.$$
Maximizing the left-hand side over the distributions with given moments shows that $\overline p_{w_t'}\le\overline p_{w_t}$. The two envelopes therefore coincide, and Lemma~\ref{lem:envelopeconst} gives $\hat U_t\kh{w_t'\mid h^t}\ge\hat U_t\kh{w_t\mid h^t}$.
\end{proof}

\begin{proof}[Proof of Lemma~\ref{lem:improvementconst}]
Fix a date $t$, a feasible history $h^t$ with state $S$ and a contract $w_t$, and suppose, as in the proof of Lemma~\ref{lem:envelopeconst}, that $\mathcal V_{t+1}$ has properties~\ref{item:constapp-convex}-\ref{item:constapp-columns} and that $\hat V_{t+1}^*$ satisfies properties~\ref{item:constapp-info} and~\ref{item:constapp-exact}. We find $\theta_t\in\Theta_t$ such that $\hat U_t\kh{w_{\theta_t}\mid h^t}\ge\hat U_t\kh{w_t\mid h^t}$, which proves the lemma because $w_{\theta_t}\in\operatorname{co}\hkh{0,y,w_1,\ldots,w_{t-1}}$.

The contracts in Parts~\ref{part:1}, \ref{part:2} and~\ref{part:3} of Lemma~\ref{lem:second} have the form $w_{\theta_T}$ with $\theta_T\in\Theta_T$: the contract $w_\tau$ in Part~\ref{part:1} has $\eta_\tau=1$, the contract in equation~\eqref{eqn:m} equals $my+\kh{1-m}w_\tau$ and has $\alpha=m$ and $\eta_\tau=1-m$, and the linear contract in Part~\ref{part:3} has $\alpha=s_T$, with all other coefficients zero. If $t=T$, Lemma~\ref{lem:second} shows that one of them is optimal at $h^T$, so let $t<T$.

We first show that an optimal final-period contract after silent periods, offered at $h^t$, guarantees at least the final-period value plus the discounted continuation value after a silent period. Let $w_T$ be an optimal final-period contract at $h^{t\to T}$ given by Lemma~\ref{lem:second}. It is formed from $y$ and the previous contracts of $h^{t\to T}$, and the contracts of the silent periods are zero, so $w_T\in\operatorname{co}\hkh{0,y,w_1,\ldots,w_{t-1}}$.

Fix a compatible technology. Since $\Omega\kh{h^{t\to T}}=\Omega\kh{h^t}$ and $w_T$ guarantees $\hat V_T^*\kh{h^{t\to T}}$ in the final period, some best response to $w_T$ in this technology gives the principal a current payoff of at least $\hat V_T^*\kh{h^{t\to T}}$. By equation~\eqref{eqn:const-updating}, observing a best response $a_t$ can only shrink the family $\Omega\kh{h^t}$, whereas a silent period leaves it unchanged, so property~\ref{item:constapp-info} gives $\hat V_{t+1}^*\kh{h^t,w_T,a_t}\ge\hat V_{t+1}^*\kh{h^{t\to t+1}}$. The best response with the highest current payoff thus has an interim payoff guarantee of at least $\hat V_T^*\kh{h^{t\to T}}+\beta\hat V_{t+1}^*\kh{h^{t\to t+1}}$, and the best response that favorable tie-breaking selects does at least as well. This holds in every compatible technology, so
\eqn{\hat U_t\kh{w_T\mid h^t}\ge\hat V_T^*\kh{h^{t\to T}}+\beta\hat V_{t+1}^*\kh{h^{t\to t+1}}.\label{eqn:silent-guarantee}}

Suppose first that $w_t$ admits no strict new response. We compare $w_t$ with $w_T$ by viewing $w_t$, too, as a final-period contract after silent periods. The history $h^{t\to T}$ has the same known set $K_t$ and the same family $\Omega\kh{h^t}$ of compatible technologies as $h^t$, so $w_t$ admits no strict new response at $h^{t\to T}$ either. At date $T$, the continuation value is zero, so equation~\eqref{eqn:no-strict-response} at $h^{t\to T}$ reads
$$\hat U_T\kh{w_t\mid h^{t\to T}}=\max_{a\in K_t:\,p\kh{a}-c\kh{a}=\underline u_{w_t}}\kh{x\kh{a}-p\kh{a}}.$$

The right-hand side is the first term in equation~\eqref{eqn:no-strict-response} at $h^t$, so
$$\hat U_t\kh{w_t\mid h^t}=\hat U_T\kh{w_t\mid h^{t\to T}}+\beta\hat V_{t+1}^*\kh{h^{t\to t+1}}.$$
Since the agent can do no better than a known action, $w_t$ teaches the principal nothing: it guarantees what it would guarantee in the final period with the same information, plus the discounted continuation value after a silent period. No final-period contract guarantees more than $\hat V_T^*\kh{h^{t\to T}}$, so
\eqn{\hat U_t\kh{w_t\mid h^t}\le\hat V_T^*\kh{h^{t\to T}}+\beta\hat V_{t+1}^*\kh{h^{t\to t+1}}.\label{eqn:no-strict-bound}}
By equation~\eqref{eqn:silent-guarantee}, $w_T$ guarantees at least this sum, so $\hat U_t\kh{w_T\mid h^t}\ge\hat U_t\kh{w_t\mid h^t}$, which proves the lemma in this case.

Suppose instead that $w_t$ admits a strict new response, so that $\hat U_t\kh{w_t\mid h^t}$ is the minimum in program~\eqref{eqn:physical-program}. As outlined in Section~\ref{sec:period1}, we take a worst-case moment vector of this program and an affine function that supports the upper payment envelope from above there, and read off $\theta_t$ from its coefficients.

The proof of Theorem~\ref{prop:1} also needs weights that sum to less than one. To obtain them, we scale output in the principal's current payoff by $1-\epsilon$. For $\epsilon\in[0,1)$, let $P^\epsilon$ be the minimum of $\kh{1-\epsilon}x-p+\beta\mathcal V_{w_t}^+\kh{x,q,p,u}$ over the feasible set of program~\eqref{eqn:physical-program}. The argument in the proof of Lemma~\ref{lem:envelopeconst} still applies, because only the continuous term $-\epsilon x$ has been added: this program has a minimizer, and every minimizer satisfies $p=\overline p_{w_t}\kh{x,q}$. Since $0\le x\le\max Y$ for every attainable $\kh{x,q}$, the perturbation lowers the minimum by at most $\epsilon\max Y$:
\eqn{\hat U_t\kh{w_t\mid h^t}-\epsilon\max Y\le P^\epsilon\le P^0=\hat U_t\kh{w_t\mid h^t}.\label{eqn:perturbed-bound}}

We construct $\theta_t\in\Theta_t$, with $\theta_t\in\Theta_t^<$ if $\epsilon>0$, and an intercept $\kappa\ge0$ such that every relaxed response $\kh{x',q',c'}$ to $\theta_t$ at $S$ satisfies
\eqn{x'-p_{\theta_t}\kh{x',q'}+\beta\mathcal V_{t+1}\kh{S_{\theta_t}^+\kh{x',q',c'}}\ge\hat U_t\kh{w_t\mid h^t}-\epsilon\max Y+\kappa.\label{eqn:transfer}}
The left-hand side is the principal's current payoff from the response under $w_{\theta_t}$ plus the discounted continuation value after it. Once equation~\eqref{eqn:transfer} is established, taking $\epsilon=0$ proves the lemma, because every best response to $w_{\theta_t}$ gives a relaxed response. Taking $\epsilon>0$ instead gives $\theta_t\in\Theta_t^<$ with $\mathcal U_t\kh{\theta_t\mid S}\ge\hat U_t\kh{w_t\mid h^t}-\epsilon\max Y$, as claimed after equation~\eqref{eqn:state-recursion}.

We obtain the supporting function from convex duality, which needs a Slater point. Since $w_t$ admits a strict new response, some attainable moment vector satisfies the strict inequality in equation~\eqref{eqn:strict-response-condition}. Mix it with a point in the relative interior of the attainable set, with a small positive weight on the latter; the mixture lies in the relative interior of the attainable set. Since $\overline p_{w_t}$ is concave and $\underline c$ is convex, $\overline p_{w_t}-\underline c$ is concave, so its value at the mixture is at least the same mixture of its values at the two points, and the strict inequality survives when the weight is small. At the resulting moment vector $\kh{x,q}$, choose $c$, $u$ and $p$ with
\eqn{c>\underline c\kh{q},\quad u>\underline u_{w_t},\quad p=c+u<\overline p_{w_t}\kh{x,q}.\label{eqn:slater}}
The point $\kh{x,q,p,u}$ is a Slater point of program~\eqref{eqn:physical-program}: it is feasible, its moment vector lies in the relative interior of the attainable set, and it satisfies the constraints $u\ge\underline u_{w_t}$, $p-u\ge\underline c\kh{q}$ and $p\le\overline p_{w_t}\kh{x,q}$ strictly.

A moment may be zero for every output distribution, for example $q_\tau$ when $w_\tau$ is the zero contract, and then the Slater point is not an interior point of the set of candidate observations.\footnote{The Slater point then has $q_\tau=0$, while every candidate observation has $q_\tau\ge0$. Lowering $q_\tau$ at the Slater point, by however little, therefore gives a point outside the set of candidate observations.} We therefore extend $\mathcal V_{w_t}^+$ to all moment vectors by treating negative moments as zero. Write $x_+$ and $q_+=\kh{\kh{q_1}_+,\ldots,\kh{q_{t-1}}_+}$ for positive parts, let $\underline c\kh{q}$ be given by the same formula for every $q\in\mathbb{R}^{t-1}$, and define
$$\tilde{\mathcal V}_{w_t}^+\kh{x,q,p,u}=\Brace{&\mathcal V_{w_t}^+\kh{x_+,q_+,p,u},&&\text{if }u\ge\underline u_{w_t}\text{ and }p-u\ge\underline c\kh{q},\\&{+\infty},&&\text{otherwise}.}$$

Since every $u_\tau$ is nonnegative, $\underline c\kh{q_+}=\underline c\kh{q}$, so $\kh{x_+,q_+,p,u}$ is a candidate observation whenever the two constraints hold, and the definition makes sense. The extension has three properties. First, it agrees with $\mathcal V_{w_t}^+$ at every candidate observation, because $x_+=x$ and $q_+=q$ there. Second, it is convex and lower semicontinuous. The two constraints define a closed convex set, because $\underline c$ is continuous and convex. On this set, the extension evaluates $\mathcal V_{w_t}^+$ at the positive parts, which are continuous and convex, and $\mathcal V_{w_t}^+$ is convex, lower semicontinuous and weakly increasing in $x$ and in each $q_\tau$. Third, the Slater point satisfies both constraints strictly, so it is an interior point of the set where $\tilde{\mathcal V}_{w_t}^+$ is finite.

Write a subgradient of $\tilde{\mathcal V}_{w_t}^+$ as $\zeta=\kh{\zeta_x,\zeta_1,\ldots,\zeta_{t-1},\zeta_p,\zeta_u}$, where $\zeta_\tau$ is the component in $q_\tau$. Every subgradient at a point where $\tilde{\mathcal V}_{w_t}^+$ is finite satisfies
\eqn{\zeta_x\ge0,\quad\zeta_\tau\ge0,\quad\zeta_p\le0,\quad\zeta_x+\sum_{\tau<t}\zeta_\tau+\zeta_p\le0.\label{eqn:subgradient-signs}}
Each inequality comes from a change of $\kh{x,q,p,u}$ that does not raise $\tilde{\mathcal V}_{w_t}^+$: for such a change, the subgradient inequality makes the inner product of $\zeta$ with the change nonpositive.\footnote{A vector $\zeta$ is a subgradient of $\tilde{\mathcal V}_{w_t}^+$ at a point $\kh{x,q,p,u}$ where it is finite if, for all changes $\Delta x$, $\Delta q=\kh{\Delta q_1,\ldots,\Delta q_{t-1}}$, $\Delta p$ and $\Delta u$,
$$\tilde{\mathcal V}_{w_t}^+\kh{x+\Delta x,q+\Delta q,p+\Delta p,u+\Delta u}-\tilde{\mathcal V}_{w_t}^+\kh{x,q,p,u}\ge\zeta_x\Delta x+\sum_{\tau<t}\zeta_\tau\Delta q_\tau+\zeta_p\Delta p+\zeta_u\Delta u.$$
The components of $\zeta$ thus act as slopes that bound every change in $\tilde{\mathcal V}_{w_t}^+$ from below. If a change does not raise $\tilde{\mathcal V}_{w_t}^+$, the left-hand side is at most zero, so the right-hand side, the inner product of $\zeta$ with the change, is nonpositive.}

Lowering $x$ by one unit keeps both constraints and does not raise $x_+$, so it does not raise $\tilde{\mathcal V}_{w_t}^+$, and hence $\zeta_x\ge0$. Lowering some $q_\tau$ by one unit gives $\zeta_\tau\ge0$ in the same way, since $\underline c$ is weakly increasing. Raising $p$ by one unit keeps both constraints and does not raise $\tilde{\mathcal V}_{w_t}^+$, because $\mathcal V_{w_t}^+$ is weakly decreasing in $p$, and hence $\zeta_p\le0$. Raising $x$, every $q_\tau$ and $p$ by one unit keeps both constraints, raises the positive parts by at most one unit and raises the cost $p-u$ by exactly one unit, so property~\ref{item:constapp-costs} shows that it does not raise $\tilde{\mathcal V}_{w_t}^+$, and hence $\zeta_x+\sum_{\tau<t}\zeta_\tau+\zeta_p\le0$. The signs thus restate the monotonicity of the continuation value: it rises with expected output and payments, falls with the current payment, and falls when output, payments and cost rise together.

Fix a minimizer $\kh{x^*,q^*,p^*,u^*}$ for $P^\epsilon$, so that $p^*=\overline p_{w_t}\kh{x^*,q^*}$. Since $\tilde{\mathcal V}_{w_t}^+=\mathcal V_{w_t}^+$ at candidate observations, $P^\epsilon$ is also the minimum of the convex function
$$\kh{1-\epsilon}x-p+\beta\tilde{\mathcal V}_{w_t}^+\kh{x,q,p,u}$$
over the points with attainable $\kh{x,q}$ at which $\tilde{\mathcal V}_{w_t}^+$ is finite, subject to the constraint $p\le\overline p_{w_t}\kh{x,q}$. This constraint is convex, because $\overline p_{w_t}$ is concave.

The Slater point satisfies the constraint strictly, its moment vector lies in the relative interior of the attainable set, and it is an interior point of the set where $\tilde{\mathcal V}_{w_t}^+$ is finite. Slater's condition therefore holds, and Lagrangian duality gives a multiplier $\lambda\ge0$ such that the minimum of the Lagrangian
$$\kh{1-\epsilon}x-p+\beta\tilde{\mathcal V}_{w_t}^+\kh{x,q,p,u}+\lambda\kh{p-\overline p_{w_t}\kh{x,q}}$$
over the same points, now without the constraint, is still $P^\epsilon$. The constraint binds at $\kh{x^*,q^*,p^*,u^*}$, so the Lagrangian equals $P^\epsilon$ there, and this point minimizes it. Together with $\lambda$, it forms the saddle point described in Section~\ref{sec:period1}.

The Lagrangian is the sum of two convex functions on $\mathbb{R}^{t+2}$. The first is
$$\kh{1-\epsilon}x+\kh{\lambda-1}p+\beta\tilde{\mathcal V}_{w_t}^+\kh{x,q,p,u},$$
and the second depends only on $\kh{x,q}$: it equals $-\lambda\overline p_{w_t}\kh{x,q}$ at attainable moment vectors and $+\infty$ elsewhere. We split a subgradient of the Lagrangian into subgradients of these two functions, using the subgradient sum rule. The rule requires a point in the relative interior of the set where each function is finite, and the Slater point is one: it is an interior point of the set where the first function is finite, and its moment vector lies in the relative interior of the attainable set, where the second function is finite. Relative interiors matter here because the attainable set may have lower dimension than $\mathbb{R}^t$.

Since $\kh{x^*,q^*,p^*,u^*}$ minimizes the Lagrangian, no change lowers it, so the zero vector is a subgradient of the Lagrangian there. By the sum rule, it is the sum of a subgradient of the first function and a subgradient of the second. A subgradient of the first function is the vector with $x$-component $1-\epsilon$, $p$-component $\lambda-1$ and all other components zero, plus $\beta\zeta$ for a subgradient $\zeta$ of $\tilde{\mathcal V}_{w_t}^+$ at $\kh{x^*,q^*,p^*,u^*}$. The second function does not depend on $p$ and $u$, so its subgradients have zero $p$- and $u$-components. The $p$- and $u$-components of the sum are therefore $\lambda-1+\beta\zeta_p=0$ and $\beta\zeta_u=0$, that is,
\eqn{\lambda=1-\beta\zeta_p\ge1,\quad\zeta_u=0,\label{eqn:stationarity-pu}}
where $\lambda\ge1$ follows from $\zeta_p\le0$ in equation~\eqref{eqn:subgradient-signs}. The $x$- and $q$-components of the sum show that the vector with entries $-\kh{1-\epsilon+\beta\zeta_x}$ and $-\beta\zeta_\tau$, $\tau<t$, is a subgradient of the second function at $\kh{x^*,q^*}$.

The subgradient of the second function found above gives, for every attainable $\kh{x,q}$,
$$-\lambda\overline p_{w_t}\kh{x,q}\ge-\lambda\overline p_{w_t}\kh{x^*,q^*}-\kh{1-\epsilon+\beta\zeta_x}\kh{x-x^*}-\sum_{\tau<t}\beta\zeta_\tau\kh{q_\tau-q_\tau^*}.$$
Define $\theta_t=\kh{\alpha,\eta_1,\ldots,\eta_{t-1}}$ and $\kappa$ by
\eqn{\alpha=\frac{1-\epsilon+\beta\zeta_x}{\lambda},\quad\eta_\tau=\frac{\beta\zeta_\tau}{\lambda},\quad\kappa=\overline p_{w_t}\kh{x^*,q^*}-\alpha x^*-\sum_{\tau<t}\eta_\tau q_\tau^*.\label{eqn:supporting-coefficients}}
Dividing the inequality by $-\lambda$ and using equation~\eqref{eqn:supporting-coefficients} shows that the affine function $p_{\theta_t}\kh{x,q}+\kappa$ supports the upper payment envelope from above at $\kh{x^*,q^*}$, as in Section~\ref{sec:period1}: for every attainable $\kh{x,q}$,
\eqn{\overline p_{w_t}\kh{x,q}\le\alpha x+\sum_{\tau<t}\eta_\tau q_\tau+\kappa=p_{\theta_t}\kh{x,q}+\kappa,\label{eqn:supporting-plane}}
with equality at $\kh{x^*,q^*}$.

By equations~\eqref{eqn:subgradient-signs} and~\eqref{eqn:stationarity-pu}, $\alpha$ and every $\eta_\tau$ are nonnegative, and
$$\alpha+\sum_{\tau<t}\eta_\tau=\frac{1-\epsilon+\beta\kh{\zeta_x+\sum_{\tau<t}\zeta_\tau}}{\lambda}\le\frac{1-\epsilon-\beta\zeta_p}{\lambda}=1-\frac{\epsilon}{\lambda},$$
so $\theta_t\in\Theta_t$, and $\theta_t\in\Theta_t^<$ if $\epsilon>0$. Evaluating equation~\eqref{eqn:supporting-plane} at the zero moment vector, where the envelope is zero, gives $\kappa\ge0$.

The subgradient $\zeta$ now extends the bound $P^\epsilon$ from the feasible set of the program to all candidate observations, including those whose moments are not attainable: for every candidate observation $\kh{x,q,p,u}$,
\eqn{\kh{1-\epsilon}x-p+\beta\mathcal V_{w_t}^+\kh{x,q,p,u}+\lambda\kh{p-p_{\theta_t}\kh{x,q}-\kappa}\ge P^\epsilon.\label{eqn:global-support}}
At $\kh{x^*,q^*,p^*,u^*}$, the left-hand side equals $P^\epsilon$, because this point attains $P^\epsilon$ and $p^*=p_{\theta_t}\kh{x^*,q^*}+\kappa$. By equations~\eqref{eqn:stationarity-pu} and~\eqref{eqn:supporting-coefficients}, $1-\epsilon-\lambda\alpha=-\beta\zeta_x$, $\lambda\eta_\tau=\beta\zeta_\tau$ and $\lambda-1=-\beta\zeta_p$, so at any candidate observation, the left-hand side minus $P^\epsilon$ equals $\beta$ times
$$\mathcal V_{w_t}^+\kh{x,q,p,u}-\mathcal V_{w_t}^+\kh{x^*,q^*,p^*,u^*}-\zeta_x\kh{x-x^*}-\sum_{\tau<t}\zeta_\tau\kh{q_\tau-q_\tau^*}-\zeta_p\kh{p-p^*}.$$
Since $\zeta_u=0$ and $\tilde{\mathcal V}_{w_t}^+=\mathcal V_{w_t}^+$ at candidate observations, this is nonnegative by the subgradient inequality for $\tilde{\mathcal V}_{w_t}^+$ at $\kh{x^*,q^*,p^*,u^*}$.

Now let $\kh{x',q',c'}$ be a relaxed response to $\theta_t$ at $S$. Keep its moments and cost, and let its payment be $p=p_{\theta_t}\kh{x',q'}+\kappa$, the value of the supporting affine function at these moments; the agent's payoff is then $u=p-c'=u_{\theta_t}+\kappa$. At $\kh{x',q',p,u}$, the term $\lambda\kh{p-p_{\theta_t}\kh{x',q'}-\kappa}$ on the left-hand side of equation~\eqref{eqn:global-support} is therefore zero. Moreover, $p-u=c'\ge\underline c\kh{q'}$, because $\kh{x',q',c'}$ is a relaxed response. Since $S$ is the state of $h^t$, the payoff floor $\underline u_{\theta_t}\kh{S}$ is computed from the known actions, and equation~\eqref{eqn:supporting-plane} gives $p_{\theta_t}\kh{x\kh{a},q\kh{a}}+\kappa\ge\overline p_{w_t}\kh{x\kh{a},q\kh{a}}\ge p\kh{a}$ for every $a\in K_t$. Hence
$$u\ge\underline u_{\theta_t}\kh{S}+\kappa=\max\hkh{\kappa,\max_{a\in K_t}\kh{p_{\theta_t}\kh{x\kh{a},q\kh{a}}+\kappa-c\kh{a}}}\ge\max_{a\in K_t}\kh{p\kh{a}-c\kh{a}}=\underline u_{w_t}.$$
So $\kh{x',q',p,u}$ is a candidate observation under $w_t$, and equation~\eqref{eqn:global-support} applies to it.

Next, we turn the state $S_{w_t}^+\kh{x',q',p,u}$ into $S_{\theta_t}^+\kh{x',q',c'}$. The two states agree except in the payments that the coordinates record for period $t$ and in the revealed payoff of period $t$, and we change these entries in two steps, each of which weakly raises $\mathcal V_{t+1}$. First, for each coordinate other than that of period $t$, with action $a\in K_t$, raise the payment recorded for period $t$ from $p\kh{a}$ to $p_{\theta_t}\kh{x\kh{a},q\kh{a}}+\kappa$. The coordinate of period $t$ needs no change: its payment for period $t$ is $p$, which equals $p_{\theta_t}\kh{x',q'}+\kappa$ by construction. After this step, every coordinate records for period $t$ the value of the supporting affine function at its moments. The resulting state is compatible, because $u\ge\underline u_{\theta_t}\kh{S}+\kappa\ge p_{\theta_t}\kh{x\kh{a},q\kh{a}}+\kappa-c\kh{a}$, so property~\ref{item:constapp-payments} shows that this step weakly raises the continuation value.

Second, lower the payment that every coordinate records for period $t$, and the revealed payoff $u$ of period $t$, by $\kappa$. This is possible because each of these entries is at least $\kappa$: the payments are values of $p_{\theta_t}+\kappa$ with $p_{\theta_t}\ge0$, and $u=u_{\theta_t}+\kappa$ with $u_{\theta_t}\ge\underline u_{\theta_t}\kh{S}\ge0$. Property~\ref{item:constapp-columns}, applied to all coordinates with the shift $\kappa$ for period $t$ and no shift for earlier periods, shows that this step also weakly raises the continuation value. The resulting state records the payment $p_{\theta_t}\kh{x,q}$ for each existing coordinate with value $\kh{x,q,c}$, the value $\kh{x',\kh{q',p_{\theta_t}\kh{x',q'}},c'}$ for period $t$ and the revealed payoff $u_{\theta_t}$, so it is $S_{\theta_t}^+\kh{x',q',c'}$. Hence,
$$\mathcal V_{t+1}\kh{S_{\theta_t}^+\kh{x',q',c'}}\ge\mathcal V_{w_t}^+\kh{x',q',p,u}.$$

Since the term $\lambda\kh{p-p_{\theta_t}\kh{x',q'}-\kappa}$ is zero and $p=p_{\theta_t}\kh{x',q'}+\kappa$, equation~\eqref{eqn:global-support} at $\kh{x',q',p,u}$ reads
$$\kh{1-\epsilon}x'-p_{\theta_t}\kh{x',q'}-\kappa+\beta\mathcal V_{w_t}^+\kh{x',q',p,u}\ge P^\epsilon.$$
Replacing $\mathcal V_{w_t}^+\kh{x',q',p,u}$ by the weakly larger $\mathcal V_{t+1}\kh{S_{\theta_t}^+\kh{x',q',c'}}$ and rearranging gives
$$x'-p_{\theta_t}\kh{x',q'}+\beta\mathcal V_{t+1}\kh{S_{\theta_t}^+\kh{x',q',c'}}\ge P^\epsilon+\epsilon x'+\kappa,$$
and equation~\eqref{eqn:transfer} follows from $x'\ge0$ and equation~\eqref{eqn:perturbed-bound}.

Finally, let $\epsilon=0$. A best response $a_t$ to $w_{\theta_t}$ in a compatible technology is compatible with $h^t$ and gives the agent at least $\underline u_{w_{\theta_t}}=\underline u_{\theta_t}\kh{S}$, so $\kh{x\kh{a_t},q\kh{a_t},c\kh{a_t}}$ is a relaxed response to $\theta_t$ at $S$. The state of $\kh{h^t,w_{\theta_t},a_t}$ is $S_{\theta_t}^+\kh{x\kh{a_t},q\kh{a_t},c\kh{a_t}}$, so, by property~\ref{item:constapp-exact}, the interim payoff guarantee from $a_t$ is the left-hand side of equation~\eqref{eqn:transfer} at this response, which is at least $\hat U_t\kh{w_t\mid h^t}+\kappa$. This holds for every best response in every compatible technology, so
$$\hat U_t\kh{w_{\theta_t}\mid h^t}\ge\hat U_t\kh{w_t\mid h^t}+\kappa\ge\hat U_t\kh{w_t\mid h^t}.$$
If instead $\epsilon>0$, then $\theta_t\in\Theta_t^<$, and taking the infimum in equation~\eqref{eqn:transfer} over relaxed responses gives $\mathcal U_t\kh{\theta_t\mid S}\ge\hat U_t\kh{w_t\mid h^t}-\epsilon\max Y$, as claimed after equation~\eqref{eqn:state-recursion}.

This completes the proof.
\end{proof}

\begin{proof}[Proof of Lemma~\ref{lem:existenceconst}]
Fix a date $t$, and suppose that $\hat V_{t+1}^*$ satisfies properties~\ref{item:constapp-bound} and~\ref{item:constapp-usc}. We show three things: $\hat U_t\kh{w_{\theta_t}\mid h^t}$ is upper semicontinuous jointly in the feasible history $h^t$ and $\theta_t\in\Theta_t$; its maximum over $\theta_t\in\Theta_t$ exists; and this maximum is upper semicontinuous in $h^t$. The first two give the lemma, and the third gives property~\ref{item:constapp-usc} at date $t$ in the proof of Theorem~\ref{prop:1}.

Take a sequence of feasible histories $h^{t,n}$ and coefficients $\theta_t^n\in\Theta_t$ with $\kh{h^{t,n},\theta_t^n}\to\kh{h^t,\theta_t}$. Write $w_\tau^n$ and $u_\tau^n$ for the previous contracts and revealed payoffs of $h^{t,n}$, and $w_t^n$ for the contract with coefficients $\theta_t^n$ formed from $w_1^n,\ldots,w_{t-1}^n$. Since the coefficients and the previous contracts converge, $w_t^n$ converges uniformly to $w_{\theta_t}$.

Fix a compatible technology $A_t\in\Omega\kh{h^t}$. Nature may be unable to choose $A_t$ at $h^{t,n}$, because an action of $A_t$ may violate the incentive restrictions of $h^{t,n}$, so we raise each cost just enough to restore them. For $a=\kh{F,c}\in A_t$, let
$$c^n\kh{a}=\max\hkh{c,\max_{\tau<t}\kh{\mathbb{E}_F\fkh{w_\tau^n\kh{y}}-u_\tau^n}},$$
with the inner maximum omitted if $t=1$, and let $A_t^n$ consist of the actions $\kh{F,c^n\kh{a}}$ with $a=\kh{F,c}\in A_t$, together with $K\kh{h^{t,n}}$. Since $c^n\kh{a}$ is continuous in $a$, the set $A_t^n$ is compact. It is compatible with $h^{t,n}$: it contains $K\kh{h^{t,n}}$, the actions $\kh{F,c^n\kh{a}}$ satisfy the incentive restrictions of $h^{t,n}$ by construction, and the known actions satisfy them because $h^{t,n}$ is feasible. Every $a=\kh{F,c}\in A_t$ is compatible with $h^t$, that is, $\mathbb{E}_F\fkh{w_\tau\kh{y}}-c\le u_\tau$ for every $\tau<t$, so the cost increase is small:
$$0\le c^n\kh{a}-c\le\max_{\tau<t}\kh{\met{w_\tau^n-w_\tau}_\infty+\abs{u_\tau^n-u_\tau}}$$
for every $a\in A_t$, and the right-hand side converges to zero, because the previous contracts and the revealed payoffs converge. Together with $K\kh{h^{t,n}}\to K_t\subseteq A_t$, this gives $A_t^n\to A_t$.

Let $a_t^n$ be the best response to $w_t^n$ in $A_t^n$ selected under the tie-breaking convention. Since nature can choose $A_t^n$ at $h^{t,n}$,
\eqn{\hat U_t\kh{w_t^n\mid h^{t,n}}\le\hat U_t\kh{w_t^n\mid h^{t,n},a_t^n}.\label{eqn:constapp-usc-comparison}}
The contracts $w_t^n$ converge uniformly, so they are uniformly bounded, and the principal's current payoffs from $a_t^n$ are bounded; together with property~\ref{item:constapp-bound} at date $t+1$, this bounds the right-hand side in $n$. Since $A_t\cup\bigcup_nA_t^n$ is compact, some subsequence makes the right-hand side converge to its limit superior and $a_t^n$ converge to some $a_t$, which belongs to $A_t$ because $A_t^n\to A_t$. Every action $b=\kh{F,c}\in A_t$ is the limit of the actions $\kh{F,c^n\kh{b}}\in A_t^n$, so taking limits in the inequalities that make $a_t^n$ a best response shows that $a_t\in BR\kh{w_{\theta_t}\mid A_t}$.

Along this subsequence, the histories $\kh{h^{t,n},w_t^n,a_t^n}$ converge to $\kh{h^t,w_{\theta_t},a_t}$. These histories are feasible, because $a_t^n$ and $a_t$ are best responses in the compatible technologies $A_t^n$ and $A_t$. The principal's current payoffs converge, and property~\ref{item:constapp-usc} at date $t+1$ bounds the limit superior of the continuation values, so equation~\eqref{eqn:constapp-usc-comparison} gives
$$\limsup_n\hat U_t\kh{w_t^n\mid h^{t,n}}\le\hat U_t\kh{w_{\theta_t}\mid h^t,a_t}\le\max_{b\in BR\kh{w_{\theta_t}\mid A_t}}\hat U_t\kh{w_{\theta_t}\mid h^t,b}.$$
Since $A_t\in\Omega\kh{h^t}$ was arbitrary, taking the infimum over compatible technologies gives
$$\limsup_n\hat U_t\kh{w_t^n\mid h^{t,n}}\le\hat U_t\kh{w_{\theta_t}\mid h^t},$$
which proves joint upper semicontinuity.

For each feasible history $h^t$, upper semicontinuity in $\theta_t$ and compactness of $\Theta_t$ give a maximizing coefficient vector. Now let $h^{t,n}\to h^t$, let $\theta_t^n$ be a maximizing coefficient vector at $h^{t,n}$, and define $w_t^n$ as above. The maximized values $\hat U_t\kh{w_t^n\mid h^{t,n}}$ are nonnegative, because the zero contract, with $\theta_t=0$, guarantees a nonnegative payoff, and at most $\kh{\max Y}\sum_{j=0}^{T-t}\beta^j$, because payments are nonnegative and property~\ref{item:constapp-bound} holds at date $t+1$. Since $\Theta_t$ is compact, some subsequence makes the maximized values converge to their limit superior and $\theta_t^n$ converge to some $\theta_t\in\Theta_t$. Joint upper semicontinuity, applied along this subsequence, gives
$$\limsup_n\hat U_t\kh{w_t^n\mid h^{t,n}}\le\hat U_t\kh{w_{\theta_t}\mid h^t}\le\max_{\theta_t'\in\Theta_t}\hat U_t\kh{w_{\theta_t'}\mid h^t},$$
which proves that the maximized value is upper semicontinuous in $h^t$.
\end{proof}

\begin{proof}[Proof of Theorem~\ref{prop:1}]
We first verify properties~\ref{item:constapp-convex}-\ref{item:constapp-columns} by backward induction. They hold at $t=T+1$, where $\mathcal V_{T+1}\equiv0$. Equation~\eqref{eqn:final-state} defines $\mathcal V_T\kh{S}$ for a state $S$ at any date $t\le T$, and equation~\eqref{eqn:state-recursion} uses $\mathcal V_T\kh{S}$ for states $S$ at dates $t<T$. At $t=T$, we therefore verify the properties for $\mathcal V_T$ on states at all dates $t\le T$.

The properties of $\mathcal V_T$ follow from those of $\xi$, which are easiest to see from an equivalent formulation. In the proof of Lemma~\ref{lem:second}, a contract with share $m$ guarantees at least $\kh{1-m}\kh{r-g/m}$, where $r$ is the principal's payoff and $g$ is the incentive gap, and $\xi\kh{r,g}$ is the best of these bounds: for $r\in\mathbb{R}$ and $g\ge0$,
$$\xi\kh{r,g}=\sup_{0<m\le1}\kh{1-m}\kh{r-\frac{g}{m}}.$$
Indeed, if $r\le g$, every bound is nonpositive and the bound with $m=1$ is zero, so the supremum is zero, as is $\xi\kh{r,g}$. If $r>g$, the supremum is $\kh{\sqrt r-\sqrt g}^2=\xi\kh{r,g}$, attained at $m=\sqrt{g/r}$ when $g>0$ and approached as $m$ decreases to zero when $g=0$.

Each function in the supremum is affine in $\kh{r,g}$, with a nonnegative coefficient on $r$ and a nonpositive coefficient on $g$, so $\xi$ is jointly convex, weakly increasing in $r$ and weakly decreasing in $g$. Moreover, $\kh{1-m}\kh{r+d-\kh{g+d}/m}=\kh{1-m}\kh{r-g/m}-\kh{1-m}^2d/m$, so
\eqn{\xi\kh{r+d,g+d}\le\xi\kh{r,g}\quad\text{for every }d\ge0;\label{eqn:common-shift}}
equivalently, subtracting the same amount from both arguments, with the second argument remaining nonnegative, weakly raises $\xi$. Finally, the closed form $\xi\kh{r,g}=\kh{\sqrt{r_+}-\sqrt g}_+^2$ shows that $\xi$ is continuous.

We now verify the properties for $\mathcal V_T$. Every term in equation~\eqref{eqn:final-state} is nonnegative and at most the expected output of its triple, because $\xi\kh{r,g}\le r_+$ and $q_\tau\ge0$, so $0\le\mathcal V_T\kh{S}\le X\kh{S}$. On the compatible states, the incentive gap $u_\tau-q_\tau+c$ in every historical term is nonnegative, and each term is $\xi$ evaluated at an affine function of the state, so it is convex in the state, and so is the maximum over coordinates. This maximum is also continuous in the state. If $S^n\to S$, the values of the coordinates and the revealed payoffs of $S^n$ converge to those of $S$, uniformly over the coordinates, and they all lie in a common compact set, on which $\xi$ is uniformly continuous. Each term for $S^n$ therefore converges to the corresponding term for $S$, uniformly over the coordinates, and since a maximum changes by at most the largest change in the functions maximized, $\mathcal V_T\kh{S^n}\to\mathcal V_T\kh{S}$. This gives property~\ref{item:constapp-convex}; its last sentence, about the infimum in equation~\eqref{eqn:relaxed-value}, concerns only dates $t<T$ and is proved below.

For property~\ref{item:constapp-outputs}, raising the expected output $x$ of a coordinate raises the first argument of each of its terms, $\xi\kh{x,c}$ and $\xi\kh{x-q_\tau,u_\tau-q_\tau+c}$, and leaves the other terms unchanged, so $\mathcal V_T$ weakly rises because $\xi$ is weakly increasing in its first argument. 

For property~\ref{item:constapp-payments}, raising the payment $q_\tau$ of a coordinate lowers both arguments of its historical term for $\tau$ by the same amount and leaves the other terms unchanged. The second argument stays nonnegative because the new state is compatible, so equation~\eqref{eqn:common-shift} shows that this term, and hence $\mathcal V_T$, weakly rises.

For property~\ref{item:constapp-costs}, the static term of a modified coordinate satisfies
$$\xi\kh{x+\gamma_0d,c+d}\le\xi\kh{x+d,c+d}\le\xi\kh{x,c},$$
by monotonicity in the first argument and equation~\eqref{eqn:common-shift}. In a historical term of a modified coordinate, the second argument rises by $\kh{1-\gamma_\tau}d$ and the first by $\kh{\gamma_0-\gamma_\tau}d$, which is at most $\kh{1-\gamma_\tau}d$, so the same two facts show that the term weakly falls. Hence $\mathcal V_T$ weakly falls.

For property~\ref{item:constapp-columns}, the change leaves every static term unchanged. For every coordinate in the set, it leaves the incentive gap $u_\tau-q_\tau+c$ unchanged and raises the principal's payoff $x-q_\tau$ by $\kappa_\tau$; for every other coordinate, it lowers the incentive gap by $\kappa_\tau$, which keeps it nonnegative because $S^-$ is compatible, and leaves the principal's payoff unchanged. Since $\xi$ is weakly increasing in its first argument and weakly decreasing in its second, $\mathcal V_T$ weakly rises.

Suppose now that properties~\ref{item:constapp-convex}-\ref{item:constapp-columns} hold at date $t+1$, and let $t<T$. Each property passes from the two branches in equation~\eqref{eqn:state-recursion} to their maximum, so we verify it for each branch.

The first branch is $\mathcal V_T\kh{S}+\beta\mathcal V_{t+1}\kh{S^0}$. The map $S\mapsto S^0$ is affine and continuous, maps compatible states to compatible states and leaves $X$ unchanged. It also commutes with each change in properties~\ref{item:constapp-outputs}-\ref{item:constapp-columns}: making the change to $S$ and then appending a silent period gives the same state as appending the silent period first and then making the change, which does not touch the entries for period $t$. The first branch therefore inherits properties~\ref{item:constapp-convex}-\ref{item:constapp-columns} from $\mathcal V_T$ on states at date $t$ and from $\mathcal V_{t+1}$; in particular, it lies between zero and $X\kh{S}+\beta X\kh{S}\sum_{j=0}^{T-t-1}\beta^j=X\kh{S}\sum_{j=0}^{T-t}\beta^j$.

We turn to the second branch, the supremum of $\mathcal U_t\kh{\theta_t\mid S}$ over $\theta_t\in\Theta_t^<$. Each property passes from the functions $\mathcal U_t\kh{\theta_t\mid\cdot}$ to their supremum, provided that the upper bound in property~\ref{item:constapp-convex} does not depend on $\theta_t$, so we fix $\theta_t=\kh{\alpha,\eta_1,\ldots,\eta_{t-1}}\in\Theta_t^<$ and verify the properties for $\mathcal U_t\kh{\theta_t\mid\cdot}$. For property~\ref{item:constapp-convex}, we first bound the principal's current payoff from below. Let $\kh{x',q',c'}$ be a relaxed response to $\theta_t$ at $S$. Since $u_{\theta_t}\ge0$, the cost is at most the payment, $c'\le p_{\theta_t}\kh{x',q'}$, and since $q_\tau'\le c'+u_\tau$ for every $\tau<t$,
$$p_{\theta_t}\kh{x',q'}=\alpha x'+\sum_{\tau<t}\eta_\tau q_\tau'\le\alpha x'+\sum_{\tau<t}\eta_\tau\kh{p_{\theta_t}\kh{x',q'}+u_\tau}.$$
Since $\sum_{\tau<t}\eta_\tau<1$, solving for the payment gives $p_{\theta_t}\kh{x',q'}\le\kh{\alpha x'+\sum_{\tau<t}\eta_\tau u_\tau}/\kh{1-\sum_{\tau<t}\eta_\tau}$, and hence
\eqn{x'-p_{\theta_t}\kh{x',q'}\ge\frac{\kh{1-\alpha-\sum_{\tau<t}\eta_\tau}x'-\sum_{\tau<t}\eta_\tau u_\tau}{1-\sum_{\tau<t}\eta_\tau}.\label{eqn:current-payoff-bound}}
Since $\mathcal V_{t+1}\ge0$, the objective in equation~\eqref{eqn:relaxed-value}, $x'-p_{\theta_t}\kh{x',q'}+\beta\mathcal V_{t+1}\kh{S_{\theta_t}^+\kh{x',q',c'}}$, is at least the right-hand side of equation~\eqref{eqn:current-payoff-bound}, which increases in $x'$ because $\alpha+\sum_{\tau<t}\eta_\tau<1$. Two consequences follow. First, since $x'\ge0$, the objective is at least $-\sum_{\tau<t}\eta_\tau u_\tau/\kh{1-\sum_{\tau<t}\eta_\tau}$ at every relaxed response, so it is bounded below. Second, if the objective at a relaxed response is at most a number $L$, then so is the right-hand side of equation~\eqref{eqn:current-payoff-bound}, and solving this inequality for $x'$ gives
$$x'\le\frac{L\kh{1-\sum_{\tau<t}\eta_\tau}+\sum_{\tau<t}\eta_\tau u_\tau}{1-\alpha-\sum_{\tau<t}\eta_\tau}.$$
In turn, the payment is bounded, because $p_{\theta_t}\kh{x',q'}\le\kh{\alpha x'+\sum_{\tau<t}\eta_\tau u_\tau}/\kh{1-\sum_{\tau<t}\eta_\tau}$; so is the cost, because $c'\le p_{\theta_t}\kh{x',q'}$; and so is every $q_\tau'$, because $q_\tau'\le c'+u_\tau$. These bounds involve the state only through its revealed payoffs $u_\tau$, and they increase with them. Along a convergent sequence of states, the revealed payoffs are bounded, so the relaxed responses at which the objective is at most $L$ lie in a single bounded set for all states in the sequence.

To bound $\mathcal U_t\kh{\theta_t\mid S}$ from above, we evaluate the objective at one relaxed response. Take a triple in $\mathcal M\kh{S}$ that maximizes $p_{\theta_t}\kh{x,q}-c$, or the zero triple if this maximum is negative. This triple is a relaxed response, because $S$ is compatible and the triple's payoff under $\theta_t$ equals the payoff floor. At this triple, the current payoff is at most the expected output, which is at most $X\kh{S}$. The update leaves $X$ unchanged, so property~\ref{item:constapp-convex} at date $t+1$ bounds the continuation value by $X\kh{S}\sum_{j=0}^{T-t-1}\beta^j$. The objective at this triple is therefore at most $X\kh{S}+\beta X\kh{S}\sum_{j=0}^{T-t-1}\beta^j=X\kh{S}\sum_{j=0}^{T-t}\beta^j$. The relaxed responses form a closed set, because they are defined by weak inequalities. The objective is lower semicontinuous in the response, because the update is continuous in the response and $\mathcal V_{t+1}$ is lower semicontinuous. The relaxed responses at which the objective is at most $X\kh{S}\sum_{j=0}^{T-t}\beta^j$ therefore form a set that is nonempty, closed and, by the previous paragraph, bounded, so the objective attains its minimum on it. Hence the infimum in equation~\eqref{eqn:relaxed-value} is attained, and $\mathcal U_t\kh{\theta_t\mid S}$ is finite and at most $X\kh{S}\sum_{j=0}^{T-t}\beta^j$.

Next, we show that $\mathcal U_t\kh{\theta_t\mid\cdot}$ is lower semicontinuous. Let compatible states $S^n$, with revealed payoffs $u_\tau^n$, converge to $S$, and choose a minimizer at each $S^n$. The values satisfy
$$-\frac{\sum_{\tau<t}\eta_\tau u_\tau^n}{1-\sum_{\tau<t}\eta_\tau}\le\mathcal U_t\kh{\theta_t\mid S^n}\le X\kh{S^n}\sum_{j=0}^{T-t}\beta^j,$$
and both bounds converge, so the limit inferior of the values is finite, and some subsequence makes the values converge to it. At each minimizer, the objective equals the value $\mathcal U_t\kh{\theta_t\mid S^n}$, and along this subsequence these values converge, so the objectives at the minimizers are bounded. With $L$ a common bound on these objectives, the bounds on relaxed responses derived above place all the minimizers in a single bounded set. A further subsequence makes them converge to some triple $\kh{x',q',c'}$. This triple is a relaxed response at $S$, because the inequalities that define relaxed responses are weak and the payoff floor $\underline u_{\theta_t}$ is continuous in the state. Along the subsequence, the principal's current payoffs at the minimizers converge to her current payoff at $\kh{x',q',c'}$. Since the update is continuous jointly in the state and the response, the updated states converge to $S_{\theta_t}^+\kh{x',q',c'}$. Since $\mathcal V_{t+1}$ is lower semicontinuous, the continuation value there is at most the limit inferior of the continuation values at the minimizers. Hence the objective at $\kh{x',q',c'}$ is at most the limit of the values along the subsequence, which is their limit inferior. Since $\kh{x',q',c'}$ is a relaxed response at $S$, this gives $\mathcal U_t\kh{\theta_t\mid S}\le\liminf_n\mathcal U_t\kh{\theta_t\mid S^n}$.

For convexity, recall that the update is affine jointly in the state and the response. A triple is a relaxed response at $S$ exactly when it satisfies the following inequalities, each affine jointly in the state and the response: nonnegativity, $q_\tau'-c'\le u_\tau$ for every $\tau<t$, $u_{\theta_t}\ge0$, and $u_{\theta_t}\ge p_{\theta_t}\kh{x,q}-c$ for the value $\kh{x,q,c}$ of every coordinate of $S$; the last two together imply that $u_{\theta_t}\ge\underline u_{\theta_t}\kh{S}$. The pairs of a compatible state and a relaxed response at it therefore form a convex set, on which the objective is jointly convex, because $\mathcal V_{t+1}$ is convex. Minimizing a jointly convex function over one argument, on a convex set, gives a convex function of the other, so $\mathcal U_t\kh{\theta_t\mid\cdot}$ is convex. Finally, the upper bound $X\kh{S}\sum_{j=0}^{T-t}\beta^j$ does not depend on $\theta_t$, so the supremum over $\Theta_t^<$ and the maximum with the first branch preserve convexity, lower semicontinuity and this bound, while the first branch, which is nonnegative, gives the lower bound zero. This proves property~\ref{item:constapp-convex} at date $t$.

Properties~\ref{item:constapp-outputs} and~\ref{item:constapp-payments} follow from the same argument. Write $\tilde S$ for the state after the change. Each change in these properties weakly raises the payment $p_{\theta_t}\kh{x,q}$ of every coordinate, so $\underline u_{\theta_t}\kh{\tilde S}\ge\underline u_{\theta_t}\kh{S}$, and every relaxed response at $\tilde S$ is also one at $S$. At such a response $\kh{x',q',c'}$, the current payoff is the same at $S$ and $\tilde S$. If, at every such response,
$$\mathcal V_{t+1}\kh{\tilde S_{\theta_t}^+\kh{x',q',c'}}\ge\mathcal V_{t+1}\kh{S_{\theta_t}^+\kh{x',q',c'}},$$
then $\mathcal U_t\kh{\theta_t\mid\tilde S}$, the infimum of a weakly higher objective over a smaller set, is at least $\mathcal U_t\kh{\theta_t\mid S}$. It remains to compare these continuation values.

For property~\ref{item:constapp-outputs}, the change raises the expected outputs of some coordinates, and hence their payments $p_{\theta_t}\kh{x,q}$, by $\alpha$ times the increases. At a relaxed response $\kh{x',q',c'}$ at $\tilde S$, the state $\tilde S_{\theta_t}^+\kh{x',q',c'}$ is obtained from $S_{\theta_t}^+\kh{x',q',c'}$ in two steps: raise the expected outputs of these coordinates, which does not affect compatibility, and then raise the payments that they record for period $t$, which keeps the state compatible because the response meets the payoff floor $\underline u_{\theta_t}\kh{\tilde S}$. By properties~\ref{item:constapp-outputs} and~\ref{item:constapp-payments} at date $t+1$, each step weakly raises $\mathcal V_{t+1}$.

For property~\ref{item:constapp-payments}, the change raises the payments of some coordinates under previous contracts, with $\tilde S$ compatible, and hence weakly raises their payments $p_{\theta_t}\kh{x,q}$, because every $\eta_\tau$ is nonnegative. At a relaxed response $\kh{x',q',c'}$ at $\tilde S$, the state $\tilde S_{\theta_t}^+\kh{x',q',c'}$ differs from $S_{\theta_t}^+\kh{x',q',c'}$ only in having weakly higher payments, including those recorded for period $t$. It is compatible, because $\tilde S$ is compatible and the response meets the payoff floor $\underline u_{\theta_t}\kh{\tilde S}$, so property~\ref{item:constapp-payments} at date $t+1$ shows that $\mathcal V_{t+1}$ weakly rises.

For property~\ref{item:constapp-costs}, write $\tilde S$ for the state after the change described there. The payment $p_{\theta_t}\kh{x,q}$ of a modified coordinate rises by $\gamma d$, where $\gamma=\alpha\gamma_0+\sum_{\tau<t}\eta_\tau\gamma_\tau$ lies in $\fkh{0,1}$ because $\alpha+\sum_{\tau<t}\eta_\tau\le1$. Its cost rises by $d$, so the agent's payoff $p_{\theta_t}\kh{x,q}-c$ from it weakly falls. Hence $\underline u_{\theta_t}\kh{\tilde S}\le\underline u_{\theta_t}\kh{S}$, and every relaxed response at $S$ is also one at $\tilde S$.

Fix such a response $\kh{x',q',c'}$. Its current payoff is the same at $S$ and $\tilde S$. The state $\tilde S_{\theta_t}^+\kh{x',q',c'}$ differs from $S_{\theta_t}^+\kh{x',q',c'}$ only in the modified coordinates. In $\tilde S_{\theta_t}^+\kh{x',q',c'}$, each modified coordinate has its value in $\tilde S$ instead of its value in $S$. Its payment for period $t$ is therefore higher by $\gamma d$. Since the response is a relaxed response at $\tilde S$, the state $\tilde S_{\theta_t}^+\kh{x',q',c'}$ is compatible. Property~\ref{item:constapp-costs} at date $t+1$, with the fraction $\gamma$ for period $t$, therefore gives
$$\mathcal V_{t+1}\kh{\tilde S_{\theta_t}^+\kh{x',q',c'}}\le\mathcal V_{t+1}\kh{S_{\theta_t}^+\kh{x',q',c'}}.$$
Hence, at every relaxed response at $S$, the objective is weakly lower at $\tilde S$ than at $S$. Taking the infimum over these responses, which are also relaxed responses at $\tilde S$, gives $\mathcal U_t\kh{\theta_t\mid\tilde S}\le\mathcal U_t\kh{\theta_t\mid S}$.

For property~\ref{item:constapp-columns}, let $S^-$ be obtained from $S$ as described there. The payment $p_{\theta_t}\kh{x,q}$ of each coordinate in the set is lower by $\sum_{\tau<t}\eta_\tau\kappa_\tau$ at $S^-$ than at $S$, and the payments of the other coordinates are equal, so $\underline u_{\theta_t}\kh{S}\le\underline u_{\theta_t}\kh{S^-}+\sum_{\tau<t}\eta_\tau\kappa_\tau$. Let $\kh{x',q',c'}$ be a relaxed response at $S^-$. Then $\kh{x',q'+\kappa,c'}$ is a relaxed response at $S$. Indeed, the revealed payoffs of $S^-$ are $u_\tau-\kappa_\tau$, so $q_\tau'-c'\le u_\tau-\kappa_\tau$, that is, $q_\tau'+\kappa_\tau-c'\le u_\tau$, for every $\tau<t$. Moreover, raising $q'$ by $\kappa$ raises the agent's payoff under $\theta_t$ by $\sum_{\tau<t}\eta_\tau\kappa_\tau$, to at least $\underline u_{\theta_t}\kh{S^-}+\sum_{\tau<t}\eta_\tau\kappa_\tau$, which is at least $\underline u_{\theta_t}\kh{S}$.

The state $\kh{S^-}_{\theta_t}^+\kh{x',q',c'}$ is obtained from $S_{\theta_t}^+\kh{x',q'+\kappa,c'}$ by a change of the kind in property~\ref{item:constapp-columns} at date $t+1$. In $\kh{S^-}_{\theta_t}^+\kh{x',q',c'}$, for every $\tau<t$, the revealed payoff $u_\tau$ is lower by $\kappa_\tau$, and so is the payment $q_\tau$ of each coordinate in the set and of the new coordinate for period $t$. The revealed payoff of period $t$, and the payments that these coordinates record for period $t$, are lower by $\sum_{\tau<t}\eta_\tau\kappa_\tau$. The change therefore uses $\kh{\kappa_1,\ldots,\kappa_{t-1},\sum_{\tau<t}\eta_\tau\kappa_\tau}$ in place of $\kappa$, with the new coordinate for period $t$ added to the set. Since $\kh{x',q',c'}$ is a relaxed response at $S^-$, the state $\kh{S^-}_{\theta_t}^+\kh{x',q',c'}$ is compatible, so property~\ref{item:constapp-columns} at date $t+1$ gives
$$\mathcal V_{t+1}\kh{\kh{S^-}_{\theta_t}^+\kh{x',q',c'}}\ge\mathcal V_{t+1}\kh{S_{\theta_t}^+\kh{x',q'+\kappa,c'}}.$$

The current payoff at $\kh{x',q',c'}$ exceeds that at $\kh{x',q'+\kappa,c'}$ by $\sum_{\tau<t}\eta_\tau\kappa_\tau\ge0$. Hence the objective for $S^-$ at $\kh{x',q',c'}$ is at least the objective for $S$ at $\kh{x',q'+\kappa,c'}$, and so at least $\mathcal U_t\kh{\theta_t\mid S}$. Taking the infimum over relaxed responses at $S^-$ gives $\mathcal U_t\kh{\theta_t\mid S^-}\ge\mathcal U_t\kh{\theta_t\mid S}$.

Properties~\ref{item:constapp-outputs}-\ref{item:constapp-columns} therefore hold for every $\mathcal U_t\kh{\theta_t\mid\cdot}$. They are preserved by the supremum over $\Theta_t^<$ and by the maximum in equation~\eqref{eqn:state-recursion}, so $\mathcal V_t$ has properties~\ref{item:constapp-convex}-\ref{item:constapp-columns} on states at date $t$. This completes the first induction.

We now verify properties~\ref{item:constapp-bound}-\ref{item:constapp-exact} by backward induction. The function $\hat V_{T+1}^*\equiv0$ satisfies them, property~\ref{item:constapp-exact} because $\mathcal V_{T+1}\equiv0$. Let $t\le T$, and suppose that $\hat V_{t+1}^*$ satisfies them. Since $\mathcal V_{t+1}$ has properties~\ref{item:constapp-convex}-\ref{item:constapp-columns} by the first induction, and $\hat V_{t+1}^*$ satisfies properties~\ref{item:constapp-info} and~\ref{item:constapp-exact}, the proofs of Lemmas~\ref{lem:envelopeconst} and~\ref{lem:improvementconst} apply at date $t$. The latter gives, for every feasible history $h^t$ and every contract $w_t$, some $\theta_t\in\Theta_t$ with $\hat U_t\kh{w_{\theta_t}\mid h^t}\ge\hat U_t\kh{w_t\mid h^t}$. Hence
$$\hat V_t^*\kh{h^t}=\sup_{\theta_t\in\Theta_t}\hat U_t\kh{w_{\theta_t}\mid h^t}.$$
Since $\hat V_{t+1}^*$ also satisfies properties~\ref{item:constapp-bound} and~\ref{item:constapp-usc}, the proof of Lemma~\ref{lem:existenceconst} shows that this supremum is attained and that $\hat V_t^*$ satisfies property~\ref{item:constapp-usc}. It remains to verify properties~\ref{item:constapp-bound}, \ref{item:constapp-info} and~\ref{item:constapp-exact} at date $t$.

For property~\ref{item:constapp-bound}, consider first the zero contract. Under it, the principal's current payoff is the expected output, which is nonnegative, and her continuation value is nonnegative by property~\ref{item:constapp-bound} at date $t+1$, so $\hat V_t^*\kh{h^t}\ge\hat U_t\kh{0\mid h^t}\ge0$. Under any contract, her current payoff from any action is at most $\max Y$, because payments are nonnegative, and her continuation value is at most $\kh{\max Y}\sum_{j=0}^{T-t-1}\beta^j$ by property~\ref{item:constapp-bound} at date $t+1$. Hence
$$\hat V_t^*\kh{h^t}\le\max Y+\beta\kh{\max Y}\sum_{j=0}^{T-t-1}\beta^j=\kh{\max Y}\sum_{j=0}^{T-t}\beta^j.$$

For property~\ref{item:constapp-info}, let $h^t$ and $\tilde h^t$ be feasible histories with $\Omega\kh{\tilde h^t}\subseteq\Omega\kh{h^t}$, and fix a contract $w_t$. Take $A_t\in\Omega\kh{\tilde h^t}$ and $a_t\in BR\kh{w_t\mid A_t}$. The technology $A_t$ also belongs to $\Omega\kh{h^t}$, and the set $BR\kh{w_t\mid A_t}$ does not depend on the history. Equation~\eqref{eqn:const-updating} gives
$$\Omega\kh{\tilde h^t,w_t,a_t}=\hkh{A'\in\Omega\kh{\tilde h^t}:a_t\in BR\kh{w_t\mid A'}}\subseteq\Omega\kh{h^t,w_t,a_t},$$
and both extended histories are feasible, because $A_t$ belongs to both families. The current payoff from $a_t$ does not depend on the history either, so property~\ref{item:constapp-info} at date $t+1$ gives $\hat U_t\kh{w_t\mid\tilde h^t,a_t}\ge\hat U_t\kh{w_t\mid h^t,a_t}$. Hence
\eqns{\hat U_t\kh{w_t\mid\tilde h^t}&=\inf_{A_t\in\Omega\kh{\tilde h^t}}\max_{a_t\in BR\kh{w_t\mid A_t}}\hat U_t\kh{w_t\mid\tilde h^t,a_t}\\&\ge\inf_{A_t\in\Omega\kh{\tilde h^t}}\max_{a_t\in BR\kh{w_t\mid A_t}}\hat U_t\kh{w_t\mid h^t,a_t}\ge\hat U_t\kh{w_t\mid h^t},}
where the last inequality holds because $\Omega\kh{\tilde h^t}\subseteq\Omega\kh{h^t}$. Taking the supremum over contracts gives $\hat V_t^*\kh{\tilde h^t}\ge\hat V_t^*\kh{h^t}$. If $\Omega\kh{\tilde h^t}=\Omega\kh{h^t}$, exchanging the roles of $h^t$ and $\tilde h^t$ gives equality.

For property~\ref{item:constapp-exact}, let $h^t$ be a feasible history with state $S$. If $t=T$, the proof of Lemma~\ref{lem:second} gives $\hat V_T^*\kh{h^T}=\hat\Phi\kh{h^T}^2=\mathcal V_T\kh{S}$, so let $t<T$. We first identify the first branch in equation~\eqref{eqn:state-recursion}. The state of $h^{t\to T}$ is obtained from $S$ by appending, for the silent periods, zero payments, coordinates with zero values and zero revealed payoffs. This does not change $\mathcal V_T$. Indeed, the payment and the revealed payoff of a silent period are zero, so the historical term of a coordinate for that period equals its static term, and every term of a coordinate with zero value is zero. Property~\ref{item:constapp-exact} at date $T$ therefore gives $\hat V_T^*\kh{h^{t\to T}}=\mathcal V_T\kh{S}$. Since $S^0$ is the state of $h^{t\to t+1}$, property~\ref{item:constapp-exact} at date $t+1$ gives $\hat V_{t+1}^*\kh{h^{t\to t+1}}=\mathcal V_{t+1}\kh{S^0}$. The right-hand sides of equations~\eqref{eqn:silent-guarantee} and~\eqref{eqn:no-strict-bound} are therefore both equal to the first branch, $\mathcal V_T\kh{S}+\beta\mathcal V_{t+1}\kh{S^0}$.

We now show that $\hat V_t^*\kh{h^t}\le\mathcal V_t\kh{S}$, by bounding the guarantee of every contract by one of the two branches. If a contract admits no strict new response, equation~\eqref{eqn:no-strict-bound} shows that it guarantees at most the first branch. Let $w_t$ instead admit a strict new response. We show that it guarantees at most the second branch. For every $\epsilon\in\kh{0,1}$, the proof of Lemma~\ref{lem:improvementconst} gives $\theta_t\in\Theta_t^<$ with $\mathcal U_t\kh{\theta_t\mid S}\ge\hat U_t\kh{w_t\mid h^t}-\epsilon\max Y$. Letting $\epsilon$ decrease to zero shows that the second branch, the supremum of $\mathcal U_t\kh{\theta_t\mid S}$ over $\theta_t\in\Theta_t^<$, is at least $\hat U_t\kh{w_t\mid h^t}$. Every contract thus guarantees at most $\mathcal V_t\kh{S}$, so $\hat V_t^*\kh{h^t}\le\mathcal V_t\kh{S}$.

Conversely, we show that $\hat V_t^*\kh{h^t}$ is at least each branch. For the first branch, equation~\eqref{eqn:silent-guarantee} shows that an optimal final-period contract $w_T$ at $h^{t\to T}$ guarantees at least the first branch. For the second branch, fix $\theta_t\in\Theta_t^<$, and let $a_t$ be a best response to $w_{\theta_t}$ in a compatible technology. As in the proof of Lemma~\ref{lem:improvementconst}, the triple $\kh{x\kh{a_t},q\kh{a_t},c\kh{a_t}}$ is a relaxed response at $S$. The state of the extended history $\kh{h^t,w_{\theta_t},a_t}$ is $S_{\theta_t}^+\kh{x\kh{a_t},q\kh{a_t},c\kh{a_t}}$, so, by property~\ref{item:constapp-exact} at date $t+1$, the objective in equation~\eqref{eqn:relaxed-value} at this triple is the interim payoff guarantee from $a_t$. This guarantee is therefore at least $\mathcal U_t\kh{\theta_t\mid S}$ for every best response in every compatible technology, so $\hat U_t\kh{w_{\theta_t}\mid h^t}\ge\mathcal U_t\kh{\theta_t\mid S}$. Taking the supremum over $\theta_t\in\Theta_t^<$ shows that $\hat V_t^*\kh{h^t}$ is at least the second branch, and hence $\hat V_t^*\kh{h^t}\ge\mathcal V_t\kh{S}$.

Properties~\ref{item:constapp-bound}-\ref{item:constapp-exact} therefore hold at date $t$, and, by backward induction, at every date. Hence the proofs of Lemmas~\ref{lem:envelopeconst}, \ref{lem:improvementconst} and~\ref{lem:existenceconst} and Corollary~\ref{cor:history-relative-concavification} apply at every date. In particular, at every feasible history $h^t$, the supremum over $\Theta_t$ that gives $\hat V_t^*\kh{h^t}$ is attained:
$$\hat V_t^*\kh{h^t}=\max_{\theta_t\in\Theta_t}\hat U_t\kh{w_{\theta_t}\mid h^t},$$
which is equation~\eqref{eqn:optimalclassconst}. If every previous contract is linear, $w_\tau\kh{y}=s_\tau y$, then every $w_{\theta_t}$ is linear with share $\alpha+\sum_{\tau<t}\eta_\tau s_\tau$, which lies in $\fkh{0,1}$ whenever every $s_\tau$ does. At the initial history, there are no previous contracts, and every $w_{\theta_1}$ is linear with share $\alpha\in\fkh{0,1}$. Offering a maximizing contract $w_{\theta_t}$ after every feasible history therefore gives an optimal policy, and this policy offers a linear contract after every feasible history consisting of linear contracts.

This completes the proof.
\end{proof}

\newpage

\renewcommand{\theequation}{\thesection.\arabic{equation}}
\setcounter{equation}{0}
\renewcommand{\theLem}{\thesection.\arabic{Lem}}
\setcounter{Lem}{0}
\renewcommand{\theDef}{\thesection.\arabic{Def}}
\setcounter{Def}{0}
\renewcommand{\theProp}{\thesection.\arabic{Prop}}
\setcounter{Prop}{0}

\section{Optimal First-Period Contract}\label{subsec:opt}
In this appendix, we examine the structure of the optimal linear first-period contract in our dynamic model and compare it with the optimal static contract identified by \cite{Carroll15}. This requires an exact calculation of the overall payoff guarantee of every linear first-period contract. When the principal knows a general set of actions, the highest payoff that agent~$1$ obtains from known actions under a linear contract is the upper envelope of linear functions of the share, one for each known action. This payoff enters the constraints of the programs below, and each kink of the envelope requires a separate case. We therefore consider the two-period model, $T=2$, and suppose that the principal knows only one available action $a_0=\kh{F_0,c_0}$ with $c_0>0$ besides the null action, so that $A_0=\hkh{a_0,\kh{\delta_0,0}}$. Let $s_0\equiv\sqrt{c_0/\mathbb{E}_{F_0}\fkh{y}}\in\kh{0,1}$ denote the optimal static share in \cite{Carroll15}, so that $\Phi_0=\kh{1-s_0}\sqrt{\mathbb{E}_{F_0}\fkh{y}}$.

By Theorems~\ref{prop:1grow} and~\ref{prop:1}, under both advancing and constant technology, the principal can restrict attention to linear first-period contracts $w_1\kh{y}=s_1y$ with $s_1\in\fkh{0,1}$. We write $U\kh{s_1}$ for the overall payoff guarantee $U_1\kh{s_1y\mid h^1}$ under advancing technology and $\hat U\kh{s_1}$ for $\hat U_1\kh{s_1y\mid h^1}$ under constant technology, and we call a share that maximizes the guarantee an optimal first-period share. If agent~$1$ chooses $a_1=\kh{F_1,c_1}$, write $\mathbb{E}_{F_1}\fkh{y}=x\mathbb{E}_{F_0}\fkh{y}$ and $c_1=z\mathbb{E}_{F_0}\fkh{y}$, where $0\le x\le\overline y\equiv\max Y/\mathbb{E}_{F_0}\fkh{y}$ and $z\ge0$.

\subsection{Advancing Technology}\label{app:opt-adv}
Under advancing technology, the final-period payoff guarantee after observing $a_1$ is $V_2^*\kh{K_2}=\Phi\kh{K_2}^2$, where $K_2=A_0\cup\hkh{a_1}$ and, by equation~\eqref{eqn:Phi}, $\Phi\kh{K_2}=\max\hkh{1-s_0,\sqrt x-\sqrt z}\sqrt{\mathbb{E}_{F_0}\fkh{y}}$. Equation~\eqref{eqn:adv-linear-guarantee} then reduces to the program
\eqn{U\kh{s_1}=\mathbb{E}_{F_0}\fkh{y}\cdot\min_{x,z}\hkh{\kh{1-s_1}x+\beta\max\hkh{1-s_0,\sqrt x-\sqrt z}^2}\label{eqn:proggrow}}
over $0\le x\le\overline y$ and $z\ge0$ subject to $s_1x-z\ge\max\hkh{0,s_1-s_0^2}$, as we show in Section~\ref{app:proofopt}. Figure~\ref{fig:Uhatgrow} depicts the resulting overall payoff guarantee. It has a unique maximum, which lies to the right of the optimal static share $s_0$.
\begin{figure}[!htbp]\centering
\begin{tikzpicture}[x=8cm,y=6.5cm,>=stealth,font=\small]
\draw[->] (0,0) to (1.15,0) node[right] {First-period share, $s_1$};
\draw[->] (0,0) to (0,0.95) node[above] {Principal's overall payoff guarantee, $U$};
\node[below left] at (0,0) {$0$};
\draw[densely dashed,thin] (0.16,0) node[below] {$s_0^2$} to (0.16,0.9);
\draw[densely dashed,thin,gray] (0.4,0) node[below,black] {$s_0$} to (0.4,0.9);
\draw[densely dashed,thin] (1,0) node[below] {$1$} to (1,0.9);
\draw[densely dashed,thin] (0,0.288) node[left] {$\beta\kh{\sqrt{\mathbb{E}_{F_0}\fkh{y}}-\sqrt{c_0}}^2$} to (1.05,0.288);
\draw[very thick] plot coordinates {
(0.0000,0.2880) (0.1600,0.2880) (0.1660,0.3181) (0.1720,0.3458) (0.1785,0.3731) (0.1850,0.3981) (0.1915,0.4210) (0.1985,0.4435)
(0.2060,0.4653) (0.2135,0.4851) (0.2250,0.5119) (0.2375,0.5368) (0.2455,0.5508) (0.2610,0.5756) (0.2730,0.5957) (0.2995,0.6439)
(0.3070,0.6567) (0.3140,0.6677) (0.3215,0.6784) (0.3290,0.6881) (0.3365,0.6969) (0.3440,0.7049) (0.3520,0.7124) (0.3600,0.7191)
(0.3685,0.7253) (0.3770,0.7307) (0.3855,0.7352) (0.3945,0.7392) (0.4035,0.7424) (0.4130,0.7450) (0.4225,0.7467) (0.4320,0.7478)
(0.4420,0.7482) (0.4515,0.7479) (0.4615,0.7469) (0.4815,0.7430) (0.5020,0.7367) (0.5235,0.7278) (0.5455,0.7164) (0.5685,0.7024)
(0.5920,0.6859) (0.6160,0.6671) (0.6420,0.6446) (0.6685,0.6195) (1.0000,0.2880)};
\fill (0.4420,0.7482) circle (1.6pt);
\node[fill=white,inner sep=1pt] (opt) at (0.64,0.39) {Optimal payoff guarantee};
\draw[->,shorten >=3pt] (opt.north) to (0.4420,0.7482);
\end{tikzpicture}
\caption{Overall payoff guarantee in the case of advancing technology ($s_0=0.4$, $\beta=0.8$).}\label{fig:Uhatgrow}
\end{figure}
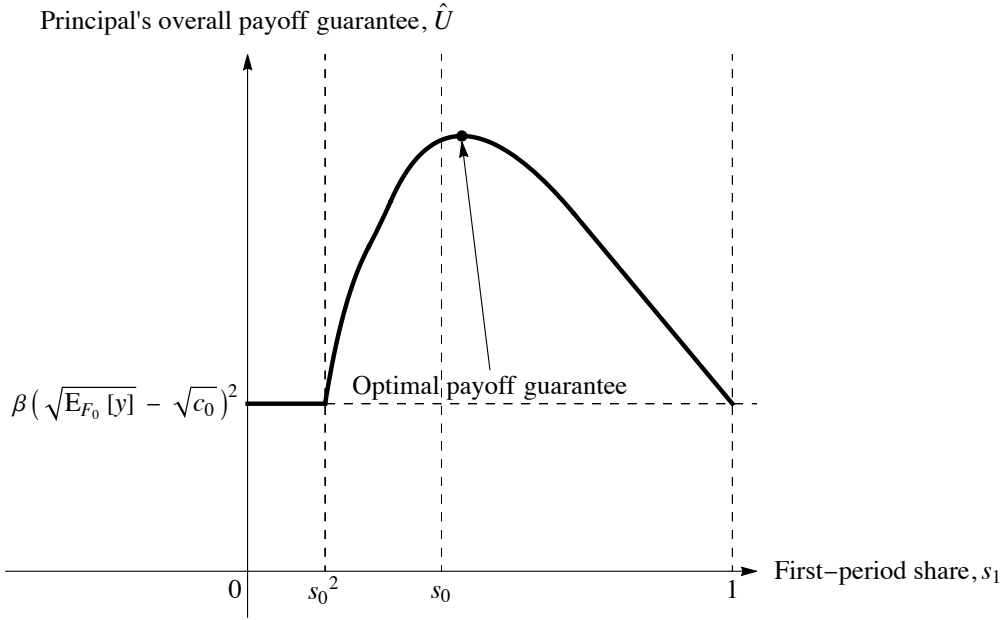

Proposition~\ref{prop:exp} formally establishes this observation. It reveals an \emph{exploration effect}: the optimal first-period share offered to agent~$1$ is always larger than the optimal static share $s_0$. Moreover, the exploration effect increases as the principal becomes more patient ($\beta$ increases), provided that $\beta<1$. When $\beta>1$, it starts to decrease, and it vanishes as $\beta\to\infty$.
\begin{Prop}\label{prop:exp}
Suppose that $T=2$ and the principal knows only one available action $a_0=\kh{F_0,c_0}$ besides the null action. In the case of advancing technology, the optimal first-period share $s_1^*$ is unique and satisfies the following properties:
\begin{enumerate}
\item For all $\beta\in\kh{0,\infty}$, the optimal first-period share is larger than the optimal static share, i.e., $s_1^*>s_0$.
\item In both limiting cases $\beta\to0$ and $\beta\to\infty$, $s_1^*$ approaches $s_0$.
\item $s_1^*$ is strictly increasing in $\beta$ on $(0,1]$ and strictly decreasing on $[1,\infty)$.
\end{enumerate}
\end{Prop}
\begin{proof}
All proofs for this appendix are in Section~\ref{app:proofopt}.
\end{proof}
The pattern identified by Proposition~\ref{prop:exp} is illustrated in Figure~\ref{fig:alpha1snorm}, which plots $s_1^*$ against the normalized discount factor $\beta/\kh{1+\beta}$. The curve is symmetric, because $s_1^*$ takes the same value at $\beta$ and $1/\beta$, as we explain at the end of Section~\ref{app:opt-const}.
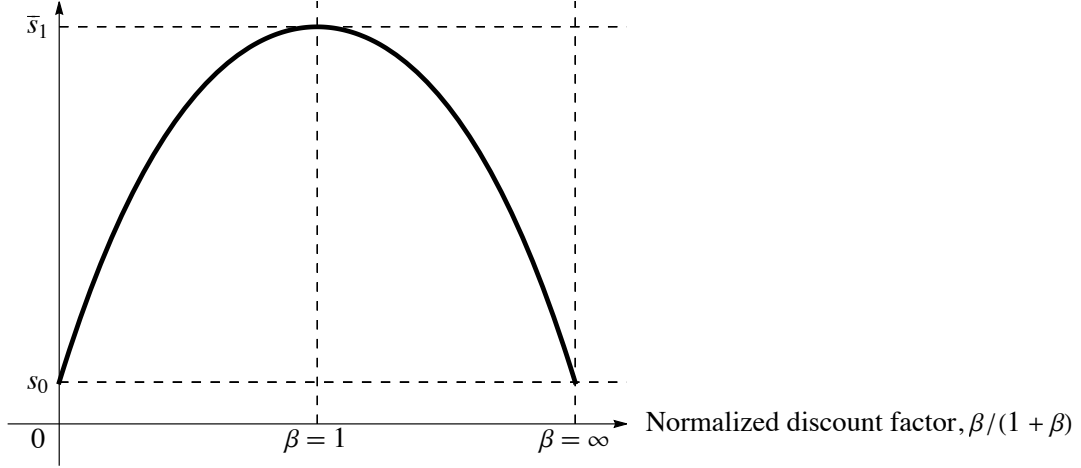
\begin{figure}[!htbp]\centering
\begin{tikzpicture}[x=7.3cm,y=80cm,>=stealth,font=\small,declare function={ra(\t)=(sqrt(\t*\t+4*(0.16-\t*\t)/(1-\t))-\t)/2;}]
\draw[->] (0,0) to (1.15,0) node[right] {Normalized discount factor, $\beta/\kh{1+\beta}$};
\draw[->] (0,0) to (0,0.0625) node[above] {Optimal first-period share, $s_1^*$};
\node[below] at (0,0) {$0$};
\fill[white] ([xshift=-4pt,yshift=7pt]0,0) rectangle ([xshift=4pt,yshift=11pt]0,0);
\draw ([xshift=-4pt,yshift=5pt]0,0) to ([xshift=4pt,yshift=9pt]0,0);
\draw ([xshift=-4pt,yshift=9pt]0,0) to ([xshift=4pt,yshift=13pt]0,0);
\draw[densely dashed,thin] (0,0.01) node[left] {$s_0$} to (1.05,0.01);
\draw[densely dashed,thin] (0,0.0525) node[left] {$\left.s_1^*\right|_{\beta=1}$} to (1.05,0.0525);
\draw[densely dashed,thin] (0.5,0) node[below] {$\beta=1$} to (0.5,0.0525);
\draw[densely dashed,thin] (1,0) node[below] {$\beta=\infty$} to (1,0.0525);
\draw[very thick] plot[domain=0:0.4,samples=161,variable=\t] ({\t*(1-ra(\t))/(\t*(1-ra(\t))+ra(\t)*(1-\t))},{ra(\t)+\t-ra(\t)*\t-0.39}) to (1,0.01);
\end{tikzpicture}
\caption{The optimal first-period share $s_1^*$ in the case of advancing technology ($s_0=0.4$).}\label{fig:alpha1snorm}
\end{figure}

It is straightforward to understand the result that the dynamic model converges to the static model as the discount factor $\beta$ approaches $0$, since the first-period payoff then dominates. To get intuition behind the opposite case, in which the optimal first-period share $s_1^*$ again approaches the optimal static share $s_0$ as $\beta$ approaches infinity, note that unlike in standard models where patience automatically leads to the option value of exploration, here the principal is concerned with the worst-case discovery. For every first-period share $s_1$, the final-period guarantee is at least the static guarantee $\kh{1-s_0}^2\mathbb{E}_{F_0}\fkh{y}$, and nature can hold it to this bound by choosing the technology $A_0$, which gives the principal at most $\kh{1-s_1}\mathbb{E}_{F_0}\fkh{y}$ in the first period. Hence $U\kh{s_1}/\beta$ converges to the static guarantee for every $s_1$ as $\beta\to\infty$. In the limit, the worst-case technology leaves the principal without any valuable discovery, and the opportunity to explore has no value to her. The optimal share then maximizes her first-period payoff against responses that do not raise her final-period guarantee. At the share $s_0$, agent~$1$ strictly prefers $a_0$ to every such action with a different expected output or cost, so the principal obtains $\kh{1-s_0}\mathbb{E}_{F_0}\fkh{y}$ in the first period, and every other share yields less.

Lemma~\ref{lem:opt-adv} in Section~\ref{app:proofopt} characterizes the optimal share exactly and identifies the worst case at it. The worst case is approached by diluted versions of $a_0$: their output distributions mix $F_0$ with $\delta_0$, and their costs fall below $c_0$ by just enough that agent~$1$ is nearly indifferent between the new action and $a_0$. As $\beta\to\infty$, the weight on $\delta_0$ vanishes, so the worst-case discovery converges to $a_0$ itself.

\subsection{Constant Technology}\label{app:opt-const}
Under constant technology, the principal adopts the more subtle rule of updating given by compatibility, and Lemma~\ref{lem:second} gives the final-period payoff guarantee $\hat V_2^*\kh{h^2}=\hat\Phi\kh{h^2}^2$. Write $g\equiv s_1x-z-\kh{s_1-s_0^2}$, so that $g\mathbb{E}_{F_0}\fkh{y}$ is the incentive gap $g_1\kh{F_0,c_0}$ of $a_0$ in Definition~\ref{def:g}, and recall the function $\xi\kh{r,g}=\kh{\sqrt{r_+}-\sqrt g}_+^2$ from Section~\ref{app:const}. By Lemma~\ref{lem:second}, $\hat\Phi\kh{h^2}^2/\mathbb{E}_{F_0}\fkh{y}$ is the largest of four terms: offering $w_1$ again gives $\kh{1-s_1}x$, the modified $w_1$ that induces $a_0$ gives $\xi\kh{1-s_1,g}$, and the linear contracts that correspond to $a_0$ and $a_1$ give $\kh{1-s_0}^2$ and $\xi\kh{x,z}$, respectively. The other terms in equation~\eqref{eqn:optsecond}, which involve the null action or the modified $w_1$ that induces $a_1$, do not exceed these. The analogue of program~\eqref{eqn:proggrow} is
\eqn{\hat U\kh{s_1}=\mathbb{E}_{F_0}\fkh{y}\cdot\min_{x,z}\hkh{\kh{1-s_1}x+\beta\max\hkh{\kh{1-s_1}x,\xi\kh{1-s_1,g},\kh{1-s_0}^2,\xi\kh{x,z}}}\label{eqn:prog}}
over $0\le x\le\overline y$ and $z\ge0$ subject to $s_1x-z\ge0$ and $g\ge0$, as we show in Section~\ref{app:proofopt}. Figure~\ref{fig:UPalpha1} depicts the resulting guarantee. It has a unique maximum, which lies to the left of the optimal static share $s_0$.
\begin{figure}[!htbp]\centering
\begin{tikzpicture}[x=8cm,y=6.5cm,>=stealth,font=\small]
\draw[->] (0,0) to (1.15,0) node[right] {First-period share, $s_1$};
\draw[->] (0,0) to (0,0.95) node[above] {Principal's overall payoff guarantee, $\hat U$};
\node[below left] at (0,0) {$0$};
\draw[densely dashed,thin] (0.16,0) node[below] {$s_0^2$} to (0.16,0.9);
\draw[densely dashed,thin,gray] (0.4,0) node[below,black] {$s_0$} to (0.4,0.9);
\draw[densely dashed,thin] (1,0) node[below] {$1$} to (1,0.9);
\draw[densely dashed,thin] (0,0.288) node[left] {$\beta\kh{\sqrt{\mathbb{E}_{F_0}\fkh{y}}-\sqrt{c_0}}^2$} to (1.05,0.288);
\draw[very thick] plot coordinates {
(0.0000,0.2880) (0.0180,0.3017) (0.0345,0.3158) (0.0500,0.3308) (0.0645,0.3466) (0.0780,0.3633) (0.0900,0.3802) (0.1010,0.3979)
(0.1110,0.4165) (0.1185,0.4324) (0.1255,0.4493) (0.1315,0.4659) (0.1370,0.4834) (0.1440,0.5104) (0.1535,0.5556) (0.1605,0.5851)
(0.1680,0.6127) (0.1760,0.6379) (0.1850,0.6619) (0.1940,0.6819) (0.2035,0.6993) (0.2135,0.7140) (0.2195,0.7213) (0.2255,0.7275)
(0.2315,0.7328) (0.2380,0.7375) (0.2440,0.7409) (0.2555,0.7465) (0.2780,0.7560) (0.3015,0.7637) (0.3250,0.7693) (0.3480,0.7725)
(0.3715,0.7736) (0.3850,0.7732) (0.3985,0.7722) (0.4120,0.7704) (0.4260,0.7678) (0.4400,0.7645) (0.4540,0.7605) (0.4680,0.7558)
(0.4820,0.7504) (0.4960,0.7444) (0.5105,0.7374) (0.5395,0.7214) (0.5690,0.7025) (0.5990,0.6807) (0.6155,0.6675) (0.6320,0.6535)
(0.6485,0.6386) (0.6655,0.6224) (1.0000,0.2880)};
\fill (0.3715,0.7736) circle (1.6pt);
\node[fill=white,inner sep=1pt] (opt) at (0.58,0.39) {Optimal payoff guarantee};
\draw[->,shorten >=3pt] (opt.north) to (0.3715,0.7736);
\end{tikzpicture}
\caption{Overall payoff guarantee in the case of constant technology ($s_0=0.4$, $\beta=0.8$).}\label{fig:UPalpha1}
\end{figure}
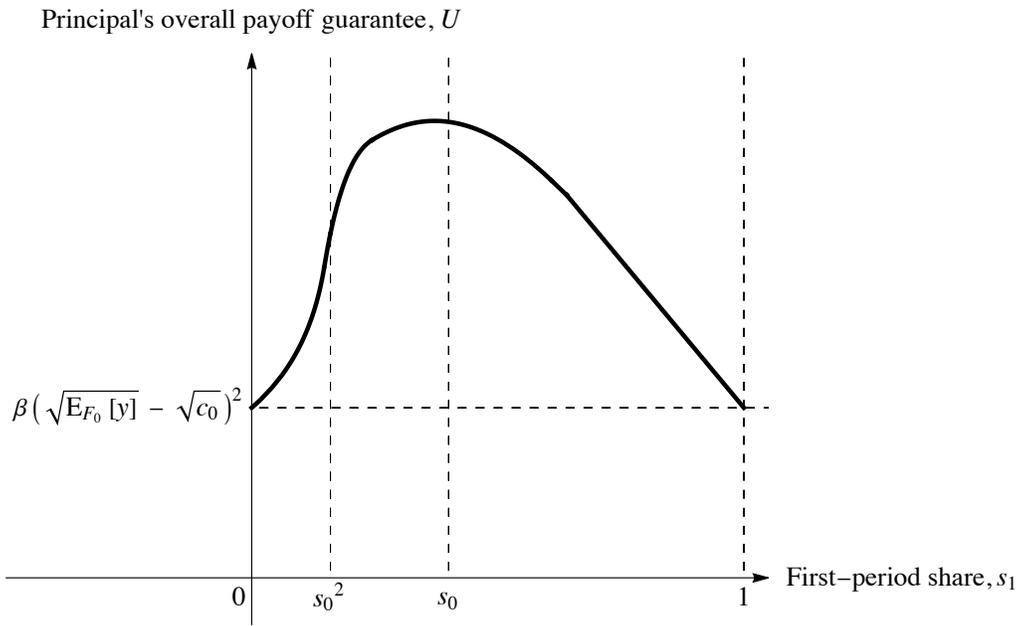

Proposition~\ref{prop:expconst} shows that the optimal first-period share lies below the optimal static share for all parameter values.
\begin{Prop}\label{prop:expconst}
Suppose that $T=2$ and the principal knows only one available action $a_0=\kh{F_0,c_0}$ besides the null action. In the case of constant technology, the optimal first-period share $\hat s_1^*$ is unique and satisfies the following properties:
\begin{enumerate}
\item For all $\beta\in\kh{0,\infty}$, the optimal first-period share is smaller than the optimal static share, i.e., $\hat s_1^*<s_0$.
\item As $\beta\to0$, $\hat s_1^*$ approaches $s_0$.
\item $\hat s_1^*$ is strictly decreasing in $\beta$ on $(0,2]$ and constant on $[2,\infty)$.
\end{enumerate}
\end{Prop}
Under constant technology, the exploration effect is therefore negative: the principal lowers the share offered to agent~$1$ below the optimal static share. Unlike under advancing technology, the effect does not vanish as the principal becomes patient. For all $\beta\ge2$, the optimal first-period share does not depend on $\beta$, and the threshold $\beta=2$ does not depend on $s_0$. In particular, $\hat s_1^*$ does not approach $s_0$ as $\beta\to\infty$. Figure~\ref{fig:s1sconst} illustrates the pattern.
\begin{figure}[!htbp]\centering
\begin{tikzpicture}[x=7.3cm,y=60cm,>=stealth,font=\small,declare function={sc(\t)=\t*\t+sqrt(0.16-2*\t*\t+\t*\t*\t*\t);}]
\draw[->] (0,0) to (1.15,0) node[right] {Normalized discount factor, $\beta/\kh{1+\beta}$};
\draw[->] (0,0) to (0,0.0833) node[above] {Optimal first-period share, $\hat s_1^*$};
\node[below] at (0,0) {$0$};
\fill[white] ([xshift=-4pt,yshift=7pt]0,0) rectangle ([xshift=4pt,yshift=11pt]0,0);
\draw ([xshift=-4pt,yshift=5pt]0,0) to ([xshift=4pt,yshift=9pt]0,0);
\draw ([xshift=-4pt,yshift=9pt]0,0) to ([xshift=4pt,yshift=13pt]0,0);
\draw[densely dashed,thin] (0,0.075) node[left] {$s_0$} to (1.05,0.075);
\draw[densely dashed,thin] (0,0.01275) node[left] {$\left.\hat s_1^*\right|_{\beta\ge2}$} to (1.05,0.01275);
\draw[densely dashed,thin] (0.6667,0) node[below] {$\beta=2$} to (0.6667,0.075);
\draw[densely dashed,thin] (1,0) node[below] {$\beta=\infty$} to (1,0.075);
\draw[very thick] plot[domain=0:0.18621,samples=121,variable=\t] ({2*(1-sc(\t))*\t/(2*(1-sc(\t))*\t+(1-\t)*(sc(\t)-\t))},{sc(\t)-0.325}) to (0.6667,0.01275) to (1,0.01275);
\end{tikzpicture}
\caption{The optimal first-period share $\hat s_1^*$ in the case of constant technology ($s_0=0.4$).}\label{fig:s1sconst}
\end{figure}

Lemma~\ref{lem:opt-const} in Section~\ref{app:proofopt} characterizes the optimal share exactly and identifies the worst case at it. In the worst case, agent~$1$ chooses a diluted version of $a_0$, with output distribution $\kh{1-\hat s_1^*}F_0+\hat s_1^*\delta_0$ and cost equal to the incentive gap of $a_0$. After observing this response, the principal is indifferent between the modified $w_1$ that induces $a_0$ (Part~\ref{part:2} of Lemma~\ref{lem:second}) and the linear contract that corresponds to the observed action (Part~\ref{part:3}). The first adds compensation to $w_1$, so its share exceeds $\hat s_1^*$, and the second has a share below $\hat s_1^*$. The lower first-period share is therefore not explained by an anticipated higher final-period share. It balances two continuation guarantees generated by the same worst-case discovery.

The difference from advancing technology comes from the inference that the principal draws in the final period. In the static problem, the guarantee against the share $s_0$ is approached by diluted versions of $a_0$ that leave agent~$1$ nearly indifferent between the new action and $a_0$. Under constant technology, observing such a response reveals that $a_0$ is nearly a best response to $w_1$, so the principal can induce $a_0$ in the final period with a modified $w_1$ that adds only a small compensation, and obtain nearly $\kh{1-s_0}\mathbb{E}_{F_0}\fkh{y}$, more than the static guarantee. To prevent this, nature must make agent~$1$ prefer the new action to $a_0$ by a sizable margin, and the resulting incentive gap of $a_0$ determines how cheaply the principal can induce $a_0$ in the final period with a modified $w_1$. The first-period share shapes this trade-off for nature, and Proposition~\ref{prop:expconst} shows that the principal resolves it by lowering the share below $s_0$, at which point the worst-case response mirrors her residual share $1-\hat s_1^*$.

The special values of $\beta$ have different interpretations under the two updating rules. Under advancing technology, the worst-case response at the optimal share hurts the principal in two ways: it dilutes $a_0$, which lowers her current output, and it has a positive cost, which makes it a less valuable discovery. The first is measured by the weight that the response places on $\delta_0$, and the second by the optimal static share for the discovered action, the square root of the ratio of its cost to its expected output. By Lemma~\ref{lem:opt-adv}, replacing $\beta$ by $1/\beta$ exchanges these two measures and leaves the optimal first-period share unchanged, so they are equal exactly when $\beta=1$. When $\beta<1$, the weight on $\delta_0$ is the larger, and nature places relatively more of the adverse distortion on current output; when $\beta>1$, the static share is the larger, and nature places relatively more on the future value of the discovered action. Together with the monotonicity in Proposition~\ref{prop:exp}, this symmetry explains why the exploration effect is largest when the two periods receive equal weight.

Under constant technology, $\beta=2$ instead marks a change in which final-period guarantees bind. Because the worst-case response mirrors the residual share $1-\hat s_1^*$ and its cost equals the incentive gap of $a_0$, the modified $w_1$ that induces $a_0$ and the linear contract that corresponds to the observed action give the same final-period guarantee. Offering $w_1$ again (Part~\ref{part:1} of Lemma~\ref{lem:second}) guarantees the first-period payoff, so the first-period payoff is a floor on the final-period guarantee. For $\beta<2$, this floor lies strictly below the other two guarantees, which bind. From $\beta=2$ on, the floor binds as well and all three coincide, so the value is $\kh{1+\beta}\kh{1-\hat s_1^*}^2\mathbb{E}_{F_0}\fkh{y}$, and additional weight on the final period does not change the optimal first-period share. The threshold is $2$ because, below it, the two equal guarantees share the weight $\beta$ on the final period, and the floor binds once the weight on each, $\beta/2$, reaches the weight one on the first-period payoff.

\subsection{Proofs for Appendix~\ref{subsec:opt}}\label{app:proofopt}
We call a first-period action $a_1$ a \emph{weak response} to $s_1y$ if agent~$1$ weakly prefers it to both known actions, that is, if $s_1x-z\ge\max\hkh{0,s_1-s_0^2}$, or equivalently, if $s_1x-z\ge0$ and $g\ge0$. Since $Y$ is compact and contains $0$, every $x\in\fkh{0,\overline y}$ is the normalized expected output of the distribution $\kh{x/\overline y}\delta_{\max Y}+\kh{1-x/\overline y}\delta_0$. If $x\le1$, the distribution $xF_0+\kh{1-x}\delta_0$ also has normalized expected output $x$, and the worst-case responses in Lemmas~\ref{lem:opt-adv} and~\ref{lem:opt-const} take this form. Hence the feasible points of programs~\eqref{eqn:proggrow} and~\eqref{eqn:prog} correspond exactly to weak responses. By equation~\eqref{eqn:Phi} and Lemma~\ref{lem:second}, the expressions in braces, which we call the objectives, are the principal's interim payoff guarantees from these responses divided by $\mathbb{E}_{F_0}\fkh{y}$. The feasible sets are compact, since the constraints give $z\le s_1x\le\overline y$, and the objectives are continuous, so both programs have minimizers.

Every best response to $s_1y$ under a technology containing $A_0$ is a weak response, so $U\kh{s_1}$ and $\hat U\kh{s_1}$ are at least the right-hand sides of equations~\eqref{eqn:proggrow} and~\eqref{eqn:prog}. For the reverse inequalities, suppose first that $s_1>0$, and let $a_1$ be a weak response. If agent~$1$ strictly prefers $a_1$ to both known actions, then $a_1$ is his unique best response under the technology $A_0\cup\hkh{a_1}$, so the overall payoff guarantee from $s_1y$ is at most the interim payoff guarantee from $a_1$. Otherwise, agent~$1$ is indifferent between $a_1$ and a known action, so $s_1x-z=\max\hkh{0,s_1-s_0^2}$. If $z>0$, lowering the cost of $a_1$ makes it strictly preferred. If $z=0$, then $s_1x=\max\hkh{0,s_1-s_0^2}<s_1$, so $x<1\le\overline y$, and raising its expected output slightly makes it strictly preferred. Every weak response is therefore a limit of strictly preferred actions, and since the objectives are continuous, $U\kh{s_1}$ and $\hat U\kh{s_1}$ are at most the right-hand sides. If $s_1=0$, nature can choose the technology $A_0$, under which agent~$1$ takes the null action, so both guarantees are at most $\beta\kh{1-s_0}^2\mathbb{E}_{F_0}\fkh{y}$. Since both objectives are at least $\beta\kh{1-s_0}^2$, with equality at $\kh{x,z}=\kh{0,0}$, this is the value of both right-hand sides. Hence equations~\eqref{eqn:proggrow} and~\eqref{eqn:prog} hold for every $s_1\in\fkh{0,1}$.

The proofs also use that $\xi$ is jointly convex (see the proof of Theorem~\ref{prop:1}) and nonincreasing in its second argument. It is differentiable where $r>g>0$, with gradient $\kh{1-\sqrt{g/r},1-\sqrt{r/g}}$.

\begin{Lem}\label{lem:opt-adv}
For every $\beta>0$, there is a unique pair $\kh{r_\beta,t_\beta}\in\kh{0,1}^2$ satisfying
\eqn{s_0^2&=r_\beta^2+t_\beta^2+r_\beta t_\beta-r_\beta^2t_\beta-r_\beta t_\beta^2,\label{eqn:opt-adv-resource}\\\beta&=\frac{t_\beta\kh{1-r_\beta}}{r_\beta\kh{1-t_\beta}}.\label{eqn:opt-adv-beta}}
Under advancing technology, the optimal first-period share is unique and equals $s_1^*=r_\beta+t_\beta-r_\beta t_\beta$, and the principal's optimal payoff guarantee is $$V_1^*\kh{h^1}=\kh{\kh{1-r_\beta}^2\kh{1-t_\beta}+\beta\kh{1-r_\beta}\kh{1-t_\beta}^2}\mathbb{E}_{F_0}\fkh{y}.$$ At $s_1^*$, the worst case is approached by responses $a_1=\kh{F_1,c_1}$ of agent~$1$ with $F_1\to\kh{1-r_\beta}F_0+r_\beta\delta_0$ and $c_1\to\kh{1-r_\beta}t_\beta^2\mathbb{E}_{F_0}\fkh{y}$.
\end{Lem}

\begin{proof}[Proof of Lemma~\ref{lem:opt-adv}]
Since $1-s_0>0$, the objective of program~\eqref{eqn:proggrow} equals $$\kh{1-s_1}x+\beta\max\hkh{\kh{1-s_0}^2,\xi\kh{x,z}}.$$

We first show that the pair $\kh{r_\beta,t_\beta}$ is well defined. For $\kh{r,t}\in\kh{0,1}^2$, the function $r^2+t^2+rt-r^2t-rt^2$ has positive partial derivatives $\kh{1-t}\kh{2r+t}$ and $\kh{1-r}\kh{2t+r}$, and it equals $t^2$ at $r=0$ and one at $r=1$. Hence, for $t\in\kh{0,1}$, it equals $s_0^2$ at some $r\in\kh{0,1}$ if and only if $t<s_0$, and then at exactly one $r$. Along the resulting curve, $r$ is a continuous and strictly decreasing function of $t\in\kh{0,s_0}$, which converges to $s_0$ as $t\downarrow0$ and to zero as $t\uparrow s_0$. The ratio $t\kh{1-r}/\kh{r\kh{1-t}}$ is therefore continuous and strictly increasing in $t$ along the curve, and it converges to zero as $t\downarrow0$ and to infinity as $t\uparrow s_0$. Hence, for every $\beta>0$, exactly one point of the curve satisfies equation~\eqref{eqn:opt-adv-beta}.

Since $1-s_1^*=\kh{1-r_\beta}\kh{1-t_\beta}$, equation~\eqref{eqn:opt-adv-resource} is equivalent to
\eqn{s_0^2=1-\kh{1-s_1^*}\kh{1+r_\beta+t_\beta},\label{eqn:opt-adv-identity}}
which gives $\kh{s_1^*}^2-s_0^2=\kh{1-s_1^*}\kh{r_\beta+t_\beta-s_1^*}=\kh{1-s_1^*}r_\beta t_\beta>0$, so $s_1^*>s_0$. Equation~\eqref{eqn:opt-adv-resource} also gives $s_0^2-s_1^*r_\beta=\kh{1-r_\beta}t_\beta^2$ and $$s_0^2-\kh{1-\kh{1-t_\beta}\sqrt{1-r_\beta}}^2=\kh{1-t_\beta}\sqrt{1-r_\beta}\kh{1-\sqrt{1-r_\beta}}^2\kh{2+\sqrt{1-r_\beta}}>0,$$ so $\kh{1-r_\beta}\kh{1-t_\beta}^2>\kh{1-s_0}^2$.

We next show that, at $s_1=s_1^*$, the objective is at least $$\kh{1-r_\beta}^2\kh{1-t_\beta}+\beta\kh{1-r_\beta}\kh{1-t_\beta}^2$$ at every feasible point $\kh{x,z}$. Since $s_1^*>s_0>s_0^2$, feasibility gives $z\le s_1^*x-s_1^*+s_0^2$, so, because $\xi$ is nonincreasing in its second argument, the objective is at least $$\kh{1-s_1^*}x+\beta\xi\kh{x,s_1^*x-s_1^*+s_0^2}.$$ This is a convex function of $x$ on the interval where $s_1^*x-s_1^*+s_0^2\ge0$, which contains every feasible $x$. At $x=1-r_\beta$, the second argument of $\xi$ equals $s_0^2-s_1^*r_\beta=\kh{1-r_\beta}t_\beta^2$, which is positive and smaller than the first argument, so $\xi$ is differentiable there and equals $\kh{1-r_\beta}\kh{1-t_\beta}^2$. By the gradient formula, the derivative of $\xi\kh{x,s_1^*x-s_1^*+s_0^2}$ with respect to $x$ at this point is $\kh{1-t_\beta}+s_1^*\kh{1-1/t_\beta}=\kh{1-t_\beta}\kh{1-s_1^*/t_\beta}$. Since $s_1^*-t_\beta=r_\beta\kh{1-t_\beta}$, this equals $-r_\beta\kh{1-t_\beta}^2/t_\beta$, and by equation~\eqref{eqn:opt-adv-beta}, $\beta r_\beta\kh{1-t_\beta}^2/t_\beta=\kh{1-r_\beta}\kh{1-t_\beta}=1-s_1^*$. Hence the derivative of the convex function is zero at $x=1-r_\beta$, so the function is minimized there, where it equals $\kh{1-r_\beta}^2\kh{1-t_\beta}+\beta\kh{1-r_\beta}\kh{1-t_\beta}^2$.

We now bound the minimum of the objective from above at every $s_1$, using feasible points. If $s_1<s_0^2$, the point $\kh{0,0}$ gives $\beta\kh{1-s_0}^2$. If $s_1\ge s_0^2$, the point $\kh{1,s_0^2}$ gives $1-s_1+\beta\kh{1-s_0}^2$. Since $\kh{1-r_\beta}\kh{1-t_\beta}^2>\kh{1-s_0}^2$, these values are strictly below $\kh{1-r_\beta}^2\kh{1-t_\beta}+\beta\kh{1-r_\beta}\kh{1-t_\beta}^2$ if $s_1<s_0^2$ or $1-s_1\le\kh{1-r_\beta}^2\kh{1-t_\beta}$. Suppose instead that $s_1\ge s_0^2$ and $1-s_1>\kh{1-r_\beta}^2\kh{1-t_\beta}$, and let $$x=\frac{\kh{1-r_\beta}^2\kh{1-t_\beta}}{1-s_1}\in\kh{0,1},\quad z=s_1x-s_1+s_0^2.$$ Since $\kh{1-s_1}x=\kh{1-s_1^*}\kh{1-r_\beta}$ and $\kh{1-r_\beta}\kh{1-t_\beta}^2=\kh{1-s_1^*}\kh{1-t_\beta}$, whose product is $\kh{1-s_1^*}^3$, equation~\eqref{eqn:opt-adv-identity} gives
\eqn{z-\kh{\sqrt x-\kh{1-t_\beta}\sqrt{1-r_\beta}}^2&=\kh{1-s_1}-3\kh{1-s_1^*}+\frac{2\kh{1-s_1^*}^{3/2}}{\sqrt{1-s_1}}\nonumber\\&=\frac{\kh{\sqrt{1-s_1}-\sqrt{1-s_1^*}}^2\kh{\sqrt{1-s_1}+2\sqrt{1-s_1^*}}}{\sqrt{1-s_1}}\ge0.\label{eqn:opt-adv-slack}}
Hence $z\ge0$, and since $x<1$ and $s_1x-z=s_1-s_0^2\ge0$, the point $\kh{x,z}$ is feasible. Moreover, $\sqrt z\ge\sqrt x-\kh{1-t_\beta}\sqrt{1-r_\beta}$, so $\xi\kh{x,z}\le\kh{1-r_\beta}\kh{1-t_\beta}^2$. Since $\kh{1-s_1}x=\kh{1-r_\beta}^2\kh{1-t_\beta}$ and $\kh{1-s_0}^2<\kh{1-r_\beta}\kh{1-t_\beta}^2$, the objective at $\kh{x,z}$ is at most $\kh{1-r_\beta}^2\kh{1-t_\beta}+\beta\kh{1-r_\beta}\kh{1-t_\beta}^2$. If $s_1\ne s_1^*$, the inequality in equation~\eqref{eqn:opt-adv-slack} is strict, so $\xi\kh{x,z}<\kh{1-r_\beta}\kh{1-t_\beta}^2$, and the objective at $\kh{x,z}$ is strictly smaller.

At $s_1=s_1^*$, the point is $\kh{1-r_\beta,\kh{1-r_\beta}t_\beta^2}$, which minimizes the objective by the lower bound above. It corresponds to the response with output distribution $\kh{1-r_\beta}F_0+r_\beta\delta_0$ and cost $\kh{1-r_\beta}t_\beta^2\mathbb{E}_{F_0}\fkh{y}$, and agent~$1$ is indifferent between this response and $a_0$. Moving a little weight from $\delta_0$ to $F_0$ makes the response strictly preferred, so, by continuity of the objective, the worst case is approached as stated. Therefore $s_1^*$ is the unique maximizer of $U$, and Theorem~\ref{prop:1grow} gives the stated value of $V_1^*\kh{h^1}$.
\end{proof}

\begin{proof}[Proof of Proposition~\ref{prop:exp}]
By Lemma~\ref{lem:opt-adv}, the optimal first-period share is unique and equals $s_1^*=r_\beta+t_\beta-r_\beta t_\beta$, and equation~\eqref{eqn:opt-adv-identity} gives $\kh{s_1^*}^2-s_0^2=\kh{1-s_1^*}r_\beta t_\beta>0$, so $s_1^*>s_0$. Along the curve in the proof of the lemma, the ratio $t\kh{1-r}/\kh{r\kh{1-t}}$ is continuous and strictly increasing in $t$, and $r$ converges to $s_0$ as $t\downarrow0$ and to zero as $t\uparrow s_0$. Hence $t_\beta$ is continuous and strictly increasing in $\beta$, with $\kh{r_\beta,t_\beta}\to\kh{s_0,0}$ as $\beta\to0$ and $\kh{r_\beta,t_\beta}\to\kh{0,s_0}$ as $\beta\to\infty$, so $s_1^*\to s_0$ in both limits. Along the curve, $$\frac{d}{dt}\kh{r+t-rt}=\kh{1-t}\frac{dr}{dt}+1-r=\frac{\kh{1-r}\kh{r-t}}{2r+t},$$ so $s_1^*$ increases in $t_\beta$ where $r_\beta>t_\beta$ and decreases where $r_\beta<t_\beta$. By equation~\eqref{eqn:opt-adv-beta}, $r_\beta>t_\beta$ if and only if $\beta<1$. Hence $s_1^*$ is strictly increasing in $\beta$ if $\beta<1$ and strictly decreasing if $\beta>1$.
\end{proof}

\begin{Lem}\label{lem:opt-const}
For every $\beta>0$, there is a unique pair $\kh{\hat s_1^*,t_\beta}\in\kh{0,1}^2$ satisfying
\eqn{s_0^2=\kh{\hat s_1^*}^2+2\kh{1-\hat s_1^*}t_\beta^2\label{eqn:opt-const-resource}}
and
\eqn{\Brace{&\beta=\frac{2\kh{1-\hat s_1^*}t_\beta}{\kh{1-t_\beta}\kh{\hat s_1^*-t_\beta}},&&\text{if }\beta<2,\\&\hat s_1^*=2t_\beta-t_\beta^2,&&\text{if }\beta\ge2.}\label{eqn:opt-const-beta}}
Under constant technology, the optimal first-period share is unique and equals $\hat s_1^*$, and the principal's optimal payoff guarantee is $$\hat V_1^*\kh{h^1}=\kh{\kh{1-\hat s_1^*}^2+\beta\kh{1-\hat s_1^*}\kh{1-t_\beta}^2}\mathbb{E}_{F_0}\fkh{y}.$$ At $\hat s_1^*$, the worst case is attained by the technology $A_0\cup\hkh{a_1}$, where $a_1$ has output distribution $\kh{1-\hat s_1^*}F_0+\hat s_1^*\delta_0$ and cost $\kh{1-\hat s_1^*}t_\beta^2\mathbb{E}_{F_0}\fkh{y}$. Under this technology, $a_1$ is agent~$1$'s unique best response to $\hat s_1^*y$, and the principal's interim payoff guarantee from $a_1$ equals $\hat V_1^*\kh{h^1}$. The cost of $a_1$ equals the incentive gap $g_1\kh{F_0,c_0}$ of $a_0$.
\end{Lem}

\begin{proof}[Proof of Lemma~\ref{lem:opt-const}]
We first show that the pair $\kh{\hat s_1^*,t_\beta}$ is well defined. If $\kh{s,t}\in\kh{0,1}^2$ satisfies $s_0^2=s^2+2\kh{1-s}t^2$ and $s>t$, then $s=t^2+\sqrt{s_0^2-2t^2+t^4}$, because the other root of this quadratic equation in $s$ is at most $t^2<t$. Such pairs therefore form a curve that starts at $\kh{s_0,0}$ and along which $s$ is a continuously differentiable function of $t$ with $$\frac{ds}{dt}=-\frac{2\kh{1-s}t}{s-t^2}<0.$$ Along this curve, the ratio $2\kh{1-s}t/\kh{\kh{1-t}\kh{s-t}}$ is positive and continuous and converges to zero as $t\downarrow0$. It is strictly increasing in $t$, since its logarithm has derivative $$-\frac{ds/dt}{1-s}+\frac{1}{t}+\frac{1}{1-t}+\frac{1-ds/dt}{s-t},$$ in which every term is positive. Substituting $1-s=\kh{1-t}^2$ into $s_0^2=s^2+2\kh{1-s}t^2$ gives $1-s_0^2=\kh{1-t}^3\kh{1+3t}$, whose right side decreases strictly from one to zero on $\fkh{0,1}$. Since also $2t-t^2>t$, exactly one point of the curve satisfies $s=2t-t^2$, and at this point $s-t=t\kh{1-t}$, so the ratio equals $2$. Since no pair with $s\le t$ satisfies equation~\eqref{eqn:opt-const-beta}, the pair $\kh{\hat s_1^*,t_\beta}$ is well defined: if $\beta<2$, it is the unique point of the curve at which the ratio equals $\beta$, and its $t$ is smaller than at this point; if $\beta\ge2$, it is this point. Since $s$ decreases from $s_0$ along the curve, $\hat s_1^*<s_0$. Since $s$ decreases and $2t-t^2$ increases in $t$, and the two are equal at this point, $\hat s_1^*\ge2t_\beta-t_\beta^2$, with equality if and only if $\beta\ge2$, so that
\eqn{\kh{1-\hat s_1^*}^2\le\kh{1-\hat s_1^*}\kh{1-t_\beta}^2,\label{eqn:opt-const-order}}
with equality if and only if $\beta\ge2$. Finally, equation~\eqref{eqn:opt-const-resource} gives $1-s_0^2=\kh{1-\hat s_1^*}\kh{1+\hat s_1^*-2t_\beta^2}$, which exceeds $\kh{1-\hat s_1^*}^2$ because $\hat s_1^*>t_\beta>t_\beta^2$.

We next show that the minimum of the objective of program~\eqref{eqn:prog} at $s_1=\hat s_1^*$ is at least $\kh{1-\hat s_1^*}^2+\beta\kh{1-\hat s_1^*}\kh{1-t_\beta}^2$. Let $\kh{x,z}$ be feasible, so that $g,z\ge0$ and, by equation~\eqref{eqn:opt-const-resource}, $$g+z=\hat s_1^*x-\hat s_1^*+s_0^2=\hat s_1^*\kh{x-\kh{1-\hat s_1^*}}+2\kh{1-\hat s_1^*}t_\beta^2.$$ At $x=1-\hat s_1^*$ and $g=z=\kh{1-\hat s_1^*}t_\beta^2$, both $\xi\kh{1-\hat s_1^*,g}$ and $\xi\kh{x,z}$ equal $\kh{1-\hat s_1^*}\kh{1-t_\beta}^2$, and $\xi$ has gradient $\kh{1-t_\beta,1-1/t_\beta}$ there. Since $\xi$ is convex, the gradient inequalities at this point give
\eqns{\xi\kh{1-\hat s_1^*,g}+\xi\kh{x,z}&\ge2\kh{1-\hat s_1^*}\kh{1-t_\beta}^2+\kh{1-t_\beta}\kh{x-\kh{1-\hat s_1^*}}-\frac{1-t_\beta}{t_\beta}\kh{g+z-2\kh{1-\hat s_1^*}t_\beta^2}\\&=2\kh{1-\hat s_1^*}\kh{1-t_\beta}^2-\frac{\kh{1-t_\beta}\kh{\hat s_1^*-t_\beta}}{t_\beta}\kh{x-\kh{1-\hat s_1^*}}.}
If $\beta\le2$, equation~\eqref{eqn:opt-const-beta} holds in its first form, also when $\beta=2$, so $\kh{1-t_\beta}\kh{\hat s_1^*-t_\beta}/t_\beta=2\kh{1-\hat s_1^*}/\beta$. Since the maximum in the objective is at least the average of its second and fourth terms, the objective is at least $$\kh{1-\hat s_1^*}x+\beta\kh{1-\hat s_1^*}\kh{1-t_\beta}^2-\kh{1-\hat s_1^*}\kh{x-\kh{1-\hat s_1^*}}=\kh{1-\hat s_1^*}^2+\beta\kh{1-\hat s_1^*}\kh{1-t_\beta}^2.$$ If $\beta>2$, then $1-\hat s_1^*=\kh{1-t_\beta}^2$ and $\hat s_1^*-t_\beta=t_\beta\kh{1-t_\beta}$, so the right side of the gradient inequalities equals $3\kh{1-\hat s_1^*}^2-\kh{1-\hat s_1^*}x$. The maximum in the objective is at least the weighted average of its first, second and fourth terms with weights $\kh{\beta-2}/\kh{3\beta}$, $\kh{\beta+1}/\kh{3\beta}$ and $\kh{\beta+1}/\kh{3\beta}$, so the objective is at least $$\kh{1-\hat s_1^*}x+\frac{\beta-2}{3}\kh{1-\hat s_1^*}x+\frac{\beta+1}{3}\kh{3\kh{1-\hat s_1^*}^2-\kh{1-\hat s_1^*}x}=\kh{1+\beta}\kh{1-\hat s_1^*}^2,$$ which equals $\kh{1-\hat s_1^*}^2+\beta\kh{1-\hat s_1^*}\kh{1-t_\beta}^2$ because $\kh{1-t_\beta}^2=1-\hat s_1^*$.

Conversely, consider the point $x=1-\hat s_1^*$ and $z=\kh{1-\hat s_1^*}t_\beta^2$, which corresponds to the action $a_1$ in the lemma. By equation~\eqref{eqn:opt-const-resource}, $g=\kh{1-\hat s_1^*}t_\beta^2$, so the cost of $a_1$ equals the incentive gap of $a_0$. Agent~$1$'s payoff from $a_1$ exceeds his payoff from $a_0$ by $g\mathbb{E}_{F_0}\fkh{y}>0$ and his payoff from the null action by $\kh{1-\hat s_1^*}\kh{\hat s_1^*-t_\beta^2}\mathbb{E}_{F_0}\fkh{y}>0$, so the point is feasible, and $a_1$ is his unique best response under the technology $A_0\cup\hkh{a_1}$. The four terms in the maximum in the objective are $\kh{1-\hat s_1^*}^2$, $\kh{1-\hat s_1^*}\kh{1-t_\beta}^2$, $\kh{1-s_0}^2$ and $\kh{1-\hat s_1^*}\kh{1-t_\beta}^2$. Since $\hat s_1^*<s_0$, equation~\eqref{eqn:opt-const-order} shows that the maximum is $\kh{1-\hat s_1^*}\kh{1-t_\beta}^2$, attained by the second and fourth terms, and by the first when $\beta\ge2$. Hence the objective at this point equals $\kh{1-\hat s_1^*}^2+\beta\kh{1-\hat s_1^*}\kh{1-t_\beta}^2$, which, together with the lower bound, proves the statement about $a_1$ in the lemma.

It remains to show that the minimum of the objective is strictly below $\kh{1-\hat s_1^*}^2+\beta\kh{1-\hat s_1^*}\kh{1-t_\beta}^2$ for every $s_1\ne\hat s_1^*$. If $s_1=0$, the point $\kh{0,0}$ gives $\beta\kh{1-s_0}^2$. If $1-s_1<\kh{1-\hat s_1^*}^2$, then $s_1>1-\kh{1-\hat s_1^*}^2>s_0^2$, and the point $\kh{1,s_0^2}$, which has $g=0$, gives $1-s_1+\beta\max\hkh{1-s_1,\kh{1-s_0}^2}$. Since $\kh{1-s_0}^2<\kh{1-\hat s_1^*}^2\le\kh{1-\hat s_1^*}\kh{1-t_\beta}^2$, both values are strictly below the bound.

Suppose finally that $0<s_1\ne\hat s_1^*$ and $1-s_1\ge\kh{1-\hat s_1^*}^2$, and let $x=\kh{1-\hat s_1^*}^2/\kh{1-s_1}\in(0,1]$, so that $\kh{1-s_1}x=\kh{1-\hat s_1^*}^2$. Equation~\eqref{eqn:opt-const-resource} then gives $s_1x-s_1+s_0^2-\kh{1-s_1}-x=-2\kh{1-\hat s_1^*}\kh{1-t_\beta^2}$, and expanding the squares yields
\eqns{&s_1x-s_1+s_0^2-\kh{\sqrt{1-s_1}-\kh{1-t_\beta}\sqrt{1-\hat s_1^*}}^2-\kh{\sqrt x-\kh{1-t_\beta}\sqrt{1-\hat s_1^*}}^2\\&\quad=2\kh{1-t_\beta}\sqrt{1-\hat s_1^*}\kh{\sqrt{1-s_1}+\sqrt x-2\sqrt{1-\hat s_1^*}}>0,}
where the inequality follows from the arithmetic-geometric mean inequality, since $\sqrt{1-s_1}\sqrt x=1-\hat s_1^*$, and is strict because $s_1\ne\hat s_1^*$. Since each squared positive part is at most the corresponding square, $$s_1x-s_1+s_0^2>\kh{\sqrt{1-s_1}-\kh{1-t_\beta}\sqrt{1-\hat s_1^*}}_+^2+\kh{\sqrt x-\kh{1-t_\beta}\sqrt{1-\hat s_1^*}}_+^2.$$ Moreover, $s_1x>\kh{\sqrt x-\kh{1-t_\beta}\sqrt{1-\hat s_1^*}}_+^2$. This is immediate if $\sqrt x\le\kh{1-t_\beta}\sqrt{1-\hat s_1^*}$, and otherwise equation~\eqref{eqn:opt-const-order} gives $$\sqrt x-\kh{1-t_\beta}\sqrt{1-\hat s_1^*}\le\sqrt x-\sqrt{1-s_1}\sqrt x<\sqrt{s_1x},$$ because $1-\sqrt{1-s_1}<\sqrt{s_1}$ for $s_1\in\kh{0,1}$. If $s_1<s_0^2$, then $s_1x-s_1+s_0^2=\kh{s_0^2-s_1}+s_1x$, so $s_1x-s_1+s_0^2>\kh{s_0^2-s_1}+\kh{\sqrt x-\kh{1-t_\beta}\sqrt{1-\hat s_1^*}}_+^2$. If $s_1\ge s_0^2$, then $\kh{s_0^2-s_1}_+=0$. In either case,
\eqn{s_1x-s_1+s_0^2>\max\hkh{\kh{s_0^2-s_1}_+,\kh{\sqrt{1-s_1}-\kh{1-t_\beta}\sqrt{1-\hat s_1^*}}_+^2}+\kh{\sqrt x-\kh{1-t_\beta}\sqrt{1-\hat s_1^*}}_+^2.\label{eqn:opt-const-slack}}
Since both sides are continuous in $x$, the inequality also holds at a slightly smaller positive $x$, at which $\kh{1-s_1}x<\kh{1-\hat s_1^*}^2$. Fix such an $x$, choose $g$ strictly between the maximum in equation~\eqref{eqn:opt-const-slack} and $s_1x-s_1+s_0^2-\kh{\sqrt x-\kh{1-t_\beta}\sqrt{1-\hat s_1^*}}_+^2$, and let $z=s_1x-s_1+s_0^2-g$. Then $g>\kh{s_0^2-s_1}_+$ and $z>\kh{\sqrt x-\kh{1-t_\beta}\sqrt{1-\hat s_1^*}}_+^2\ge0$, so the point $\kh{x,z}$ is feasible. Moreover, $\sqrt g>\sqrt{1-s_1}-\kh{1-t_\beta}\sqrt{1-\hat s_1^*}$ and $\sqrt z>\sqrt x-\kh{1-t_\beta}\sqrt{1-\hat s_1^*}$, so $\xi\kh{1-s_1,g}$ and $\xi\kh{x,z}$ are strictly below $\kh{1-\hat s_1^*}\kh{1-t_\beta}^2$, and so are $\kh{1-s_1}x$ and $\kh{1-s_0}^2$. The objective at $\kh{x,z}$ is therefore strictly below $\kh{1-\hat s_1^*}^2+\beta\kh{1-\hat s_1^*}\kh{1-t_\beta}^2$.

These cases cover every $s_1\ne\hat s_1^*$. Hence $\hat s_1^*$ is the unique maximizer of $\hat U$, and Theorem~\ref{prop:1} gives the stated value of $\hat V_1^*\kh{h^1}$. This completes the proof.
\end{proof}

\begin{proof}[Proof of Proposition~\ref{prop:expconst}]
By Lemma~\ref{lem:opt-const}, the optimal first-period share is unique and equals $\hat s_1^*$, and the proof of the lemma shows that $\hat s_1^*<s_0$. Along the curve in that proof, the ratio $2\kh{1-s}t/\kh{\kh{1-t}\kh{s-t}}$ is continuous and strictly increasing in $t$, converges to zero as $t\downarrow0$, and equals $2$ at the point where $s=2t-t^2$. Hence $t_\beta$ is continuous and strictly increasing in $\beta$ on $(0,2]$, with $t_\beta\to0$ as $\beta\to0$, and constant on $[2,\infty)$. Since $s$ is strictly decreasing in $t$ along the curve and converges to $s_0$ as $t\downarrow0$, $\hat s_1^*$ is strictly decreasing in $\beta$ on $(0,2]$, constant on $[2,\infty)$, and converges to $s_0$ as $\beta\to0$.
\end{proof}

\end{appendices}

\end{document}